\documentclass[a4paper, 11pt]{article}
\usepackage[english]{babel}
\usepackage[utf8]{inputenc}
\usepackage[T1]{fontenc}
\usepackage{amsmath, amsthm, amssymb}
\usepackage{mathtools}
\usepackage{mathrsfs}
\usepackage{enumitem}
\usepackage[hidelinks]{hyperref}
\usepackage{geometry}
\usepackage{tikz}
\usetikzlibrary{positioning}
\usepackage[dvipsnames]{xcolor}
\usepackage{lmodern}
\usepackage{csquotes}

\colorlet{darkblue}{blue!50!black}
\hypersetup{
	colorlinks,%
	citecolor=darkblue,%
	filecolor=red,%
	linkcolor=darkblue,%
	urlcolor=darkblue,%
	pdfnewwindow=true,%
	pdfstartview={FitH},%
	pdftitle={Locality of KMS states for infinite fermionic lattice systems},%
	pdfauthor={Lennart Becker},%
	pdfkeywords={KMS states, fermionic lattice systems, quantum belief propagation, decay of correlations, local indistinguishability, LPPL}
}

 \usepackage[doi=true,isbn=false, url=true, date=year, backend=biber,maxbibnames=99, style=alphabetic
  ]{biblatex} 
  \renewbibmacro*{journal+issuetitle}{%
     \usebibmacro{journal}%
     \setunit*{\addspace}%
     \iffieldundef{series}
      {}
       {\newunit
        \printfield{series}%
        \setunit{\addspace}}%
     \usebibmacro{volume+number+eid}%
     \setunit{\bibpagespunct}%
    \printfield{pages}%
     \setunit{\addspace}%
   \usebibmacro{issue+date}%
     \setunit{\addcolon\space}%
   \usebibmacro{issue}%
     \newunit}
   \renewbibmacro*{note+pages}{%
    \printfield{note}%
   \newunit}

\theoremstyle{plain}
\newtheorem{theorem}{Theorem}[section]
\newtheorem{lemma}[theorem]{Lemma}
\newtheorem{proposition}[theorem]{Proposition}
\newtheorem{corollary}[theorem]{Corollary}

\theoremstyle{definition}
\newtheorem{definition}[theorem]{Definition}

\theoremstyle{remark}
\newtheorem{remark}[theorem]{Remark}

\DeclareMathOperator{\aut}{Aut}

\newcommand{\E}{\mathbb{E}}
\newcommand{\tr}{\mathrm{tr}}
\newcommand{\e}{\mathrm{e}}
\newcommand{\diam}{\mathrm{diam}}
\renewcommand{\d}{\mathrm{d}}
\renewcommand{\i}{\mathrm{i}}
\renewcommand{\Phi}{\varPhi}
\renewcommand{\Psi}{\varPsi}

\newcommand{\mT}{\mathcal{T}}

\newcommand{\N}{\mathbb{N}}
\newcommand{\Z}{\mathbb{Z}}
\newcommand{\C}{\mathbb{C}}
\newcommand{\mA}{\mathcal{A}}
\newcommand{\mP}{\mathcal{P}}
\newcommand{\mH}{\mathcal{H}}
\newcommand{\R}{\mathbb{R}}
\newcommand{\eps}{\varepsilon}
\newcommand{\mL}{\mathcal{L}}
\newcommand{\cor}{\mathrm{cor}}
\newcommand{\lppl}{\mathrm{LPPL}}
\newcommand{\li}{\mathrm{LI}}
\newcommand{\lr}{\mathrm{LR}}
\newcommand{\ent}{\mathrm{ent}}
\newcommand{\sur}{\mathrm{sur}}
\newcommand{\vol}{\mathrm{vol}}
\newcommand{\car}{\mathrm{CAR}}
\newcommand{\inter}{\mathrm{int}}
\renewcommand{\phi}{\varphi}

\DeclarePairedDelimiter\norm{\lVert}{\rVert}

\newcommand\numberthis{\addtocounter{equation}{1}\tag{\theequation}}

\usepackage{parskip}

\title{Locality of KMS states for infinite fermionic lattice systems}
\author{Lennart Becker%
\texorpdfstring{\footnote{\parbox[t]{.7\textwidth}{
                \foreignlanguage{ngerman}{Fachbereich Mathematik,  Universität Tübingen,\\
                Auf~der~Morgenstelle~10, 72076~Tübingen,} Germany
            }
        }
    }{}%
}
\date{\today}

\begin{document}

\maketitle
\begin{abstract}
We study how thermal equilibrium states of infinite fermionic lattices respond to perturbations of the Hamiltonian. Our main result concerns extensive perturbations of a chosen KMS state. We show that exponential decay of correlations and exponential local indistinguishability under the extensive perturbation yield stretched-exponential decay of correlations for the perturbed state. Neither the initial nor the perturbed state is required to be unique, allowing stability to be studied in the presence of phase coexistence. Our approach uses Quantum Belief Propagation (QBP), which we establish directly on the underlying $C^*$-algebra. Combining QBP with Lieb--Robinson bounds, we show implications between decay of correlations, the principle that local perturbations perturb locally, and local indistinguishability, closing the cycle for perturbations of the tracial state under the stated decay, support and uniformity assumptions. For extensive perturbations under uniform clustering conditions, we also obtain a differentiable path of local equilibrium expectations and bound their change linearly in the interaction norm of the perturbation. 
\end{abstract}

\section{Introduction}

We study the stability of locality properties of thermal (KMS) states of infinitely extended fermionic lattice systems under perturbations of the Hamiltonian. As a central result we show that a quantitative clustering property persists under an extensive perturbation if the effect of enlarging finite perturbation patches is sufficiently local. This criterion applies to an already interacting reference state and does not require uniqueness of either the initial or the perturbed KMS state. Our approach also gives quantitative bounds on the response to quasi-local perturbations and on the change of local expectations under extensive perturbations.

Let us describe the setting and the stability result in more detail. We fix an inverse temperature $\beta$ and start with a $(\tau,\beta)$-KMS state $\omega$, where $\tau$ is generated by a Hamiltonian $H$ on a lattice $\Gamma$. The locality property of interest is decay of correlations: for local observables $A\in\mA_X$ and $B\in\mA_Y$, with at least one of them even, the covariance
\begin{align*}
    |\omega(AB)-\omega(A)\omega(B)|
\end{align*}
should decay with the distance between $X$ and $Y$, with explicit control of its dependence on the two supports. We consider exponential and polynomial decay. This quantitative form of clustering is an additional assumption on the state. Under translation-invariance and asymptotic-abelianness assumptions, clustering criteria are related to extremality of equilibrium states \cite[Corollary IV.4.17]{S93}, but this relation alone does not specify a decay rate or its support dependence. For background on locality at finite temperature, see \cite{A23,KGRE14,BK18,A69,PP23}.

An extensive perturbation $V$ cannot, in general, be represented by a bounded element of the quasi-local algebra. Such perturbations can change the bulk behavior of the system and may induce a phase transition. We therefore first perturb $\omega$ by the bounded operators $V_M$ obtained by retaining the terms of $V$ contained in a finite region $M$, and denote the resulting Araki-perturbed states by $\omega_{V_M}$ \cite{A73a}. To control the limit as the perturbation fills the lattice, we introduce local indistinguishability (LI) relative to $\omega$ and $V$. It requires the difference
\begin{align*}
    |\omega_{V_Z}(A)-\omega_{V_Y}(A)|,
    \qquad A\in\mA_X,\quad X\subsetneq Y\subsetneq Z,
\end{align*}
to decay with the distance from $X$ to $Z\setminus Y$, uniformly over the finite regions. Thus LI controls the effect of adding perturbation terms far from the observable. It implies that the states $\omega_{V_M}$ converge, independently of the exhaustion, to a KMS state $\omega_V$ for $H+V$ (Proposition \ref{prop: convergence from li}). When $H=0$, the reference state is tracial and the patch states restrict to finite-volume Gibbs states, recovering the usual finite-volume formulation of LI \cite{KGRE14,BK18,CMTW25,BCP22}.

Our main stability theorem shows that LI controls not only convergence to the perturbed state, but also the persistence of clustering. Suppose that $H$ and $V$ are exponentially local, that $\omega$ satisfies exponential decay of correlations, and that it satisfies exponential LI for $V$, with the support bounds stated in Theorem \ref{thm: doc from LI + doc}. Then $\omega_V$ has stretched-exponentially decaying correlations. The reference state may already be a clustering KMS state of a nonzero Hamiltonian. This extends the scope of the implication from LI to clustering established for finite-volume Gibbs states by Capel, Moscolari, Teufel and Wessel \cite[Theorem 31]{CMTW25}. The extension concerns the admissible reference state and stability in the presence of phase coexistence; the decay rates and hypotheses of the two results differ.

The possibility of starting from a specified KMS state is particularly important when several equilibrium phases coexist. Finite-volume Gibbs states and infinite-volume KMS states are complementary descriptions of thermal equilibrium, but clustering of a chosen KMS state need not imply clustering uniformly for a prescribed family of finite-volume Gibbs states. Moreover, the restriction of a KMS state to a finite region is generally not the Gibbs state of the truncated Hamiltonian, since it also reflects the coupling to the exterior. The zero-field Ising model on $\mathbb Z^2$ illustrates this distinction: at $\beta>\beta_c$, the plus and minus phases cluster, whereas Gibbs states on square boxes with free boundary conditions converge to their non-clustering symmetric mixture by spin-flip symmetry and \cite[Propositions 6.9 and 8.5]{GHM01}. In general, this is not an obstruction to the finite-volume approach: appropriate boundary conditions can select clustering phases. The KMS formulation makes the chosen equilibrium state and the locality assumptions imposed on it explicit without first requiring a finite-volume approximation. In one dimension (\cite{A75}, \cite{K76}) and at a sufficiently high temperature in any dimension (\cite{FU15}), the KMS state for sufficiently short-range Hamiltonians is unique. In these situations, the Gibbs states converge to the KMS state.

The perturbation estimates underlying our results are obtained from Quantum Belief Propagation (QBP). Introduced for finite-dimensional Gibbs states by Hastings \cite{H07a}, put in the form used here by Kim \cite{K12}, and extended to KMS states of $W^*$-dynamical systems by Ejima and Ogata \cite{EO19}, QBP implements the change of state under a bounded perturbation by a bounded operator. We give an algebra-valued formulation: Theorem \ref{thm: QBP} constructs this operator in the original $C^*$-algebra by a norm-convergent Dyson series. It implements the same state change as Araki's perturbation theory. Quantitative locality estimates follow by combining this representation with suitable Lieb--Robinson bounds.

For a localized perturbation, these estimates yield the principle that local perturbations perturb locally (LPPL). Originally formulated for gapped ground states \cite{BMNS11,DS15}, this principle asks that, for an even self-adjoint perturbation $P\in\mA_X$ and an observable $A\in\mA_Y$, the change
\begin{align*}
    |\omega_P(A)-\omega(A)|
\end{align*}
decays with the distance between $X$ and $Y$. We prove two versions for KMS states, allowing also quasi-local perturbations with suitably decaying tails. Theorem \ref{thm: LPPL} assumes clustering uniformly along the path $\omega_{sP}$, $s\in[0,1]$, whereas Theorem \ref{thm: LPPL without doc} assumes clustering only in the initial state, at the price of exponential dependence on the decay norm of the perturbation.

These bounded-perturbation estimates also provide a route to LI. If LPPL holds with common bounds for the states $\omega_{V_M}$ over all finite patches $M$, then $\omega$ satisfies LI for $V$ (Theorem \ref{thm: li from lppl}). Moreover, starting from the tracial state, LI gives a common covariance bound for all finite-patch states and the limiting KMS state (Theorem \ref{thm: doc from LI}). Together, these results reproduce the cycle of implications for perturbations of the tracial state under the respective decay, support and uniformity assumptions, and extend individual implications in the equivalence established in \cite{CMTW25} to infinitely extended fermionic systems and quasi-local perturbations. Each implication has its own decay, support and uniformity requirements, and generally entails losses in the quantitative bounds.

Uniform clustering over both finite patches and coupling strengths gives further control of the selected states. If the states $\omega_{sV_M}$ satisfy the common clustering bounds of Theorem \ref{thm: QBP SLT} for all finite $M$ and $s\in[0,1]$, we construct KMS states $\omega_{sV}$ for $H+sV$ whose local expectations are continuously differentiable in $s$. Their derivatives are given by an absolutely convergent QBP formula, and the change from $\omega$ to $\omega_V$ is bounded linearly in the interaction norm of $V$. Related continuity and differentiability results appear in \cite{CMTW25,MSAC26} and \cite[Appendix A]{RGK24}.

All statements are formulated for lattice fermions and the proofs also apply to spin systems. The extension in the opposite direction is not automatic: observables on disjoint fermionic regions need not commute. We use even perturbations, parity-preserving conditional expectations and fermionic Lieb--Robinson bounds to account for this distinction; see the discussions in \cite{AM03,TW25,BTW26}. We treat both short- and long-range interactions. The long-range estimates rely on the fermionic Lieb--Robinson bounds of \cite{TW25}, which are based on the spin-system bounds of \cite{EMNY20}. Locality of temperature and decay of correlations at high temperature for finite fermionic systems were studied in \cite{KGRE14}.

The results hold at arbitrary inverse temperature $\beta$ under their stated hypotheses. We keep track of the temperature dependence in the QBP and locality estimates. We stress that, apart from the illustrative Ising examples, our purpose is to establish conditional implications rather than verify their hypotheses in particular models. Proving any of the notions of locality for physically interesting models and regimes is a formidable task. We discuss available high-temperature results in Section \ref{sec: locality of KMS}.

The paper is organized as follows. Section \ref{sec: Setup} introduces the mathematical setup. Section \ref{sec: general qbp} states the QBP results for $C^*$-dynamical systems, and Section \ref{sec: locality of KMS} formulates and discusses the locality and stability results for fermionic KMS states. Their proofs are given in Sections \ref{sec: proof general qbp} and \ref{sec: proof fermions}, respectively. Technical lemmas and locality estimates are collected in the appendices.

\section{Mathematical setup} \label{sec: Setup}

In this section, we introduce the notions we use to describe $C^*$-dynamical systems and interacting fermions on a lattice. We prove the QBP statements for general $C^*$-dynamical systems. For the locality results, we take a $C^*$-dynamical system describing interacting lattice fermions.

\subsection{\texorpdfstring{$C^*$}{C-star}-dynamical systems}

We consider $C^*$-dynamical systems $(\mA, \tau)$, where $\mA$ is a unital $C^*$-algebra and $\tau : \R \to \aut (\mA)$ is a one-parameter group of automorphisms\footnote{In the following the term {\it automorphism} is used in the sense of a $*$-automorphism as defined for example in~\cite{BR1}.}. We assume that for any $A \in \mA$ the map $t \mapsto \tau_t (A)$ is norm-continuous. We introduce the derivation $\mL$, called Liouvillian, by
\begin{align*}
    \i \mL (A) = \lim_{t \to 0} (\tau_t (A) - A) / t \, ,
\end{align*}
for all $A \in \mA$ for which the limit exists. The dense domain $D(\mL) \subset \mA$ consists of all such elements. 

For a self-adjoint element $P \in \mA$, we define the bounded derivation $\mL_P : \mA \to \mA$ by
\begin{align*}
    \mL_P (A) \coloneq [P, A] \, .
\end{align*}
The following lemma asserts that, given a dynamics $\tau$ generated by $\mL$, one can construct a new $C^*$-dynamical system $(\mA, \tau^P)$ whose generator is the (densely defined) derivation $\mL + \mL_P$.

\medskip
\begin{lemma}[{\cite[Proposition 5.4.1]{BR2}}] \label{lem: perturbed dynamics}
    Let $P = P^* \in \mA$. There exists a unique group of automorphisms $\tau^P$ generated by the derivation $\mL + \mL_P$. For every $A \in \mA$, $\tau^P$ can be obtained by the absolutely converging series
    \begin{align*}
        \tau_t^P  (A) = \sum_{n = 0}^\infty \i^n \int_0^t \d t_1 \int_0^{t_1} \d t_2 \dots \int_0^{t_{n-1}} \d t_n \, [\tau_{t_n} (P), \,  [\dots [\tau_{t_1} (P), \, \tau_t (A)]]] \, .
    \end{align*}
Duhamel's formula gives, for all $A\in\mA$,
    \begin{align*}
        \lVert\tau_t^P(A)-\tau_t(A)\rVert
        &\leq 2|t|\,\lVert P\rVert\,\lVert A\rVert\, .
    \end{align*}
\end{lemma}

\subsection{KMS states and Araki perturbation theory} \label{sec: KMS states}

At inverse temperature $\beta \in \R$, equilibrium states of a $C^*$-dynamical system are characterized by the Kubo–Martin–Schwinger (KMS) condition. We recall the definition.

\medskip
\begin{definition}
    Let $(\mA, \tau)$ be a $C^*$-dynamical system and $\beta \in \R$. We denote by $\mA_\ent^\tau$ the set of $\tau$-entire analytic elements of $\mA$. We call a state a $(\tau, \beta)$-KMS state if it satisfies 
    \begin{align*}
        \omega(A \, \tau_{\i  \beta} (B)) = \omega( B A)
    \end{align*}
for all $A \in \mA$ and $B \in \mA_\ent^\tau$. For $\beta = 0$ we additionally require $\omega\circ\tau_t = \omega$ for all $t\in\R$\footnote{For $\beta \neq 0$ this property follows automatically \cite[Proposition 5.3.3]{BR2}.}.
\end{definition}

\medskip
\begin{remark}
    The set $\mA_\ent^\tau$ is dense in $\mA$. Moreover, if $P = P^* \in \mA_\ent^\tau$ then for any $A \in \mA_\ent^\tau$ it also holds that $A \in \mA_\ent^{\tau^P}$. 
\end{remark}

We are interested in the KMS states of perturbed dynamics. Hence, we recall a construction of a $(\tau^P, \, \beta)$-KMS state starting from a $(\tau, \beta)$-KMS state $\omega$, where $P = P^* \in \mA$. This method was introduced by Araki in \cite{A73a}. Denote by $(\mathcal{H}, \, \pi, \, \Omega)$ the GNS triple associated to $\omega$ and $H$ the Hamiltonian that implements the dynamics in this representation. We define 
\begin{align*} \label{eq: perturbed vector}
    \Omega_P \coloneq \e^{- \beta \,(H + \pi(P)) / 2} \, \Omega \in \mH \, . \numberthis
\end{align*}
Then the state
\begin{align*} \label{eq: perturbed state}
    \omega_P (A) \coloneq \frac{(\Omega_P, \, \pi(A) \, \Omega_P)}{(\Omega_P, \, \Omega_P)} \, . \numberthis
\end{align*}
defines a $(\tau^P, \, \beta)$-KMS state. For the proof, see, for example, \cite[Theorem 5.4.4]{BR2}. Notice that
\begin{align*}
    (\omega_{P_1})_{P_2} (A) = \omega_{P_1 + P_2} (A) \, .
\end{align*}

\subsection{Fermionic lattice systems and spaces of localized operators}

In the second part of this paper we consider $C^*$-dynamical systems that describe fermionic lattice systems. As lattices we consider countably infinite connected graphs $(\Gamma, \, E)$ with graph distance $d(\cdot, \cdot)$. We assume that this graph is surface-regular with dimension $D \in \N$. That is, there is a constant $C_\sur \geq1$ such that for all $x \in \Gamma$ and $r \geq0$
\begin{align*}
    |S_r (x)| \coloneq |\{y \in \Gamma ~|~ d(x,y) = r\}| \leq C_\sur \, (1 + r)^{D - 1} \, .
\end{align*}
From an integration argument it follows that all surface-regular graphs are also $D$-regular in the sense that there is a constant $C_\vol$ such that for all $x \in \Gamma$ and $r > 0$
\begin{align*}
    |B_r (x)| \coloneq | \{y \in \Gamma ~|~ d(x, \, y) \leq r\} | \leq C_\vol \, (1 + r)^D \, .
\end{align*}
A particular example for a surface-regular graph is $\Z^D$ equipped with the $\ell^1$-distance. The set of finite, nonempty subsets of $\Gamma$ will be called $P_0 (\Gamma)$. For $X\in P_0(\Gamma)$ and $r\geq0$, we define the fattening of $X$ by
\begin{align*}
    X_r \coloneq \{x\in\Gamma ~|~ d(x,X)\leq r\}\, .
\end{align*}
Often, we want to approximate quantities on $\Gamma$ by quantities on $\Lambda_n \in P_0 (\Gamma)$, where $\Lambda_n \subseteq \Lambda_{n + 1}$ and $\bigcup_{n \in \N} \, \Lambda_n = \Gamma$. We call such a sequence $(\Lambda_n)_{n \in \N}$ an \textit{increasing and absorbing sequence} (IAS). 

The fermionic Fock space over the lattice $\Gamma$, where the fermions at each lattice site $x \in \Gamma$ can have $n \in \N$ degrees of freedom, is given by  
\begin{align*}
   \mathcal{F}(\Gamma, \C^n) \coloneq \bigoplus_{N = 0}^{\infty} \ell^2(\Gamma ,\C^n)^{\wedge N} \, ,
\end{align*}
where $\ell^2(\Gamma ,\C^n)^{\wedge N}$ denotes the $N$-fold antisymmetric tensor product of $\ell^2(\Gamma ,\C^n)$. For any $M \subseteq \Gamma$ we define the CAR algebra $\mA_M \coloneq \car (\ell^2 (M, \C^n))$ as the $C^*$-subalgebra of the algebra of bounded operators on $\mathcal{F} (M, \C^n)$ that is generated by the fermionic creation and annihilation operators $a^*_{x,i}$ and $a_{x,i}$ for $x \in M$, $i \in \{1, \dots ,n\}$. These creation and annihilation operators satisfy the canonical anti-commutation relations (CAR):
\begin{align*}
    \{a_{x, i}, \, a_{y, j}\} = \{a_{x, i}^*, \, a_{y, j}^*\} = 0 ~ \text{and} ~ \{a_{x, i}, \, a_{y, j}^*\} = \delta_{x, y} \, \delta_{i, j} \, , 
\end{align*}
for all $x$,  $y \in M$ and  $i$, $j \in \{1, \ldots, n\}$. Here, $\{A, B\} = AB + BA$ denotes the anti-commutator of $A$ and $B$. We will refer to the $C^*$-algebra $\mA_{\Gamma}$ as the quasi-local algebra and
\begin{align*}
    \mA_0 \coloneq \bigcup_{M \in P_0(\Gamma)} \mA_M \subseteq \mA_\Gamma
\end{align*}
as the local algebra. Consequently, an operator is called quasi-local if it lies in $\mA_\Gamma$ and local if it lies in $\mA_0$. For a local operator $A$, we call the smallest finite set $M\subseteq\Gamma$ such that $A\in\mA_M$ the support of $A$.

For any $X \in P_0(\Gamma)$ we define the parity automorphism $\Theta_X$ by
\begin{align*}
    \Theta_X(A) =  (-1)^{N_X} A (-1)^{N_X}, \quad \text{for all} ~~ A \in \mA_X,
\end{align*}
where $N_X \coloneq \sum_{x \in X} \sum_{i = 1}^n a_{x, i}^* a_{x, i}$ denotes the local number operator. By the quasi-local structure of $\mA_\Gamma$ there is a unique automorphism $\Theta$ on $\mA_\Gamma$ such that $\Theta |_{\mA_X} = \Theta_X$ for all $X \in P_0(\Gamma)$. 

    We define the set of even quasi-local operators
\begin{align*}
    \mA_\Gamma^+ \coloneq \{A\in \mA_\Gamma~|~   \Theta (A) = A \} \, .
\end{align*}
We denote the algebra of even operators supported in $M \subseteq \Gamma$ by $\mA^+_M \coloneq  \mA_\Gamma^+\cap \mA_M$. For disjoint regions $M_1, M_2 \subseteq \Gamma$ operators  $A \in \mA_{M_1}^+$ and $B\in \mA_{M_2}$ commute, $[A,B] = 0$.

\medskip
\begin{definition}\label{def:dec funct}
    We call a function $F : [0, \infty) \to (0, \infty)$ a decay function if $F$ is non-increasing, $F(0) = 1$ and $\log (F)$ is convex.
\end{definition}

Using these decay functions, we make the following definition for spaces of localized operators.

\medskip

\begin{definition}\label{def:norm}
    Let $F$ be a decay function, and $\E$ the fermionic conditional expectation introduced in Proposition \ref{Ex+UniqueExpectation}. We say that an observable $A \in \mA_\Gamma$ is $F$-localized if for all $x \in \Gamma$ 
    \begin{align*}
        \lVert A \rVert_{F, x} \coloneq \lVert A \rVert + \sup_{k \in \N_0} \frac{\lVert A - \E_{B_k (x)} A \rVert}{F(k)} < \infty \, .
    \end{align*}
    We denote the set of all $F$-localized observables by $\mA_F$. We set $\mA_F^+ \coloneq \mA_F \cap \mA_\Gamma^+$. For $\nu > 0$ and $F( \cdot) = (1 + \cdot)^{- \nu}$ we abbreviate $\lVert \cdot \rVert_{\nu, x} \coloneq \lVert \cdot \rVert_{F, x}$ and $\mA_\nu\coloneq\mA_F$. 
\end{definition}

\subsection{Interactions, derivations and Lieb--Robinson bounds}

In order to define dynamics on $\mA_\Gamma$ we introduce the notion of an interaction. An interaction is a map $\Phi: P_0(\Gamma) \to \mA_\Gamma^+$ such that $\Phi(M) \in \mA_M$ and $\Phi(M)^* = \Phi(M)$ for every $M\in P_0(\Gamma)$, and such that, for every $M\in P_0(\Gamma)$, the sum
\begin{equation*}
    \sum_{\substack{K\in P_0(\Gamma)\\ M\cap K \neq \emptyset}} \Phi(K)
\end{equation*}
converges unconditionally. For the empty set we set $\Phi(\emptyset)=0$. For $M \in P_0 (\Gamma)$, we set $\Phi_M \coloneq \sum_{Z \subseteq M} \Phi (Z)$. Interactions define derivations on $\mA_\Gamma$ in the following way:    
For an interaction $\Phi$, let
\begin{align*}
    \mL_{\Phi}^\circ: \mA_0 \to \mA_\Gamma,~ A\mapsto \sum_{M\in P_0(\Gamma)} [\Phi(M),A] \, .
\end{align*}
It follows from \cite[Proposition 3.2.22]{BR1} and \cite[Proposition 3.1.15]{BR1} that $\mL_{\Phi}^\circ$ is closable. We denote its closure by $\mL_{\Phi}$. 

We define the following spaces of decaying interactions.

\medskip

\begin{definition}\label{norm interaction}
    For an interaction $\Phi$ and some decay function $F$, we define
    \begin{align*}
        \|\Phi\|_{F} \coloneq \sup_{x\in \Gamma} \sum_{\substack{Z\in
        P_0(\Gamma)\\x\in Z}} \frac{\| \Phi(Z) \|}{F(\diam(Z))}  ~,
    \end{align*}
    where $\diam(Z) = \max \{ d(x, y) ~|~ x, \, y \in Z \}$. The set of interactions with finite $\norm{\cdot}_{F}$ is denoted by $\mP_{F}$.
    For $\nu \geq 0$ and $F( \cdot ) = (1 + \cdot)^{-\nu}$ we abbreviate $\norm{\cdot}_\nu \coloneq \norm{\cdot}_F$ and $\mP_\nu \coloneq \mP_F$.
\end{definition}
One can sum all the local terms of an interaction that are associated to one lattice point to obtain a quasi-local observable. In computations, it is convenient to work with this family of quasi-local terms, also called a $0$-chain.

\medskip
\begin{definition}\label{def: quasi-local terms}
    Fix once and for all a strict total order $<$ on $\Gamma$. For each $M \in P_0 (\Gamma)$, we define the center of $M$ as $\mathrm{C}(M)\in M$, where $\mathrm{C}(M)$ minimizes the map $M\to\R_+$, $x\mapsto\sum_{y\in M}d(x,y)$. If there are several such points, choose the smallest one with respect to $<$. Let $x\in \Gamma$. We define $R_x \subseteq P_0(\Gamma)$ to be the set of all finite subsets of $\Gamma$ that have their center in $x$. Given an interaction $\Phi$ we define 
    \begin{align*}
        \Phi_x \coloneq \sum_{M \in R_x} \Phi(M) \, .
    \end{align*}
\end{definition}

\medskip

\begin{remark} \label{rem: bound zero chain}
    Note that for any decay function $F : \R_+ \to \R_+$, interaction $\Phi \in \mP_F$ and lattice point $x$ the quasi-local observable $\Phi_x$ is well-defined and that $\lVert \Phi_x \rVert_{F,x} \leq 3 \, \lVert \Phi \rVert_F $.
\end{remark}

\medskip
\begin{remark}
    If $\Gamma$ is embedded in $\R^D$, then the strict total order can be defined by the lexicographical order. For an explicit construction, see \cite{BTW26}.
\end{remark}

In Section \ref{sec: locality of KMS}, we consider two classes of interactions, short-range interactions that decay exponentially and long-range interactions that decay polynomially. These classes of interactions generate time evolutions, which we call locally generated. The resulting dynamics are local, as quantified by so-called Lieb--Robinson bounds.

\medskip
\begin{definition}
    We say that the dynamics $\tau$ satisfies a Lieb--Robinson bound with $\zeta_\lr ~:~ P_0 (\Gamma) \times P_0(\Gamma) \times [0, \infty) \to [0, \infty)$ if
    \begin{align*}
        \lVert [\tau_t (A), \, B] \rVert \leq \lVert A \rVert \, \lVert B \rVert \, \zeta_\lr (X, Y, |t|) \, ,
    \end{align*}
    for all $X$, $Y \in P_0(\Gamma)$, $A \in \mA_X^+$, $B \in  \mA_Y$ and $t \in \R$. 
\end{definition}

\section{Quantum Belief Propagation} \label{sec: general qbp}

In this section, we construct the QBP implementation of the map introduced in Section \ref{sec: KMS states} for $C^*$-dynamical systems. We build on the results of Ejima and Ogata for $W^*$-dynamical systems \cite{EO19}. Their proof uses modular theory and therefore does not immediately yield an operator in the original $C^*$-algebra. The $C^*$-dynamical statement can be obtained by passing to a GNS representation, completing the represented algebra to a $W^*$-algebra, and then verifying that the relevant operators in the completion are representations of operators in the original $C^*$-algebra. We instead give a more elementary proof directly in the $C^*$-dynamical setting. The proofs of the statements in this section are given in Section \ref{sec: proof general qbp}.

Let $(\mA, \tau)$ be some $C^*$-dynamical system and $A \in \mA$ be arbitrary. For $\beta \neq 0$ we define 
\begin{align} \label{def: Phi}
    \Phi_\beta^{\tau} (A) \coloneq \int_\R \d t \, f_\beta(t) \, \tau_{-t} (A) \, ,
\end{align}
where the function $f_\beta : \R \to \R$ is defined via its Fourier transform
\begin{align} \label{def: f_hat}
    \hat{f}_\beta(\omega) \coloneq 2 (1 + \e^{\beta \omega})^{-1} \, \int_0^1 \d \tau \, \e^{\beta \tau \omega} = \frac{\tanh{\tfrac{\beta \omega}{2}}}{\tfrac{\beta \omega}{2}} \, .
\end{align}
Following \cite{CMTW25} and \cite[Lemma 10]{EO19}, we note the properties of $f_\beta$ that are relevant for our discussion:
\begin{itemize}
    \item[(i)]  $f_\beta$ is even and non-negative.
    \item[(ii)] $f_{\beta}$ is $L^1 (\R)$-normalized.
    \item[(iii)] $f_\beta (t)$ falls off exponentially for $|t| \to \infty$.
\end{itemize}
For $\beta = 0$ we set $\Phi_\beta^\tau (A) = A$.

It follows that $\bigl( \Phi_\beta^{\tau} (A) \bigr)^* = \Phi_\beta^\tau (A^*)$ for any $A \in \mA$. The following theorem gives the bounded operator that implements the perturbed state.

\medskip
\begin{theorem} \label{thm: QBP}
    Let $(\mA, \tau)$ be a $C^*$-dynamical system, $\omega$ a $(\tau, \beta)$-KMS state, and $P = P^* \in \mA$. Then there exists a norm convergent Dyson series $\eta_P \in \mA$, explicitly given by Equation \eqref{eq: eta}, such that
    \begin{align*}
        \omega_P(A) = \omega(\eta_P \, A \, \eta_P^*) \qquad \forall A \in \mA \, .
    \end{align*}
    In particular, $ A \mapsto \eta_P A \eta_P^*$ is a completely positive bounded map transporting $\omega$ to $\omega_P$.
\end{theorem}

\medskip
\begin{remark}
    Formally, there is a similar concept for ground states of gapped Hamiltonians, see, for example, \cite{H07b, BMNS11, MO19, BTW26}. In this setting, one constructs an automorphism that maps the ground state of a gapped Hamiltonian to the ground state of another gapped Hamiltonian, assuming that the spectral gap does not close along a path of Hamiltonians connecting them. There are some important differences to the map constructed in the previous theorem. The map between KMS states constructed here is completely positive, but need not preserve the unit. In the lattice setting, locality of the implementing operator near the perturbation does not imply that its conjugation map acts as the identity on distant observables. Indeed, even if $[\eta_P,A]=0$, we have
    \[
        \eta_P A\eta_P^*-A=A(\eta_P\eta_P^*-1).
    \]
    The right-hand side need not be small. Thus the locality of the implementing operator must be distinguished from locality of the observable map near the support of $A$. In particular, it need not approach the identity on observables far from the perturbation. Constructing transport maps with stronger locality and product properties remains a separate problem.
\end{remark}

The following proposition gives the differential equation used to prove the theorem; we refer to it as the QBP equation.

\medskip
\begin{proposition} \label{prop: QBP differential equation}
Let $(\mA, \tau)$ be a $C^*$-dynamical system, $\beta \in \R$, $\omega$ a $(\tau, \beta)$-KMS state, and $P = P^* \in \mA$. Then for any $A \in \mA$ and $s \in [0, 1]$ the following differential equation holds:
\begin{align*}
    \frac{\d}{\d s} \omega_{s P} (A) = - \frac{\beta}{2} \, \omega_{s P} \bigl( \bigl\{ \Phi_\beta^{\tau^{s P}} (P - \omega_{s P} (P)),  \, A \bigr\} \bigr) \, . \numberthis \label{eq: QBP differential eqtn. prop}
\end{align*}
\end{proposition}
As a corollary of this proposition, we find a continuity statement for the map $P \mapsto \omega_P$.

\medskip
\begin{corollary} \label{cor: cont. KMS state}
    Let $(\mA, \tau)$ be a $C^*$-dynamical system, $\omega$ a $(\tau, \beta)$-KMS state, and $P = P^* \in \mA$. Then
    \begin{align*}
        |\omega_P (A) - \omega (A)| \leq |\beta| \, \lVert A \rVert \, \lVert P \rVert 
    \end{align*}
    for all $A \in \mA$.
\end{corollary}
\begin{proof}
    For $\beta=0$ the two states agree. Assume $\beta\neq0$. By stationarity and normalization of $f_\beta$,
    \begin{align*}
        \omega_{sP}\bigl(\Phi_\beta^{\tau^{sP}}(P)\bigr)=\omega_{sP}(P).
    \end{align*}
    Proposition \ref{prop: QBP differential equation} and the Cauchy--Schwarz inequality for $\omega_{sP}$ give
    \begin{align*}
        \left|\frac{\d}{\d s}\omega_{sP}(A)\right|
        &\leq |\beta|\,\norm A
        \sqrt{\omega_{sP}\bigl((\Phi_\beta^{\tau^{sP}}(P))^2\bigr)
                     -\omega_{sP}(P)^2}\\
        &\leq |\beta|\,\norm A\,\norm{\Phi_\beta^{\tau^{sP}}(P)}
        \leq |\beta|\,\norm A\,\norm P.
    \end{align*}
    Integrating over $s\in[0,1]$ proves the claim.
\end{proof}

\section{Locality properties of KMS states} \label{sec: locality of KMS}

We give three concepts that quantify the locality of KMS states for infinitely extended systems, namely, decay of correlations, local perturbations perturb locally (LPPL) and local indistinguishability (LI). In Section \ref{sec: qbp fermions} we establish implications between suitable versions of these notions under the stated uniformity assumptions. For concreteness, we consider the quasi-local algebra $\mA_\Gamma$ and some locally generated $\tau$ with a $(\tau, \beta)$-KMS state $\omega$ for some $\beta \in \R$.

First, we define the notion of decay of correlations. This expresses that the covariance of two distant local operators is small. For one-dimensional translation-invariant finite-range systems, exponential clustering is available at every positive temperature \cite{A69}. More general interactions with tails require additional decay assumptions and, depending on the interaction class, temperature restrictions; see \cite{PP23,KK25}. In arbitrary dimension, high-temperature clustering for finite-range spin and fermion systems is established in \cite{KGRE14}. See also \cite[Section 3]{CMTW25} for a more thorough discussion.

\medskip
\begin{definition} \label{def: decay of corr}
   Let $\beta \in \R$ and $\omega$ be a $(\tau, \beta)$-KMS state on $\mA_\Gamma$. We say that $\omega$ satisfies decay of correlations with respect to the non-increasing function $\zeta_\cor: [0, \infty) \to (0, \infty)$, the non-decreasing function $f_\cor: [0, \infty) \to [0, \infty)$ and $n_\cor \in \N_0$ if
    \begin{align*}
        |\omega(A B) - \omega(A) \, \omega(B)| \leq \lVert A \rVert \, \lVert B \rVert \, |X|^{n_\cor} \, f_\cor(|Y|) \, \zeta_\cor(d(X, Y)) \, ,
    \end{align*}
    for all $X$, $Y \in P_0 (\Gamma)$ and all $A \in \mA_X$, $B \in \mA_Y$, where one of them is even. We require $\lim_{r \to \infty} \zeta_\cor (r) = 0$. Moreover, we assume that $f_\cor(r) \geq 1$ for $r \geq 1$.
\end{definition}

\medskip
\begin{remark}
    For clarity in our exposition we have chosen not to symmetrize the bound in the cardinality of $X$ and $Y$. For disjoint supports the observables commute, since one is even. Interchanging $A$ and $B$ therefore gives the same estimate with $|X|^{n_\cor}f_\cor(|Y|)$ replaced by $|Y|^{n_\cor}f_\cor(|X|)$. As in the preceding work \cite{CMTW25}, we allow the dependence on the support of one of the operators to be exponential, while the dependence on the other support must be polynomial. In the long-range, high-temperature setting of \cite{MSAC26}, the clustering assumption allows exponential dependence on both supports.
\end{remark}

We also deal with extensive perturbations of KMS states. To quantify the locality in such situations, we introduce the notion of $V$-uniform decay of correlations.

\medskip
\begin{definition} \label{def: uniform doc}
    Let $\beta \in \R$ and $\omega$ be a $(\tau, \beta)$-KMS state on $\mA_\Gamma$. Let $V \in \mP_F$ be an extensive perturbation. We say that $\omega$ satisfies $V$-uniform decay of correlations with respect to $\zeta_\cor$, $f_\cor$ and $n_\cor$ if the states $\omega_{s V_M}$ satisfy decay of correlations uniformly in $M \in P_0 (\Gamma)$ and $s \in [0, 1]$.
\end{definition}

Next, we define the notion of LPPL for quasi-local perturbations. This expresses that a local perturbation $P$ has only a small effect on the expectation values of observables supported far from it.

\medskip
\begin{definition} \label{def: LPPL}
    Let $\beta \in \R$ and $\omega$ be a $(\tau, \beta)$-KMS state on $\mA_\Gamma$. We say that $\omega$ satisfies LPPL with respect to the non-decreasing functions $f_\lppl: [0, \infty) \to [0, \infty)$, $g_\lppl: [0, \infty) \to [0, \infty)$ and the non-increasing function $\zeta_\lppl: [0, \infty) \to [0, \infty)$ if
    \begin{align*}
        |\omega_P (A) - \omega (A)| \leq \lVert A \rVert  \, f_\lppl (|Y| )  \, g_\lppl(\lVert P \rVert_{F, x}) \, \zeta_\lppl (d(x, Y))  \, ,
    \end{align*}
    for all $Y \in P_0(\Gamma)$, $x \in \Gamma$ and $P = P^* \in \mA_F^+$, $A \in \mA_Y$. We require $\lim_{r\to\infty}\zeta_\lppl(r)=0$.
\end{definition}
The third notion we use to quantify the locality of KMS states is LI. We formulate LI relative to a fixed KMS state $\omega$ of $H$ and a perturbing interaction $V$. The comparison is between the states obtained by adding $V$ on different finite patches. For $H=0$, these are finite-volume Gibbs states on the patches, extended by the tracial state outside, and the definition recovers the corresponding finite-volume comparison, see, for example, \cite[Definition 3]{CMTW25}.

\medskip
\begin{definition} \label{def: LI}
    Let $\beta \in \R$ and $\omega$ be a $(\tau, \beta)$-KMS state on $\mA_\Gamma$. We fix an extensive perturbation $V \in \mP_F$. We say that $\omega$ satisfies LI with respect to the non-increasing function $\zeta_\li: [0, \infty) \to [0, \infty)$ and some non-decreasing function $f_\li: [0, \infty) \to [0, \infty)$, $g_\li: [0, \infty) \to [0, \infty)$ if
    \begin{align*}
        |\omega_{V_Z} (A) - \omega_{V_Y} (A)| \leq \lVert A \rVert \, f_\li (|X|) \,  g_\li(\lVert V \rVert_F) \, \zeta_\li (d(X, Z \setminus Y))  \, ,
    \end{align*}
    for all $X \subsetneq Y \subsetneq Z \in P_0 (\Gamma)$ and $A \in \mA_X$. We require that $ \lim_{r \to \infty}\zeta_\li (r) =0$.
\end{definition}

\medskip
\begin{remark}
LI is sometimes described by saying that enlarging a finite box has an asymptotically negligible effect on observables far from its boundary. Our definition is stronger, since the bound is uniform over all finite regions, including disconnected ones.
\\
Consider the nearest-neighbor ferromagnetic Ising model on $\mathbb Z^2$ at zero magnetic field and inverse temperature $\beta>\beta_c$, with $H=0$ and open boundary conditions for the finite patches of $V$. Along increasing boxes, the states $\omega_{V_X}$ converge to the equal-weight mixture of the plus and minus phases by spin-flip symmetry and \cite[Propositions 6.9 and 8.5]{GHM01}. Nevertheless, LI fails. Indeed, let $r\geq1$, choose $x,y\in\mathbb Z^2$ with $d(x,y)>2r+1$, and set $X=\{x,y\}$ and $Y=B_r(x)\cup B_r(y)$. Since the two balls do not interact, factorization and spin-flip symmetry give
\begin{align*}
    \omega_{V_Y}(\sigma_x^z\sigma_y^z)=0.
\end{align*}
Writing $m_\beta=\omega^+(\sigma_0^z)>0$, convergence along boxes and Griffiths' inequality show that a sufficiently large box $Z\supsetneq Y$ can be chosen such that
\begin{align*}
    \omega_{V_Z}(\sigma_x^z\sigma_y^z)\geq\frac{m_\beta^2}{2}.
\end{align*}
Since $|X|=2$ and $d(X,Z\setminus Y)>r$, LI would imply
\begin{align*}
    0<\frac{m_\beta^2}{2}
    &\leq\bigl|\omega_{V_Z}(\sigma_x^z\sigma_y^z)
              -\omega_{V_Y}(\sigma_x^z\sigma_y^z)\bigr|\\
    &\leq f_\li(2)g_\li(\lVert V\rVert_F)\zeta_\li(r).
\end{align*}
Sending $r\to\infty$ gives a contradiction. Thus considering only connected boxes does not capture LI.
\end{remark}

LI implies that for any IAS $(\Lambda_n)_{n \in \N}$ the corresponding sequence of perturbed KMS states $(\omega_{V_n})_{n \in \N}$ converges to a $(\tau^V, \beta)$-KMS state. We have the following proposition.

\medskip
\begin{proposition} \label{prop: convergence from li}
    Let $H,V\in\mP_F$ for a decay function $F$, and assume that $\mL_{H+sV}$ generates a dynamics $\tau^{sV}$ for every $s\in[0,1]$. Let $\omega$ be a $(\tau^{0V},\beta)$-KMS state. Fix $s\in[0,1]$ and assume that $\omega$ satisfies LI for $sV$. For any IAS $(\Lambda_n)_{n\in\N}$, set $V_n=V_{\Lambda_n}$. Then there is a $(\tau^{sV},\beta)$-KMS state $\omega_{sV}$ such that
    \begin{align*}
        \lim_{n\to\infty}\omega_{sV_n}(A)=\omega_{sV}(A),\qquad A\in\mA_\Gamma.
    \end{align*}
    The limit is independent of the IAS. If LI holds for all $sV$, $s\in[0,1]$, with the same functions $\zeta_\li$, $f_\li$ and $g_\li$, then the convergence is uniform in $s\in[0,1]$.
\end{proposition}

\begin{proof}
    Let $A\in\mA_X$. For all sufficiently large $n$ and $m>n$, LI gives
    \begin{align*}
        |\omega_{sV_m} (A) - \omega_{sV_n} (A)| & \leq \norm A \,  f_\li(|X|) \, g_\li(s\norm V_F)
        \zeta_\li(d(X,\Lambda_m\setminus\Lambda_n))\\
        &\leq \norm A f_\li(|X|)g_\li(\norm V_F)
        \zeta_\li(d(X,\Gamma\setminus\Lambda_n)).
    \end{align*}
    Every finite ball around $X$ is eventually contained in $\Lambda_n$. Hence the last expression tends to zero, uniformly in $s$ whenever the LI functions are uniform. For fixed $s$, the limit defines a linear functional on $\mA_0$ satisfying
    \begin{align*}
        \omega_{sV}(1)=1,\qquad
        |\omega_{sV}(A)|\leq\norm A,\qquad
        \omega_{sV}(A^*A)=\lim_{n\to\infty}\omega_{sV_n}(A^*A)\geq0.
    \end{align*}
    It therefore extends uniquely to a state on $\mA_\Gamma$. For $A\in\mA_\Gamma$ and $A_k\in\mA_0$ with $\norm{A-A_k}\to0$, we calculate
    \begin{align*}
        |\omega_{sV}(A)-\omega_{sV_n}(A)|
        &\leq 2\norm{A-A_k}
        +|\omega_{sV}(A_k)-\omega_{sV_n}(A_k)|.
    \end{align*}
    First choosing $k$ and then $n$ proves convergence on $\mA_\Gamma$, with the same uniformity in $s$. Lemma \ref{lem: str conv time evol} and \cite[Proposition 5.3.25]{BR2} show that the limit is a $(\tau^{sV},\beta)$-KMS state.

    To show independence of the IAS, let $(\widetilde\Lambda_n)_{n\in\N}$ be another IAS. For $A\in\mA_X$, choose $n$ so large that both $\Lambda_n$ and $\widetilde\Lambda_n$ contain the fattening $X_r$ for a given $r$ of $X$. Comparing both states with the state for $\Lambda_n\cup\widetilde\Lambda_n$ gives
    \begin{align*}
        |\omega_{sV_{\Lambda_n}}(A)-\omega_{sV_{\widetilde\Lambda_n}}(A)|
        &\leq 2\norm A f_\li(|X|)g_\li(\norm V_F)\zeta_\li(r).
    \end{align*}
    Taking $n\to\infty$ and then $r\to\infty$ proves equality of the limits on $\mA_0$, and hence on $\mA_\Gamma$.
\end{proof}

\medskip
\begin{remark}
    LI does not imply uniqueness of the KMS state for the perturbed dynamics, even if the reference dynamics is trivial while $V$ generates a nontrivial dynamics. Consider the spin system on $\mathbb Z^2$ and define
    \begin{align*}
        V(\{x\})&=-4\sigma_x^z,\\
        V(\{x,y\})&=-\sigma_x^z\sigma_y^z+\sigma_x^z+\sigma_y^z,
        \qquad d(x,y)=1,
    \end{align*}
    with $V(M)=0$ otherwise. For a finite region $\Lambda$,
    \begin{align*}
        \sum_{M\subseteq\Lambda}V(M)
        =-\sum_{\substack{\{x,y\}\subseteq\Lambda\\d(x,y)=1}}
             \sigma_x^z\sigma_y^z
         -\sum_{x\in\Lambda}
             \bigl(4-\deg_\Lambda(x)\bigr)\sigma_x^z,
    \end{align*}
    where $\deg_\Lambda(x)=|\{y\in\Lambda:d(x,y)=1\}|$. This is the nearest-neighbor Ising Hamiltonian with $+$ boundary condition. At sufficiently large $\beta$, these Gibbs states satisfy exponential decay of correlations uniformly over all finite $\Lambda$ by \cite[Corollary~7.11 and Theorem~7.12]{GHM01}. Arbitrary regions are covered by the volume monotonicity in \cite[Eq.~(14)]{GHM01}, since the connection event in Theorem~7.12 is decreasing. Since $V$ is finite-range, Theorems \ref{thm: LPPL without doc} and \ref{thm: li from lppl} imply LI for the tracial state with respect to $V$. The estimate extends to quantum observables because diagonal projection in the $\sigma^z$-basis preserves expectations and products on disjoint supports, and contracts the operator norm.
    \\
    In infinite volume, the four linear edge contributions at each site cancel the singleton term. Hence $V$ generates the standard zero-field Ising dynamics. The states selected by its finite patches converge to the plus phase, whereas for $\beta>\beta_c$ the same dynamics has the two distinct KMS states $\omega^+$ and $\omega^-$. Thus LI selects a state for the specified finite patches but does not imply uniqueness. It also depends on the chosen interaction, rather than only on the infinite-volume dynamics.
\end{remark}

\subsection{Relations between decay of correlations, LPPL, and LI} \label{sec: qbp fermions}

In this section, we use QBP, as introduced in Section \ref{sec: general qbp}, to relate locality properties of KMS states of infinitely extended fermionic systems. For perturbations of the tracial state, the three implications proved below form the following cycle under their respective decay, support and uniformity assumptions (see also \cite[Figure 1]{CMTW25}):

\begin{center}
\begin{tikzpicture}[
    node distance=3.2cm,
    locality/.style={rounded corners, align=center, inner sep=6pt},
    implication/.style={ ->, >=latex, thick}
]
    \node[locality] (cor) {Decay of\\correlations};
    \node[locality,right=of cor] (lppl) {LPPL};
    \node[locality,right=of lppl] (li) {LI};
    \draw[implication] (cor) -- node[above, align=center, inner sep=2pt] {\scriptsize Theorems \ref{thm: LPPL}\\[-1mm]\scriptsize and \ref{thm: LPPL without doc}} (lppl);
    \draw[implication] (lppl) -- node[above, inner sep=2pt] {\scriptsize Theorem \ref{thm: li from lppl}} (li);
    \draw[implication] (li) to[bend left=35] node[below, inner sep=2pt] {\scriptsize Theorem \ref{thm: doc from LI}} (cor);
\end{tikzpicture}
\end{center}

Theorem \ref{thm: doc from LI + doc} is a separate stability statement: it starts from an already clustering KMS state and shows that exponential LI yields stretched-exponential clustering with an exponential support prefactor under an extensive perturbation, without assuming uniqueness. Theorem \ref{thm: QBP SLT} gives continuity and differentiability under the stronger assumption of clustering uniformly in both the finite patch and the coupling. The proofs are given in Section \ref{sec: proof fermions}. As a preparation, Proposition \ref{prop: loc of phi} establishes decay-norm bounds for the operator $\Phi_\beta^\tau$ defined in Equation \eqref{def: Phi}; its proof is given in Appendix \ref{app: LR bounds}.

\medskip
\begin{proposition}\label{prop: loc of phi}
    Let $\Psi$ and $V$ be interactions, $P=P^*\in\mA_\Gamma^+$ and $\beta\in\R$. For $X\in P_0(\Gamma)$, let $\tau^{X,P}$ be the dynamics generated by $\mL_\Psi+\mL_{V_X+P}$. We distinguish two regimes:
    \begin{enumerate}
    \item \textbf{Polynomial decay.} Fix $C_\inter>0$, $\nu>0$ and $\mu>3D+\nu$. Assume that $\Psi,V\in\mP_\mu$, $P\in\mA_\nu^+$ and
    \begin{align*}
        \norm\Psi_\mu+\norm V_\mu<C_\inter.
    \end{align*}
    Then there are polynomials $f$ and $h$ such that
    \begin{align*}
        \norm{\tau_t^{X,P}(A)}_{\nu,x}
        &\leq f(|t|)(1+\norm P_{\nu,x})\norm A_{\nu,x},\\
        \norm{\Phi_\beta^{\tau^{X,P}}(A)}_{\nu,x}
        &\leq h(|\beta|)(1+\norm P_{\nu,x})\norm A_{\nu,x}
    \end{align*}
    for $t\in\R$, $x\in\Gamma$ and $A\in\mA_\nu$. Their coefficients depend only on the lattice, $\mu$, $\nu$ and $C_\inter$, and can be chosen uniformly in $X$, $P$, $\Psi$ and $V$ subject to these assumptions.

    \item \textbf{Exponential decay.} Fix $a,b,c>0$ such that $b > 4 a$ and $c > 8 a$. There exists $C_\inter=C_\inter(a,c)>0$ such that, if $P\in\mA_{\exp(-b\cdot)}^+$ and
    \begin{align*}
        \norm\Psi_{\exp(-c\cdot)}+\norm V_{\exp(-c\cdot)}
        <\frac{C_\inter}{1+|\beta|},
    \end{align*}
    then
    \begin{align*}
        \norm{\Phi_\beta^{\tau^{X,P}}(A)}_{\exp(-a\cdot),x}
        \leq C(1+|\beta|)(1+\norm P_{\exp(-b\cdot),x})
        \norm A_{\exp(-b\cdot),x}
    \end{align*}
    for $x\in\Gamma$ and $A\in\mA_{\exp(-b\cdot)}$. The constant $C>0$ depends only on the lattice, $a$, $b$ and $c$, and is uniform in $X$, $P$, $\beta$, $\Psi$ and $V$ subject to these assumptions.
    \end{enumerate}
\end{proposition}

We start with the LPPL statement assuming decay of correlations along a path of KMS states. This statement generalizes \cite[Theorem 25]{CMTW25}, where only finite systems, lattice spins and strictly local perturbations were considered.

\medskip
\begin{theorem}[LPPL from clustering along a path]\label{thm: LPPL}
    Let $H$ be an interaction, $\tau$ the dynamics generated by $\mL_H$, $\beta\in\R$, and $\omega$ a $(\tau,\beta)$-KMS state. Fix $P=P^*\in\mA_\Gamma^+$ and assume that $\omega_{sP}$ satisfies decay of correlations with respect to $\zeta_\cor$, $f_\cor$ and $n_\cor$, uniformly in $s\in[0,1]$.
    \begin{enumerate}
    \item \textbf{Polynomial decay.} Fix $C_\inter>0$, $\nu>0$ and $\mu>3D+\nu$. Assume that $H\in\mP_\mu$, $P\in\mA_\nu^+$ and
    \begin{align*}
        \norm H_\mu<C_\inter,\qquad
        (1+r)^{-\nu}\leq(1+r)^{Dn_\cor}\zeta_\cor(r),\quad r\geq0.
    \end{align*}
    Then there is a polynomial $g$ such that
    \begin{align*}
        |\omega_P(A)-\omega(A)|
        &\leq g(|\beta|)\norm A f_\cor(|Y|)
           (1+d(x,Y))^{Dn_\cor}\zeta_\cor(d(x,Y)/2)\\
        &\quad\times(1+\norm P_{\nu,x})\norm P_{\nu,x}.
    \end{align*}

    \item \textbf{Exponential decay.} Fix $a,b,c>0$ such that $b> 4 a$ and $c> 8 a$, and let $C_\inter=C_\inter(a,c)$ be as in Proposition \ref{prop: loc of phi}. Assume that $H\in\mP_{\exp(-c\cdot)}$, $P\in\mA_{\exp(-b\cdot)}^+$ and
    \begin{align*}
        \norm H_{\exp(-c\cdot)}&<\frac{C_\inter}{1+|\beta|}, \qquad \exp(-ar) \leq(1+r)^{Dn_\cor}\zeta_\cor(r),\qquad r\geq0.
    \end{align*}
    Then there is a constant $C>0$ such that
    \begin{align*}
        |\omega_P(A)-\omega(A)|
        &\leq C|\beta|(1+|\beta|)\norm A f_\cor(|Y|)
          (1+d(x,Y))^{Dn_\cor}\zeta_\cor(d(x,Y)/2)\\
        &\quad\times(1+\norm P_{\exp(-b\cdot),x})
          \norm P_{\exp(-b\cdot),x}.
    \end{align*}
    \end{enumerate}
    Both bounds hold for $x\in\Gamma$, $Y\in P_0(\Gamma)$ and $A\in\mA_Y$. The coefficients of $g$ and the constant $C$ are independent of $H$, $P$ and $\beta$ subject to the stated restrictions.
\end{theorem}
In the next theorem we show LPPL assuming decay of correlations only for the initial KMS state. The price is exponential, rather than polynomial, dependence on the decay norm of the perturbation. This extends \cite[Theorem 22]{CMTW25} to infinite fermionic systems and quasi-local perturbations. Common clustering data give bounds uniform over finite background patches. We include these patches in the statement, as needed for the passage to extensive perturbations.

\medskip
\begin{theorem}[LPPL from clustering of the initial state]\label{thm: LPPL without doc}
    Let $H,V$ be interactions, $X\in P_0(\Gamma)$, and $\tau^X$ the dynamics generated by $\mL_H+\mL_{V_X}$. Let $\beta\in\R$ and let $\omega_{V_X}$ be a $(\tau^X,\beta)$-KMS state satisfying decay of correlations with respect to $\zeta_\cor$, $f_\cor$ and $n_\cor$. Fix $P=P^*\in\mA_\Gamma^+$ and write $\omega_{V_X+P}=(\omega_{V_X})_P$.
    \begin{enumerate}
    \item \textbf{Polynomial decay.} Fix $C_\inter>0$, $\nu>0$ and $\mu>3D+\nu$. Assume that $H,V\in\mP_\mu$, $P\in\mA_\nu^+$ and
    \begin{align*}
        \norm H_\mu+\norm V_\mu&<C_\inter, \qquad (1+r)^{-\nu}\leq(1+r)^{Dn_\cor}\zeta_\cor(r),\qquad r\geq0.
    \end{align*}
    Then there are a polynomial $g$ and a constant $C>0$ such that
    \begin{align*}
        |\omega_{V_X+P}(A)-\omega_{V_X}(A)|
        &\leq C\norm A f_\cor(|Y|)(1+d(x,Y))^{Dn_\cor}
            \zeta_\cor(d(x,Y)/2)\\
        &\quad\times\Bigl(\exp\bigl(g(|\beta|)(1+\norm P_{\nu,x})
                 \norm P_{\nu,x}\bigr)-1\Bigr).
    \end{align*}

    \item \textbf{Exponential decay.} Fix $a,b,c>0$ such that $b > 4a$ and $c> 8 a$, and let $C_\inter=C_\inter(a,c)$ be as in Proposition \ref{prop: loc of phi}. Assume that $H,V\in\mP_{\exp(-c\cdot)}$, $P\in\mA_{\exp(-b\cdot)}^+$ and
    \begin{align*}
        \norm H_{\exp(-c\cdot)}+\norm V_{\exp(-c\cdot)}
        &<\frac{C_\inter}{1+|\beta|},\\
        \exp(-ar)&\leq(1+r)^{Dn_\cor}\zeta_\cor(r),\qquad r\geq0.
    \end{align*}
    Then there are constants $C$ ,$C^\prime>0$ such that
    \begin{align*}
        &|\omega_{V_X+P}(A)-\omega_{V_X}(A)|\\
        &\quad\leq C\norm A f_\cor(|Y|)(1+d(x,Y))^{Dn_\cor}
              \zeta_\cor(d(x,Y)/2)\\
        &\qquad\times\Bigl(\exp \bigl( C^\prime|\beta|(1+|\beta|)
              (1+\norm P_{\exp(-b\cdot),x})
              \norm P_{\exp(-b\cdot),x} \bigr)-1\Bigr).
    \end{align*}
    \end{enumerate}
    Both bounds hold for $x\in\Gamma$, $Y\in P_0(\Gamma)$ and $A\in\mA_Y$. The coefficients of $g$ and the constants $C$ and $C^\prime$ are independent of $H$, $V$, $X$, $P$ and $\beta$ subject to the stated restrictions.
\end{theorem}

Next, we deal with extensive perturbations $V$. Uniform LPPL for the states $\omega_{V_M}$ over finite patches $M$ implies LI at fixed coupling. Uniformity in the coupling is needed only for the final uniform assertion. This extends \cite[Theorem 29]{CMTW25} to infinite fermionic systems.

\medskip
\begin{theorem}[LI from uniform LPPL] \label{thm: li from lppl}
    Let $H$, $V\in\mP_F$ for a decay function $F$, and let $\tau$ be the dynamics generated by $\mL_H$. Let $\beta\in\R$ and let $\omega$ be a $(\tau,\beta)$-KMS state. Assume that $\omega_{V_M}$ satisfies LPPL with respect to $f_\lppl$, $g_\lppl$ and $\zeta_\lppl$, uniformly in $M\in P_0(\Gamma)$. Then
    \begin{align*}
        |\omega_{V_Z}(A)-\omega_{V_Y}(A)|
        &\leq \norm A f_\lppl(|X|)g_\lppl(3\norm V_F)
        \sum_{z\in Z\setminus Y}\zeta_\lppl(d(z,X))
    \end{align*}
    for $X\subseteq Y\subsetneq Z\in P_0(\Gamma)$ and $A\in\mA_X$.

    If $\sum_{k=0}^\infty(1+k)^{D-1}\zeta_\lppl(k)<\infty$, define
    \begin{align*}
        h(r) \coloneq \sum_{k=\lceil r\rceil}^\infty(1+k)^{D-1}\zeta_\lppl(k),\qquad r\geq0.
    \end{align*}
    Then
    \begin{align*}
        |\omega_{V_Z}(A)-\omega_{V_Y}(A)|
        &\leq C_\sur\norm A\,|X|f_\lppl(|X|)
        g_\lppl(3\norm V_F)h(d(X,Z\setminus Y))
    \end{align*}
    for the same $X$, $Y$, $Z$, $A$. Moreover, if $\omega_{sV_M}$ satisfies LPPL with the same functions uniformly in $M\in P_0(\Gamma)$ and $s\in[0,1]$, then both bounds hold with $\omega_{V_Z}$ and $\omega_{V_Y}$ replaced by $\omega_{sV_Z}$ and $\omega_{sV_Y}$, uniformly in $s$.
\end{theorem}

In the next theorem we derive a bound on the decay of correlations from LI for perturbations of the tracial state. The bound holds uniformly for all finite patches and also for the limiting KMS state. This theorem extends the setting of \cite[Theorem 31]{CMTW25} to infinitely extended fermionic systems.

\medskip
\begin{theorem}[Decay of correlations from LI for the tracial state] \label{thm: doc from LI}
    Let $H \in \mP_F$ for some decay function $F$ and let $\tau$ be the dynamics generated by $\mL_{H}$. Let $\beta \in \R$ and let $\omega$ be the tracial state\footnote{The tracial state is the $(\Tilde{\tau}, \beta)$-KMS state for the dynamics $\Tilde{\tau}$ generated by $\mL_0$}. Assume that $\omega$ satisfies LI for the perturbation $H$ with respect to $\zeta_\li$, $f_\li$ and $g_\li$. Assume that
    \begin{align*}
        (1 + r)^D \, F(r) \leq \zeta_\li (r) \, , \quad r \geq 0 \, .
    \end{align*}
    Denote by $\omega_H$ the $(\tau, \beta)$-KMS state constructed in Proposition \ref{prop: convergence from li}. It follows that there exists a constant $C > 0$ that depends only on the lattice, $F$, the LI data and $\lVert H \rVert_F$ such that
    \begin{align*}
        |\omega_H(AB)-\omega_H(A)\omega_H(B)|
        &\leq C(1+|\beta|)\bigl(f_\li(|X|+|Y|)+|X|\bigr)\\
        &\quad\times\zeta_\li(d(X,Y)/3)\norm A\norm B,
    \end{align*}
    for all $X$, $Y \in P_0 (\Gamma)$ and $A \in \mA_X$, $B \in \mA_Y$.
    The same bound holds with $\omega_H$ replaced by $\omega_{H_M}$, uniformly in $M\in P_0(\Gamma)$.
\end{theorem}

The following theorem treats an already interacting reference state. Exponential clustering is assumed only for the initial KMS state $\omega$, while exponential LI controls finite patches of $V$.

\medskip
\begin{theorem}[Stretched-exponential decay of correlations under LI]\label{thm: doc from LI + doc}
    Let $\beta\in\R$, $c>0$ and $H$, $V\in\mP_{\exp(-c\cdot)}$. Let $\omega$ be a $(\tau,\beta)$-KMS state, where $\tau$ is generated by $\mL_H$. Assume that $\omega$ satisfies decay of correlations and LI for $V$ with data
    \begin{align*}
        \text{clustering:}\qquad
        \zeta_\cor(r)&=C_\cor\exp(-\mu_\cor r),&
        f_\cor(r)&=C_\cor\exp(C_\cor(1+r)),\quad n_\cor\in\N_0,\\
        \text{LI:}\qquad
        \zeta_\li(r)&=C_\li\exp(-\mu_\li r),&
        f_\li(r)&=C_\li\exp(C_\li(1+r)),\quad g_\li,
    \end{align*}
    where $C_\cor$, $C_\li \geq 1$ and $\mu_\cor$, $\mu_\li>0$. Let $\omega_V$ be the KMS state for $H+V$ selected by Proposition \ref{prop: convergence from li}. There are constants $C,\kappa>0$ such that, for $X$, $Y\in P_0(\Gamma)$ and $A\in\mA_X$, $B\in\mA_Y$, provided that at least one observable is even,
    \begin{align*}
        |\omega_V(AB)-\omega_V(A)\omega_V(B)|
        &\leq C\exp\bigl(C(1+|X|+|Y|)\bigr) \, \norm A \, \norm B\\
        &\quad\times
        \exp\left(-\kappa(1+d(X,Y))^{1/((D+1)(D+2))}\right).
    \end{align*}
    The constants $C$ and $\kappa$ are independent of $X$, $Y$, $A$ or $B$.
\end{theorem}

\medskip
\begin{remark}
The proof of the preceding theorem gives the sharper support-dependent estimate
\begin{align*}
    |\omega_V(AB)-\omega_V(A)\omega_V(B)|
    &\leq  C f_\li (|X|+|Y|) \,  \norm A \, \norm B \,  \exp \left( -\kappa
    \left(\frac{(1+d(X,Y))^{1/(D+2)}}{1+|X|+|Y|}\right)^{1/D}\right).
\end{align*}
Using this bound, polynomial growth of the LI support function yields super-polynomial decay with a polynomial support prefactor. Moreover, choosing polynomial growth of the support function in the clustering assumption permits a better stretched-exponential exponent. We have chosen the assumptions and the bound in the theorem because they allow exponential support dependence while keeping the statement transparent.
\end{remark}

\subsection{Continuity under extensive perturbations}

As a final result, we estimate the change of local expectations under extensive perturbations $V$. We require decay of correlations uniformly in both the finite patch and the coupling, as in Definition \ref{def: uniform doc}. Under this assumption, we obtain continuity in the topology of pointwise convergence on local observables and a continuously differentiable path of expectations. This extends \cite[Theorem 34]{CMTW25} to infinitely extended fermionic systems, see also \cite{MSAC26} and \cite[Appendix A]{RGK24}.

\medskip
\begin{theorem}[Continuity under extensive perturbations]\label{thm: QBP SLT}
    Let $\beta\in\R$, $C_\inter>0$ and $\mu$, $\nu>0$. Let $H$, $V\in\mP_\mu$. For $s\in[0,1]$, let $\tau^{sV}$ be the dynamics generated by $\mL_{H+sV}$. Let $\omega$ be a $(\tau^{0V},\beta)$-KMS state satisfying $V$-uniform decay of correlations with respect to
    \begin{align*}
        \zeta_\cor(r)=C_\cor(1+r)^{-m_\cor},\qquad f_\cor,\qquad n_\cor,
    \end{align*}
    for some $C_\cor>0$. Assume that
    \begin{align*}
        \nu>D,\qquad \mu>3D+\nu,\qquad m_\cor\geq\nu+Dn_\cor,\qquad
        \norm H_\mu+\norm V_\mu<C_\inter.
    \end{align*}
    For each $s\in[0,1]$, denote by $\omega_{sV}$ the $(\tau^{sV},\beta)$-KMS state constructed in Proposition \ref{prop: convergence from li}. Then the following holds:
    \begin{enumerate}
    \item \textbf{Continuity of local expectations.} There is a polynomial $g$ such that
    \begin{align*}
        |\omega_{sV}(A)-\omega_{tV}(A)|
        \leq |s-t|\,g(|\beta|)\norm V_\nu\norm A\,|Y|f_\cor(|Y|)
    \end{align*}
    for $Y\in P_0(\Gamma)$, $A\in\mA_Y$ and $s,t\in[0,1]$. Its coefficients depend only on the lattice, the decay exponents, $C_\inter$ and the common clustering data, but not on $V$ within this class.

    \item \textbf{A continuously differentiable path.} For every $A\in\mA_0$, the map $s\mapsto\omega_{sV}(A)$ is continuously differentiable, with one-sided derivatives at the endpoints, and
    \begin{align*}
        \frac{\d}{\d s}\omega_{sV}(A)
        =-\frac{\beta}{2}\sum_{x\in\Gamma}
        \omega_{sV}\bigl(\{\Phi_\beta^{\tau^{sV}}
            (V_x-\omega_{sV}(V_x)),A\}\bigr).
\end{align*}
    The sum converges absolutely and uniformly in $s\in[0,1]$.
    \end{enumerate}
\end{theorem}

\section{Proofs of the QBP results} \label{sec: proof general qbp}

\begin{proof}[Proof of Theorem \ref{thm: QBP}]
Let $A \in \mA$ be arbitrary. Proposition \ref{prop: QBP differential equation} gives a differential equation for $s \mapsto \omega_{sP}(A)$. We construct a functional $\phi_{sP}(A)$ by a Dyson series and show that it satisfies the same differential equation and initial condition. A Grönwall argument then identifies the two functionals.

Proposition \ref{prop: QBP differential equation} gives
\begin{align*}
    \frac{\d}{\d s} \omega_{s P} (A) = - \frac{\beta}{2} \, \omega_{s P} \bigr( \bigr\{ \Phi_\beta^{\tau^{s P}} (P - \omega_{sP} (P)),  \, A \bigr\} \bigr) \, .
\end{align*}
Since all involved operators are bounded, we can solve this differential equation explicitly by the Dyson series. That is, we define the operators
\begin{align*} \label{eq: eta}
    \eta_{s P} &\coloneq \mT \exp \bigr( - \frac{\beta}{2} \int_0^s \Phi_\beta^{\tau^{\sigma P}} (P - \omega_{\sigma P} (P)) \, \d  \sigma \bigr) \numberthis
    \\
    &\coloneq \sum_{n = 0}^\infty (- \tfrac{\beta}{2})^n \, \int_{0 \leq \sigma_{n} \leq \cdots \leq \sigma_1 \leq s} \Phi_\beta^{\tau^{ \sigma_n P}} (P - \omega_{\sigma_n P}(P)) \dots \Phi_\beta^{\tau^{ \sigma_1 P}} (P - \omega_{\sigma_1 P}(P)) \, \d \sigma_n \ldots \d \sigma_1 \, ,
\end{align*}
where $\mT$ denotes the time-ordering operator. The adjoint operator is given by
\begin{align*}
    \eta_{s P}^* &\coloneq \overline{\mT} \exp \bigr( - \frac{\beta}{2} \int_0^s \Phi_\beta^{\tau^{\sigma P}} (P - \omega_{\sigma P} (P)) \, \d  \sigma \bigr)
    \\
    &\coloneq \sum_{n = 0}^\infty (- \tfrac{\beta}{2})^n \, \int_{0 \leq \sigma_1 \leq \cdots \leq \sigma_n \leq s }\, \Phi_\beta^{\tau^{\sigma_n P}} (P - \omega_{\sigma_n P}(P)) \dots \Phi_\beta^{\tau^{\sigma_1 P}} (P - \omega_{\sigma_1 P}(P)) \, \d \sigma_1 \ldots \d \sigma_n \, ,
\end{align*}
where $\overline{\mT}$ denotes the reverse time-ordering operator. From Lemma \ref{lem: continuity of phi} and the Dyson series it follows that the maps $s \mapsto \eta_{sP}$ and $s \mapsto \eta_{s P}^*$ are differentiable with derivative
\begin{align*} \label{eq: derivative eta}
    \frac{\d}{\d s} \eta_{s P } = - \frac{\beta}{2} \,\eta_{s P} \, \Phi_\beta^{\tau^{s P}} (P - \omega_{s P} (P)) \, , \numberthis 
\end{align*}
and 
\begin{align*} \label{eq: derivative eta*}
    \frac{\d}{\d s} \eta_{s P }^* = - \frac{\beta}{2} \, \Phi_\beta^{\tau^{s P}} (P - \omega_{s P} (P)) \,\eta_{s P}^* \, , \numberthis
\end{align*}
respectively. Using the operator $\eta_{sP} \in \mA$ we define the functional 
\begin{align*}
    \phi_{sP} (A) \coloneq \omega(\eta_{s P} \, A \, \eta_{sP}^*) \, .
\end{align*}
We have $\phi_{0P}=\omega$, and Equations \eqref{eq: derivative eta} and \eqref{eq: derivative eta*} give
\begin{align*}
    \frac{\d}{\d s} \phi_{s P} (A) &= - \frac{\beta}{2} \, \omega \bigr( \eta_{s P} \, \bigr\{ \Phi_\beta^{\tau^{s P}} (P - \omega_{sP} (P)),  \, A \bigr\}  \, \eta_{sP}^* \bigr)
    \\
    &= - \frac{\beta}{2} \, \phi_{sP} \bigr( \bigr\{ \Phi_\beta^{\tau^{s P}} (P - \omega_{sP} (P)),  \, A \bigr\} \bigr) \, .
\end{align*}
Thus, the maps $s \mapsto \omega_{s P} (A)$ and $s \mapsto \phi_{s P} (A)$ satisfy the same initial value problem for every $A \in \mA$. To identify the two functionals, first note that
\begin{align*}
    \lVert \eta_{sP} \rVert \leq \exp ( |\beta| \, \lVert P \rVert ) \, ,
\end{align*}
since $\lVert \Phi_\beta^{\tau^{s P}} (P) \rVert \leq \lVert P \rVert$. Consequently,
\begin{align*}
    |\phi_{s P} (A)| \leq \exp(2 \, |\beta | \, \lVert P \rVert ) \, \lVert A \rVert \, , 
\end{align*}
which yields $\sup_{s\in[0,1]}\norm{\phi_{sP}}<\infty$. Integrating the two scalar differential equations, and using $\norm{\Phi_\beta^{\tau^{sP}}(P-\omega_{sP}(P))}\leq2\norm P$, gives
\begin{align*}
    |\omega_{sP}(A)-\omega_{uP}(A)|
    &\leq 2|\beta|\norm P\norm A\,|s-u|,\\
    |\phi_{sP}(A)-\phi_{uP}(A)|
    &\leq 2|\beta|\norm P\e^{2|\beta|\norm P}\norm A\,|s-u|.
\end{align*}
Taking the supremum over $\norm A\leq1$ shows norm-continuity of both paths in $\mA^*$. Set $\delta_{sP}=\omega_{sP}-\phi_{sP}$. We have that $\delta_{0P}=0$. Since both functionals satisfy the same differential equation we also have 
\begin{align*}
    \delta_{sP}(A)
    &=-\frac{\beta}{2}\int_0^s
      \delta_{uP}\bigl(\{\Phi_\beta^{\tau^{uP}}(P-\omega_{uP}(P)),A\}\bigr)\,\d u.
\end{align*}
This yields the estimate 
\begin{align*}
    \norm{\delta_{sP}}
    &\leq2|\beta|\norm P\int_0^s\norm{\delta_{uP}}\,\d u.
\end{align*}
The preceding norm estimates show that the integrand is continuous. Grönwall's lemma gives $\delta_{sP}=0$, which concludes the proof.
\end{proof}

\medskip
\begin{proof}[Proof of Proposition \ref{prop: QBP differential equation}]
For $\beta=0$, Equation \eqref{eq: perturbed vector} gives $\Omega_{sP}=\Omega$, so $\omega_{sP}=\omega$ and the derivative is zero. Assume $\beta\neq0$.
Let $A\in\mA$ be arbitrary. First, we assume that $P=P^*\in\mA_\ent^\tau$, and hence $P\in\mA_\ent^{\tau^{sP}}$ for all $s\in[0,1]$. The general case will then follow by approximating $P$.

Assuming entire analyticity, we show the differentiability of the map $s \mapsto \Omega_{s P}$. Here, $\Omega_{s P}$ is the vector defined in Equation \eqref{eq: perturbed vector}. The following lemma was proved in \cite[Lemma 2]{EO19} for $W^*$-dynamical systems. Their argument can also be adapted to $C^*$-dynamical systems. Here, we give a different proof based on Araki's expansionals, see \cite{A73b}.
\begin{lemma}\label{lem: derivative vector}
    Let $(\mA, \tau)$ be a $C^*$-dynamical system, $P = P^* \in \mA_\ent^\tau$, and $s \in [0, 1]$. Furthermore, let $\omega$ be a $(\tau, \beta)$-KMS state and $\Omega_{s P}$ the vector on the GNS Hilbert space constructed in Equation \eqref{eq: perturbed vector}. Then the map 
    \begin{align*}
        s \mapsto \Omega_{s P}
    \end{align*}
    is differentiable with respect to the norm of $\mH$ and the derivative is given by
    \begin{align*}
        \frac{\d}{\d s} \Omega_{s  P} = - \beta \int_0^{1/2} \d \sigma \, \pi(\tau_{\i \sigma \beta}^{s P} (P)) \, \Omega_{s P} \, .
\end{align*}
\end{lemma}
\begin{proof}
For $z\in\C$, we define the expansional by
\begin{align*}
    E_{sP}(z)
    \coloneq\sum_{n=0}^\infty\i^n
    \int_{0\leq\sigma_n\leq\cdots\leq\sigma_1\leq z}
    \tau_{\sigma_n}(sP)\cdots\tau_{\sigma_1}(sP)
    \,\d\sigma_1\cdots\d\sigma_n.
\end{align*}
Here $\sigma_j=zt_j$ with $0\leq t_n\leq\cdots\leq t_1\leq1$, so the integrals are along the line segment from $0$ to $z$. Since $P$ is entire, $M_z\coloneq\sup_{0\leq t\leq1}\norm{\tau_{zt}(P)}$ is finite. The $n$-th summand and its $s$-derivative for $n \geq 1$ are bounded uniformly in $s\in[0,1]$ by
\[
    \frac{(|z|M_z)^n}{n!},
\]
and 
\[
    \frac{(|z|M_z)^n}{(n-1)!},
\]
respectively. Thus we may differentiate term by term.

In each differentiated product, let $u$ be the time of the differentiated factor. The factors with times below $u$ give $E_{sP}(u)$, while shifting the remaining times by $u$ gives $\tau_u(E_{sP}(z-u))$. Summing the absolutely convergent series yields
\begin{align*}
    \frac{\d}{\d s}E_{sP}(z)
    =\i\int_0^z E_{sP}(u)\tau_u(P)\tau_u(E_{sP}(z-u))\,\d u.
\end{align*}
The expansional and the perturbed dynamics satisfy
\begin{align*}
    E_{sP}(z)&=E_{sP}(u)\tau_u(E_{sP}(z-u)),\\
    \tau_u^{sP}(P)&=E_{sP}(u)\tau_u(P)E_{sP}(u)^{-1},
\end{align*}
see \cite[Proposition 5.4.1]{BR2}. Consequently,
\begin{align*}
    \frac{\d}{\d s}E_{sP}(z)
    &=\i\int_0^z\tau_u^{sP}(P)E_{sP}(z)\,\d u.
\end{align*}
Choosing $z=\i\beta/2$ and substituting $u=\i\beta\sigma$ gives
\begin{align*}
    \frac{\d}{\d s}E_{sP}\left(\frac{\i\beta}{2}\right)
    &=-\beta\int_0^{1/2}\tau_{\i\beta\sigma}^{sP}(P)
        E_{sP}\left(\frac{\i\beta}{2}\right)\,\d\sigma.
\end{align*}
By \cite[Corollary 5.4.5]{BR2}, $\Omega_{sP}=\pi(E_{sP}(\i\beta/2))\Omega$. Boundedness of $\pi$ therefore gives
\[
    \frac{\d}{\d s}\Omega_{sP}
    =-\beta\int_0^{1/2}\pi(\tau_{\i\beta\sigma}^{sP}(P))
        \Omega_{sP}\,\d\sigma. \qedhere
\]

\end{proof}

The preceding lemma gives the derivative of $s \mapsto \omega_{sP}(A)$.

\medskip
\begin{lemma} \label{lem: Derivative of state}
    Let $(\mA, \tau)$ be a $C^*$-dynamical system, $P = P^* \in \mA_\ent^\tau$, and $s \in [0, 1]$. Furthermore, let $\omega$ be a $(\tau, \beta)$-KMS state and $\Omega_{s P}$ the vector on the GNS Hilbert space constructed in Equation \eqref{eq: perturbed vector} and $\omega_{s P}$ the state constructed in Equation \eqref{eq: perturbed state}. Then the map
    \begin{align*}
        s \mapsto \omega_{sP} (A)
    \end{align*}
    is differentiable for every $A \in \mA$. Moreover, the derivative is given by
    \begin{align*}
        \frac{\d}{\d s} \omega_{sP} (A) = - \beta \int_0^1 \d \sigma \, \bigr( \omega_{sP} (A \, \tau_{\i \sigma \beta}^{s P} (P)) - \omega_{sP} (A) \, \omega_{sP} (P) \bigr) \, .
    \end{align*}
\end{lemma}
\begin{proof}
    Let $A \in \mA$. Differentiating the state in the GNS representation gives
    \begin{align*} \label{eq: derivative of KMS}
        \frac{\d}{\d s} \omega_{sP} (A) &= \frac{(\frac{\d}{\d s} \Omega_{s P}, \, \pi(A) \, \Omega_{s P} ) + (\Omega_{s P}, \, \pi(A) \, \frac{\d}{\d s} \Omega_{s P} ) }{(\Omega_{s P}, \, \Omega_{s P})} 
        \\
        & \quad- \omega_{sP} (A) \, \frac{(\frac{\d}{\d s} \Omega_{s P}, \, \Omega_{s P} ) + (\Omega_{s P}, \, \frac{\d}{\d s} \Omega_{s P} ) }{(\Omega_{s P}, \, \Omega_{s P})} \, . \numberthis
    \end{align*}
    We focus on the first term, where we use Lemma \ref{lem: derivative vector} and $P \in \mA_\ent^{\tau^{s P}}$ to calculate
    \begin{align*}
        \hspace{4ex} & \hspace{-4ex} \frac{(\frac{\d}{\d s} \Omega_{s P}, \, \pi(A) \, \Omega_{s P} ) + (\Omega_{s P}, \, \pi(A) \, \frac{\d}{\d s} \Omega_{s P} ) }{(\Omega_{s P}, \, \Omega_{s P})} 
        \\
        &= - \beta \int_0^{1/2} \d \sigma \, \bigr( \omega_{s P} ( \tau_{- \i \sigma \beta}^{s P} (P) \, A) + \omega_{s P} ( A \, \tau_{\i \sigma \beta}^{s P} (P) ) \bigr)
        \\
        & = - \beta \int_0^{1/2} \d \sigma \, \omega_{s P} ( A \,\tau_{ \i (1 - \sigma) \beta}^{s P} (P) + A \, \tau_{\i \sigma \beta}^{s P} (P) )
        \\
        & = - \beta \int_0^{1} \d \sigma \, \omega_{s P} ( A \, \tau_{\i \sigma \beta}^{s P} (P)) \, .
    \end{align*}
    In the second equality we used the KMS condition and the group property
    \begin{align*}
        \tau_z(\tau_w(P))=\tau_{z+w}(P),\qquad
        P\in\mA_\ent^{\tau^{sP}},\quad z,w\in\C.
    \end{align*} The third equality follows by a change of integration variable in the first summand. The second term in Equation \eqref{eq: derivative of KMS} can be obtained by choosing $A = 1$. Combining both statements concludes the proof.
\end{proof}

We use Lemma \ref{lem: Derivative of state} and the functional calculus to prove Equation \eqref{eq: QBP differential eqtn. prop}. By the lemma,
\begin{align*} \label{eq: deriv state proof qbp}
    \frac{\d}{\d s} \omega_{s P} (A) = - \beta \int_0^1 \d \sigma \, \bigr( \omega_{sP} (A \, \tau_{\i \sigma \beta}^{s P} (P)) - \omega_{sP} (A) \, \omega_{sP} (P) \bigr)  \, . \numberthis
\end{align*}
We first deal with the first summand, the second summand will follow by normalization. Let $(\mathcal{H}_s, \, \pi_s, \, \Omega_s)$ be the GNS triple associated to $\omega_{s P}$, and $H_{s}$ the Hamiltonian that implements the group of automorphisms $\tau^{s P}$ in this representation. Note that by stationarity of $\omega_{s P}$ the GNS vector $\Omega_s$ is an eigenvector of the GNS Hamiltonian with eigenvalue zero. Passing to the GNS representation gives
\begin{align*}
    - \beta \int_0^1 \d \sigma \, \omega_{s P} ( A \, \tau_{\i \sigma \beta}^{s P} (P) ) = - \beta  \, ( \Omega_s, \, \pi_s(A) \int_0^1 \d \sigma \, \e^{- \beta \sigma H_{s}} \, \pi_s ( P ) \, \Omega_s ) \, ,
\end{align*}
where we have used Lemma \ref{lem: impl is entire}. Inserting the identity $(1+\e^{-\beta H_s})(1+\e^{-\beta H_s})^{-1}$ gives
\begin{align*}
    \hspace{4ex} & \hspace{-4ex} - \beta \int_0^1 \d \sigma \, \omega_{s P} ( A \, \tau_{\i \sigma \beta}^{s P} (P) ) 
    \\
    &= - \beta  \, ( \Omega_s, \, \pi_s(A)\, (1 + \e^{- \beta H_{s}}) \, (1 + \e^{- \beta H_{s}})^{-1} \int_0^1 \d \sigma \, \e^{- \beta \sigma H_{s}} \, \pi_s ( P )   \, \Omega_s ) \, . \numberthis \label{eq: proof qbp}
\end{align*}
This expression is well-defined by Lemma \ref{lem: impl is entire}, since $\pi_s(P)\Omega_s$ belongs to the domain of $\e^{-\beta H_s}$. Functional calculus gives
\begin{align*}
    (1+\e^{-\beta H_s})\hat f_\beta(H_s)\pi_s(P)\Omega_s
    =2\int_0^1\e^{-\beta\sigma H_s}\pi_s(P)\Omega_s\,\d\sigma.
\end{align*}
Hence, we can rewrite the integral in Equation \eqref{eq: proof qbp} as
\begin{align*}
    - \beta \int_0^1 \d \sigma \, \omega_{s P} ( A \, \tau_{\i \sigma \beta}^{s P} (P) )& = - \frac{\beta}{2}  \, ( \Omega_s, \, \pi_s(A) \, (1 + \e^{- \beta H_{s}}) \, \hat{f}_\beta (H_s) \, \pi_s(P) \, \Omega_s) \, ,
\end{align*}
where we used the function $\hat{f}_\beta (x)$ defined in Equation \eqref{def: f_hat}. Using its Fourier transform gives
\begin{align*}
    - \beta \int_0^1 \d \sigma \, \omega_{s P} ( A \, \tau_{\i \sigma \beta}^{s P} (P) ) &= - \frac{\beta}{2} \, \int_\R \d t \, f_\beta(t) \, ( \Omega_s, \, \pi_s(A) \, (1 + \e^{- \beta H_s}) \, \e^{-\i t H_s} \, \pi_s ( P ) \, \Omega_s ) 
    \\
    &= - \frac{\beta}{2} \, \int_\R \d t \, f_\beta(t) \, ( \Omega_s, \, \pi_s(A) \, (\e^{- \i t H_s} + \e^{- (\i t + \beta) H_s}) \, \pi_s ( P ) \, \Omega_s ) \, ,
\end{align*}
The interchange with the unbounded operator is justified as follows. Since $P$ is entire, $\pi_s(P)\Omega_s$ belongs to the domain of $\e^{-\beta H_s}$, and
\begin{align*}
    &\int_\R f_\beta(t)
       \norm{(1+\e^{-\beta H_s})\e^{-\i tH_s}\pi_s(P)\Omega_s}\,\d t\\
    &\quad\leq
      \norm{\pi_s(P)\Omega_s}
       +\norm{\e^{-\beta H_s}\pi_s(P)\Omega_s}<\infty.
\end{align*}
The integral thus converges in the graph norm of the closed operator
$1+\e^{-\beta H_s}$, which permits moving this operator through the integral. Using Lemma \ref{lem: impl is entire} we can go back to the state to get
\begin{align*}
    - \beta \int_0^1 \d \sigma \, \omega_{s P} ( A \, \tau_{\i \sigma \beta}^{s P} (P) ) &= - \frac{\beta}{2}  \, \int_\R \d t \, f_\beta(t) \, \omega_{s P} ( A \, \tau_{-t}^{s P} (P) +  A  \, \tau_{- t + \i \beta}^{s P} (P) ) \, .
\end{align*}
The KMS condition gives
\begin{align*}
    - \beta \int_0^1 \d \sigma \, \omega_{s P} ( A \, \tau_{\i \sigma \beta}^{s P} (P) ) &= - \frac{\beta}{2}  \,  \int_\R \d t \, f_\beta(t) \, \omega_{s P} ( \{ \tau_{-t}^{s P} (P),  \, A \} ) 
    \\
    &= - \frac{\beta}{2} \, \omega_{s P} ( \{ \Phi_\beta^{\tau^{s P}} (P ),  \, A \} ) \, .
\end{align*}
Substitution into Equation \eqref{eq: deriv state proof qbp} gives
\begin{align*}
    \frac{\d}{\d s} \omega_{s P} (A) = - \frac{\beta}{2} \, \omega_{s P} \bigr( \bigr\{ \Phi_\beta^{\tau^{s P}} (P - \omega_{s P} (P)),  \, A \bigr\} \bigr) \, ,
\end{align*}
where we have used that $f_\beta$ is normalized. This proves the proposition for $P=P^*\in\mA_\ent^\tau$ and every $A\in\mA$.

It remains to prove the differential equation for an arbitrary perturbation $P=P^*\in\mA$. Since $\mA_\ent^\tau$ is dense in $\mA$, choose $P_n=P_n^*\in\mA_\ent^\tau$ with $P_n\to P$. It suffices to show that the derivatives of $s\mapsto\omega_{sP_n}(A)$ converge uniformly in $s\in[0,1]$ to the derivative in the statement. We use
\begin{align*}
    \omega_{s P_n} (A) &= \omega_{s P + s(P_n - P)} (A)
    \\
    & = (\omega_{s P})_{s(P_n - P)} (A) \, .
\end{align*}
By \cite[Theorem 5.4.4.(3)]{BR2},  $Q \mapsto \omega_Q$ is norm-continuous for bounded self-adjoint perturbations. This gives
\begin{align*}
    \sup_{s\in[0,1]}\norm{\omega_{sP_n}-\omega_{sP}}\longrightarrow0.
\end{align*}
Indeed, otherwise there would be a subsequence and $s_j\to s$ for which the differences stayed bounded away from zero. Since
\begin{align*}
    \norm{s_jP_{n_j}-sP}
    &\leq\norm{P_{n_j}-P}+|s_j-s|\norm P\longrightarrow0,
\end{align*}
norm continuity gives both $\omega_{s_jP_{n_j}}\to\omega_{sP}$ and $\omega_{s_jP}\to\omega_{sP}$, which is a contradiction. It remains to show that
\begin{align*}
    \Phi_\beta^{\tau^{ s P_n}} (P_n) = \int_\R \d t\, f_\beta (t) \, \tau_{-t}^{s P_n} (P_n) \, ,
\end{align*}    
    converges to $\Phi_\beta^{\tau^{s P}} (P)$ uniformly in $s \in [0,1]$. Adding and subtracting $\Phi_\beta^{\tau^{sP_n}}(P)$ gives
\begin{align*}
    \lVert \Phi_\beta^{\tau^{s P_n}} (P_n)  - \Phi_\beta^{\tau^{s P}} (P) \rVert &\leq \lVert \Phi_\beta^{\tau^{s P_n}} (P_n) - \Phi_\beta^{\tau^{s P_n}} (P) \rVert  + \lVert \Phi_\beta^{\tau^{s P_n}} (P) - \Phi_\beta^{\tau^{s P}} (P) \rVert \, .
\end{align*}
    For the first term, linearity and the normalization of $f_\beta$ give
\begin{align*}
    \lVert \Phi_\beta^{\tau^{s P_n}} (P_n) - \Phi_\beta^{\tau^{s P_n}} (P) \rVert \leq \lVert P_n - P \rVert \, , 
\end{align*}
    which converges to zero uniformly in $s \in [0, 1]$. For the second term, Duhamel's formula gives, for every $t\in\R$,
    \[
        \sup_{s\in[0,1]}
        \norm{\tau_t^{sP_n}(P)-\tau_t^{sP}(P)}
        \leq 2|t|\norm{P_n-P}\norm P.
    \]
    This yields
    \[
        \sup_{s\in[0,1]}
        \norm{\Phi_\beta^{\tau^{sP_n}}(P)-
        \Phi_\beta^{\tau^{sP}}(P)}
        \leq 2\norm{P_n-P}\norm P
        \int_\R |t|f_\beta(t)\,\d t
        \longrightarrow0.
    \]
    This proves the required convergence uniformly in $s\in[0,1]$.

Together, the preceding estimates show that
\begin{align*}
    \omega_{s P_n} \bigr( \bigr\{ \Phi_\beta^{\tau^{s P_n}} (P_n - \omega_{s P_n} (P_n)),  \, A \bigr\} \bigr)
\end{align*}
converges to
\begin{align*}
    \omega_{s P} \bigr( \bigr\{ \Phi_\beta^{\tau^{s P}} (P - \omega_{s P} (P)),  \, A \bigr\} \bigr)
\end{align*}
uniformly in $s \in [0, 1]$. Therefore, the map $s \mapsto \omega_{s P} (A)$ is differentiable with derivative
\begin{align*}
    \frac{\d}{\d s} \omega_{s P} (A) = - \frac{\beta}{2} \omega_{s P} \bigr( \bigr\{ \Phi_\beta^{\tau^{s P}} (P - \omega_{s P} (P)),  \, A \bigr\} \bigr) \, ,
\end{align*}
which concludes the proof of the proposition.
\end{proof}

\section{Proofs of the locality results} \label{sec: proof fermions}

\begin{proof}[Proof of Theorem \ref{thm: LPPL}]
    We set $F(r)=(1+r)^{-\nu}$ in the polynomial regime and $F(r)=\exp(-ar)$ in the exponential regime.
    Let $A \in \mA_Y$ be arbitrary. We start by writing the difference of expectation values as an integral
    \begin{align*}
        |\omega_P (A) - \omega(A)| = \bigl| \int_0^1 \d s \, \frac{\d}{\d s} \omega_{s P} (A) \bigr| \, .
    \end{align*}
    We use the triangle inequality to move the absolute value inside the integral and Proposition \ref{prop: QBP differential equation} to evaluate the derivative
    \begin{align*}
        |\omega_P (A) - \omega(A) | &\leq \int_0^1 \d s \, \bigl| \frac{\d}{\d s} \omega_{s P} (A) \bigr|
        \\
        & = \frac{|\beta|}{2} \int_0^1 \d s \, \bigr| \omega_{s P} \bigr( \bigr\{ \Phi_\beta^{\tau^{s P}} (P) - \omega_{s P} \bigl( \Phi_\beta^{\tau^{s P}} (P) \bigr),  \, A \bigr\} \bigr) \bigr| \, ,
    \end{align*}
    where $\tau^{s P}$ is the dynamics generated by $\mL_{H} + \mL_{s P}$. By Lemma \ref{lem: dynamics is even} $\Phi_\beta^{\tau^{sP}}(P)$ is even. Hence, we can apply decay of correlations in the form of Lemma \ref{lem: decay of correlations for A_infty} and estimate
    \begin{align*} 
        |\omega_P (A) - \omega(A) | &\leq C \, |\beta| \, \lVert A \rVert \, \sup_{s \in [0, 1]} \, \bigl\lVert \Phi_\beta^{\tau^{s P}} (P) \bigr\rVert_{F, x} \, f_\cor(|Y|) \, (1 + d (x, Y))^{D \, n_\cor} \zeta_\cor( d(x, Y) / 2 ) \, . 
    \end{align*}
    Bounding the decay norm of $\Phi_\beta^{\tau^{s P}} (P)$ using Proposition \ref{prop: loc of phi} yields the claim.
\end{proof}

\medskip
\begin{proof}[Proof of Theorem \ref{thm: LPPL without doc}]
    We set $F(r)=(1+r)^{-\nu}$ in the polynomial regime and $F(r)=\exp(-ar)$ in the exponential regime.
    Let $A \in \mA_Y$ be arbitrary. We use Theorem \ref{thm: QBP} to find an operator $\eta_P \in \mA$. This operator is even, as can be seen by Equation \eqref{eq: eta}. We write
    \begin{align*}
        \omega_{V_X + P} (A) = \omega_{V_X}(\eta_P \, A \, \eta_P^*) \, .
    \end{align*}
    We use this expression to estimate the difference
    \begin{align*} \label{eq: difference LPPL}
        |\omega_{V_X + P} (A) - \omega_{V_X} (A) | &= |\omega_{V_X} (\eta_P \, A \, \eta_P^*) - \omega_{V_X} (A)|
        \\
        & \leq \bigl| \omega_{V_X} ([\eta_P, A] \, \eta_P^*) \bigr| + \bigl| \omega_{V_X} (A \, \eta_P \eta_P^*) - \omega_{V_X} (A) \, \omega_{V_X} (\eta_P \eta_P^*) \bigr| \, , \numberthis
    \end{align*}
    where we have inserted a zero and used the triangle inequality. Using that $\omega_{V_X}$ is a state we bound the first term by
    \begin{align*}
        |\omega_{V_X} ([\eta_P, A] \, \eta_P^*)| &\leq \lVert [\eta_P, \, A] \rVert \, \lVert \eta_P^* \rVert = \lVert [\eta_P - 1, A] \rVert \, \lVert \eta_P^* \rVert \, .
    \end{align*}
    Using Lemma \ref{lem:commutator bound} we can estimate the norm of the commutator by
    \begin{align*}
        \lVert [\eta_P, A] \rVert \leq \frac{2}{F (1)} \, F(d (x, Y)) \, \lVert \eta_P - 1 \rVert_{F, x} \, \lVert A \rVert \, ,
    \end{align*}
    for any $x \in \Gamma$. Furthermore, we use the trivial bound
    \begin{align*}
        \lVert \eta_P^* \rVert = \lVert \eta_P \rVert \leq \lVert \eta_P \rVert_{F, x} \, .
    \end{align*}
    We turn to the second term in Equation \eqref{eq: difference LPPL}. Using Lemma \ref{lem: dynamics is even} we can see that $\eta_P$ is even. Hence, we can apply our assumptions on decay of correlations in the form of Lemma \ref{lem: decay of correlations for A_infty} of the initial state and find
    \begin{align*}
        \hspace{4ex} & \hspace{-4ex} |\omega_{V_X} (A \, \eta_P \eta_P^*) - \omega_{V_X} (A) \, \omega_{V_X}(\eta_P \eta_P^*)| = |\omega_{V_X} (A \, (\eta_P \eta_P^* - 1)) - \omega_{V_X} (A) \, \omega_{V_X}(\eta_P \eta_P^* - 1)|
        \\
        &\leq C \, \lVert A \rVert \, \lVert \eta_P \eta_P^* - 1 \rVert_{F,x} \, f_\cor (|Y|) \, (1 + d (x, Y))^{D \, n_\cor} \zeta_\cor(d (x, Y) / 2) \, ,
    \end{align*}
    for some constant $C > 0$. Combining both bounds yields
    \begin{align*}
        &|\omega_{V_X + P} (A) - \omega_{V_X} (A) |
        \\
        & \quad\leq C^\prime \, \lVert A \rVert \, \bigl( \lVert \eta_P - 1 \rVert_{F, x} \, \lVert \eta_P \rVert_{F, x} + \lVert \eta_P \eta_P^* - 1 \rVert_{F, x} \bigr) \, f_\cor(|Y|) \, (1 + d (x, Y))^{D \, n_\cor} \zeta_\cor(d (x, Y) / 2) \, .
    \end{align*}
    We find that the generator is bounded by
    \begin{align*}
        \norm{\Phi_\beta^{\tau^{V_X+sP}}(P)-\omega_{V_X+sP}(P)}_{F,x}\leq 2\norm{\Phi_\beta^{\tau^{V_X+sP}}(P)}_{F,x}.
    \end{align*}
    By definition of $\eta_P$ and Lemma \ref{lem: continuity decay norm} we find that 
    \begin{align*}
        \norm{\eta_P-1}_{F,x} \leq \exp \bigl( 2 |\beta|\sup_{s\in[0,1]} \norm{\Phi_\beta^{\tau^{V_X+sP}}(P)}_{F,x} \bigr) - 1.
    \end{align*}
    Using $\eta_P \eta_P^* - 1 = (\eta_P - 1) \eta_P^*+ (\eta_P^* - 1)$ and again Lemma \ref{lem: continuity decay norm} we find that
    \begin{align*}
        \norm{\eta_P\eta_P^*-1}_{F,x}
        &\leq3 \bigl( \exp\bigl(4|\beta|\sup_{s\in[0,1]} \norm{\Phi_\beta^{\tau^{V_X+sP}}(P)}_{F,x}\bigr)-1 \bigr).
    \end{align*}
    Inserting these bounds and applying Proposition \ref{prop: loc of phi} proves the statement.
\end{proof}

\medskip
\begin{proof}[Proof of Theorem \ref{thm: li from lppl}]
    Let $X \subseteq Y \subsetneq Z\in P_0 (\Gamma)$ and $A \in \mA_X$ be arbitrary. We first show the bound for
    \begin{align*}
        |\omega_{V_Z} (A) - \omega_{V_Y} (A)| \, .
    \end{align*}
    We remove the sites of $Z\setminus Y$ one at a time and use LPPL in each of the states before a site is removed.
    
    To formalize this idea, we set $N \coloneq |Z \setminus Y|$ and label the sites in $Z \setminus Y$ by $z_i$ with $i \in \{1, \dots, N\}$. We define $V_Y^i \coloneq V_{Y \cup \{z_1, \ldots, z_i\}}$ for $i=0,\ldots,N$, with $V_Y^0=V_Y$. A telescoping series gives us that 
    \begin{align*} \label{eq: sum LI}
        |\omega_{V_Z} (A) - \omega_{V_Y} (A)| \leq \sum_{i = 1}^N |\omega_{V_Y^i} (A) - \omega_{V_Y^{i - 1}} (A)| \, . \numberthis
    \end{align*}
    Composition of bounded perturbations gives
    \begin{align*}
        \omega_{V_Y^{i-1}}
        =(\omega_{V_Y^i})_{V_Y^{i-1}-V_Y^i}.
    \end{align*}
    Applying LPPL to the initial state $\omega_{V_Y^i}$ therefore gives
    \begin{align*}
        |\omega_{V_Y^i} (A) - \omega_{V_Y^{i - 1}} (A)| \leq \lVert A \rVert \, f_\lppl (|X|) \, g_\lppl(\lVert V_Y^i - V_Y^{i - 1} \rVert_{F, z_i} ) \, \zeta_\lppl (d (z_i, X) ) \, .
    \end{align*}
    We bound the decay norm of the interaction by
    \begin{align*}
        \lVert V_Y^i - V_Y^{i - 1} \rVert_{F, z_i} \leq \sum_{\substack{\Lambda \subseteq Y \cup \{z_1, \ldots, z_i \} \\ \Lambda \nsubseteq Y \cup \{z_1, \ldots, z_{i - 1} \}}} \lVert V(\Lambda) \rVert_{F, z_i} \, .
    \end{align*}
    The bound 
    \begin{align*}
        \lVert V(\Lambda) \rVert_{F, z_i} \leq 3 \frac{\lVert V (\Lambda) \rVert}{F(\diam (\Lambda))}
    \end{align*}
    then yields the bound 
    \begin{align*}
        \lVert V_Y^i - V_Y^{i - 1} \rVert_{F, z_i} \leq 3 \lVert V \rVert_{F} \, ,
    \end{align*}
    and hence
    \begin{align*}
        |\omega_{V_Y^i} (A) - \omega_{V_Y^{i - 1}} (A)| \leq \lVert A \rVert \, f_\lppl (|X|) \, g_\lppl(3 \lVert V \rVert_{F} ) \, \zeta_\lppl (d (z_i, X) ) \, .
    \end{align*}
    Inserting this bound into Equation \eqref{eq: sum LI} gives
    \begin{align*}
        |\omega_{V_Z}(A)-\omega_{V_Y}(A)|
        &\leq\norm A f_\lppl(|X|)g_\lppl(3\norm V_F)
                 \sum_{z\in Z\setminus Y}\zeta_\lppl(d(z,X))
        \\
        & \leq C_\sur \, |X| \, \norm A f_\lppl(|X|)g_\lppl(3\norm V_F) \, h(d(X, Z \setminus Y)).
    \end{align*}
    The first bound uses LPPL only in the states $\omega_{V_Y^i}$, $i=1,\ldots,N$; these form the stated sequence when read in reverse order. The second bound uses the summability assumption. For the uniform assertion, apply these estimates to $sV$ and use
    \[
        g_\lppl(3\norm{sV}_F)
        =g_\lppl(3s\norm V_F)
        \leq g_\lppl(3\norm V_F). \qedhere
    \]
\end{proof}

\medskip
\begin{proof}[Proof of Theorem \ref{thm: doc from LI}]
    Let $X$, $Y \in P_0 (\Gamma)$ and $A \in \mA_X$, $B\in \mA_Y$ be arbitrary. We fix $r = d(X, Y) / 3$ for the rest of the proof. If $r\leq1$, then $\zeta_\li(r)\geq\zeta_\li(1)\geq2^DF(1)>0$, and
    \begin{align*}
        |\omega_H(AB)-\omega_H(A)\omega_H(B)|
        \leq\frac{2}{\zeta_\li(1)}
        \bigl(f_\li(|X|+|Y|)+|X|\bigr)\zeta_\li(r)\norm A\norm B.
    \end{align*}
    The same estimate holds for $\omega_{H_M}$ for every $M\in P_0(\Gamma)$. We therefore assume $r>1$, so the enlarged supports are disjoint. We first prove the estimate for the finite-patch states. Fix $M\in P_0(\Gamma)$ with $X\cup Y\subseteq M$. Throughout this estimate, $X_r$ and $Y_r$ denote the fattenings intersected with $M$. By LI we may restrict the perturbation to $X_r\cup Y_r$. We then remove the terms connecting the two fattenings; the remaining state has zero covariance by the product property of the tracial state.

    If $X_r\cup Y_r=X\cup Y$, then $d(X,M\setminus X)\geq r$. Corollary \ref{cor: cont. KMS state} and the product property of the tracial state (see \cite[Proposition 4.6]{AM03}) give
    \begin{align*}
        |\omega_{H_M}(AB)-\omega_{H_M}(A)\omega_{H_M}(B)|
        &\leq 3|\beta|\norm A\norm B
            \norm{H_M-H_X-H_{M\setminus X}}\\
        &\leq 3|\beta|\norm A\norm B
            \sum_{x\in X}\sum_{\substack{x\in Z\subseteq M\\\diam Z\geq r}}\norm{H(Z)}\\
        &\leq 3|\beta|\,|X|F(r)\norm H_F\norm A\norm B.
    \end{align*}
    This proves the bound in this case, since $F(r)\leq\zeta_\li(r)$. We may therefore assume that $X\cup Y\subsetneq X_r\cup Y_r$.
    If $X_r\cup Y_r=M$, the following difference vanishes; otherwise, LI gives
    \begin{align*}
        \hspace{4ex}&\hspace{-4ex} |\omega_{H_M} (A B) - \omega_{H_{X_r \cup Y_r}} (A B) - \bigl( \omega_{H_M} (A) \, \omega_{H_M} (B) - \omega_{H_{X_r \cup Y_r}} (A) \, \omega_{H_{X_r \cup Y_r}} (B) \bigr) | 
        \\
        &\leq 3 \, \lVert A \rVert \, \lVert B \rVert \, f_\li (|X| + |Y|) \, g_\li (\lVert H \rVert_F) \, \zeta_\li (r) \, . 
    \end{align*}
    We use Corollary \ref{cor: cont. KMS state} and find
    \begin{gather*}
        | \omega_{H_{X_r \cup Y_r}} (A B) - \omega_{H_{X_r} + H_{Y_r}} (A B) - \bigl( \omega_{H_{X_r \cup Y_r}} (A) \, \omega_{H_{X_r \cup Y_r}} (B) - \omega_{H_{X_r} + H_{Y_r}} (A) \, \omega_{H_{X_r} + H_{Y_r}} (B) \bigr) | 
        \\
        \leq 3 |\beta| \, \lVert A \rVert \, \lVert B \rVert \, \lVert \sum_{ \substack{ Z\subseteq X_r \cup Y_r \\ Z \nsubseteq X_r, \, Z \nsubseteq Y_r}} H(Z) \rVert \,  .  \numberthis \label{eq: proof doc interaction norm}
    \end{gather*}
    Since $d(X_r,Y_r)\geq d(X,Y)-2r=r$, every interaction term in the sum meets both fattenings and has diameter at least $r$. Hence,
    \begin{align*}
        \norm{\sum_{\substack{Z\subseteq X_r\cup Y_r\\
        Z\nsubseteq X_r,\;Z\nsubseteq Y_r}}H(Z)}
        &\leq\sum_{x\in X_r}
        \sum_{\substack{x\in Z\in P_0(\Gamma)\\\diam Z\geq r}}\norm{H(Z)}\\
        &\leq |X_r|F(r)\norm H_F.
    \end{align*}
    Inserting this estimate into Equation \eqref{eq: proof doc interaction norm} yields
    \begin{gather*}
        | \omega_{H_{X_r \cup Y_r}} (A B) - \omega_{H_{X_r} + H_{Y_r}} (A B) - \bigl( \omega_{H_{X_r \cup Y_r}} (A) \, \omega_{H_{X_r \cup Y_r}} (B) - \omega_{H_{X_r} + H_{Y_r}} (A) \, \omega_{H_{X_r} + H_{Y_r}} (B) \bigr) | 
        \\
        \leq 3 |\beta| \, F(r) \,C_\vol  (1 + r)^D |X| \, \lVert A \rVert \, \lVert B \rVert \, \lVert  H \rVert_F  \, .
    \end{gather*}
    Finally, we notice that by the product property of the tracial state
    \begin{align*}
        |\omega_{H_{X_r} + H_{Y_r}} (A B) -  \omega_{H_{X_r} + H_{Y_r}} (A) \, \omega_{H_{X_r} + H_{Y_r}} (B)| = 0 \, .
    \end{align*}
    Combining the preceding estimates gives
    \begin{align*}
        |\omega_{H_M} (A B) - \omega_{H_M} (A) \, \omega_{H_M} (B)| \leq  C \, (1 + |\beta|) \, \bigl( f_\li(|X| + |Y|) + |X| \bigr) \, \zeta_\li (r) \, \lVert A \rVert \, \lVert B \rVert \,  ,
    \end{align*}
    for some constant $C > 0$ independent of $M$. This proves the bound whenever $X\cup Y\subseteq M$.

    Let now $X$ and $Y$ be disjoint but not necessarily contained in $M$. Since $H_M\in\mA_M$, the defining property of the conditional expectation gives $\omega_{H_M}=\omega_{H_M}\circ\E_M$. Moreover, the product property of the tracial state gives $\E_M(AB)=\E_M(A)\E_M(B)$. Therefore,
    \begin{align*}
        &\omega_{H_M}(AB)-\omega_{H_M}(A)\omega_{H_M}(B) =\omega_{H_M}(\E_M(A)\E_M(B))
            -\omega_{H_M}(\E_M(A))\omega_{H_M}(\E_M(B)).
    \end{align*}
    If one of the intersections $X\cap M$ and $Y\cap M$ is empty, the right-hand side vanishes. Otherwise, $\E_M(A)\in\mA_{X\cap M}$ and $\E_M(B)\in\mA_{Y\cap M}$, and contractivity of $\E_M$ gives the same bound. Overlapping supports are covered by the trivial covariance bound. Finally, choosing $M=\Lambda_n$ for an IAS and using Proposition \ref{prop: convergence from li} proves the bound for $\omega_H$.
\end{proof}

\medskip
\begin{proof}[Proof of Theorem \ref{thm: doc from LI + doc}]
    Let $X$, $Y\in P_0(\Gamma)$ and $A\in\mA_X$, $B\in\mA_Y$ be arbitrary, with at least one of $A$, $B$ even. For $\beta=0$ the state and all its perturbations are the tracial state. Hence, the statement becomes trivial. Assume $\beta \neq 0$.

    The proof proceeds through the following approximations:
    \begin{align*}
        \omega_V
        \xrightarrow{\mathrm{LI}}\omega_{V_{X_R\cup Y_R}}
        \xrightarrow{\text{remove connecting terms}}\omega_{V_{X_R}+V_{Y_R}}
        \xrightarrow{\text{QBP}}\phi_L.
    \end{align*}
    Each arrow means that the covariance is replaced by the covariance in the state on its right, with an error estimated below. The first step uses LI, the second uses exponential locality of $V$, and the third uses the QBP operators localized in a region $X_{R + L} \cup Y_{R+ L}$. The covariance in $\phi_L$ is then written as covariances of the initial state $\omega$ and estimated by clustering.

    The parameters $R$ and $L$ balance these errors. Increasing $R$ improves the LI estimate and increasing $L$ improves the QBP localization estimates. On the other hand, increasing them also increases the QBP norms and the supports entering the clustering estimate. We therefore choose them so that the relevant volumes grow sublinearly in $d(X, Y)$. This gives the stretched-exponential bound in the theorem. Decay of correlations that is only polynomial in the distance would not suffice for this argument. For a tracial reference state the final QBP step is unnecessary and a radius proportional to the separation can be used. Thus the loss is a limitation of the estimates and does not rule out exponential clustering of the perturbed state.

    All constants and lower bounds on $d(X,Y)$ used below are independent of $X$, $Y$, $A$ and $B$. They depend only on the lattice, $c$, $\beta$, the interaction norms and the stated clustering and LI data. We specify more restricted dependencies where they are useful. The bound $|\omega_V(AB)-\omega_V(A)\omega_V(B)|\leq2\norm A\norm B$ includes bounded distances after increasing the constant in the theorem, so we may restrict to sufficiently large $d(X,Y)$.

\subsubsection*{Restricting the perturbation to two separated regions}
    
    Fix $\vartheta \coloneq 1 / (D + 2)$. We divide the proof into two cases and first consider the case $1 + |X| + |Y| \leq(1+d (X, Y))^\vartheta$. The complementary case will follow from the elementary covariance bound at the end of the proof.
    We define the fattening parameter
    \[
        R\coloneq \left\lfloor \left(\frac{(1+d(X, Y))^\vartheta}{1 + |X| + |Y|}\right)^{1/D} \right\rfloor.
    \]
    The denominator in this choice compensates for the sizes of the original supports. Since $R\geq1$, the volume bound gives, with $C_0=2^DC_\vol$ depending only on the lattice,
    \begin{align}
        |X_R|+|Y_R| \leq C_0 (1+d (X, Y))^\vartheta. \label{eq: proof generalized doc size U W}
    \end{align}
    Furthermore, $d(X_R,Y_R) \geq d(X,Y)-2R$. Since $R \leq (1+d(X, Y))^{\vartheta/D}$ and $\vartheta/D<1$, it holds that
    \begin{align}
        d(X_R, Y_R)\geq \frac {d(X, Y)}{2}
        \label{eq: proof generalized doc distance U W}
    \end{align}
    for all sufficiently large distances. 

    We next use LI to replace the limiting perturbation by the perturbation supported in $X_R\cup Y_R$. Let $(\Lambda_j)_{j\in\N}$ be an IAS and set $V_j\coloneq V_{\Lambda_j}$. For all sufficiently large $j$, it holds that $X_R \cup Y_R \subseteq\Lambda_j$ and that $d(X\cup Y,\Lambda_j\setminus(X_R\cup Y_R)) \geq R$. LI applied to $AB\in\mA_{X\cup Y}$ gives
    \begin{align*}
        |\omega_{V_j}(AB)-\omega_{V_{X_R\cup Y_R}}(AB)|
        \leq
        \lVert A\rVert\lVert B\rVert
        f_\li(|X\cup Y|)
        g_\li(\lVert V\rVert_{\exp(-c\cdot)})
        \zeta_\li(R).
    \end{align*}
    Applying LI also to $A$ and $B$ and using that states have norm one gives
    \begin{align*}
        &|\omega_{V_j}(A)\omega_{V_j}(B)
        -\omega_{V_{X_R\cup Y_R}}(A)\omega_{V_{X_R \cup Y_R}}(B)|\\
        & \quad \leq
        2\lVert A\rVert\lVert B\rVert
        f_\li(|X| + |Y|)
        g_\li(\lVert V\rVert_{\exp(-c\cdot)})
        \zeta_\li(R).
    \end{align*}
    Combining the preceding estimates, taking $j\to\infty$, using Proposition \ref{prop: convergence from li} and the assumptions on $f_\li$ and $\zeta_\li$, we obtain
    \begin{align*}
        &\big|
        \omega_V(AB)-\omega_V(A)\omega_V(B)
        -\omega_{V_{X_R\cup Y_R}}(AB)
        +\omega_{V_{X_R\cup Y_R}}(A)
         \omega_{V_{X_R\cup Y_R}}(B)
        \big|
        \\
        & \qquad \leq
        C_1 \lVert A\rVert\lVert B\rVert
            \exp\bigl(C_\li(1+|X|+|Y|)\bigr) \e^{-\mu_\li R}, \numberthis \label{eq: proof generalized doc LI}
    \end{align*}
    where $C_1=3C_\li^2g_\li(\norm V_{\exp(-c\cdot)})$ depends only on the LI data and the interaction norm. In the next step, we remove the terms connecting $X_R$ and $Y_R$.
    \begin{align*}
        V_{X_R\cup Y_R} - V_{X_R} - V_{Y_R} = \sum_{\substack{Z\subseteq (X_R\cup Y_R)\\
        Z\nsubseteq X_R,\;Z\nsubseteq Y_R}}V(Z).
    \end{align*}
    Every set in this sum meets both $X_R$ and $Y_R$, and therefore has diameter at least $d(X_R,Y_R)$. Consequently,
    \begin{align*}
        \lVert V_{X_R\cup Y_R}-V_{X_R}-V_{Y_R} \rVert
        &\leq
        \sum_{x\in X_R}
        \sum_{\substack{x\in Z\in P_0(\Gamma)\\
        \diam(Z)\geq d(X_R, Y_R)}}
        \lVert V(Z)\rVert\\
        &\leq
        |X_R|\e^{-c\,d(X_R,Y_R)}
        \lVert V\rVert_{\exp(-c\cdot)}.
    \end{align*}
    Equations \eqref{eq: proof generalized doc size U W} and \eqref{eq: proof generalized doc distance U W} yield
    \begin{align*}
        \lVert V_{X_R\cup Y_R}-V_{X_R}-V_{Y_R}\rVert
        \leq C_0 \, \lVert V \rVert_{\exp(- c \cdot)} \,  (1 + d (X, Y))^\vartheta\e^{- c d (X, Y)/2},
    \end{align*}
    Applying Corollary \ref{cor: cont. KMS state} to $AB$, $A$ and $B$ gives
    \begin{align}
        &\big|
        \omega_{V_{X_R \cup Y_R}}(AB)
        -\omega_{V_{X_R \cup Y_R}}(A)\omega_{V_{X_R \cup Y_R}}(B)
        -\omega_{V_{X_R} + V_{Y_R}} (AB)
        +\omega_{V_{X_R} + V_{Y_R}}(A) \omega_{V_{X_R} + V_{Y_R}} (B)
        \big|\nonumber\\
        & \quad \leq
        3|\beta|\lVert A\rVert \, \lVert B\rVert
        \lVert V_{X_R \cup Y_R} - V_{X_R} - V_{Y_R} \rVert\nonumber\\
        & \quad \leq 3 |\beta| \, C_0 \, \lVert V \rVert_{\exp(- c \cdot)} \, \lVert A \rVert \, \lVert B \rVert \, (1+d (X, Y))^\vartheta\e^{- c\, d(X, Y)/2}.
        \label{eq: proof generalized doc remove crossing}
    \end{align}

\subsubsection*{Localizing the two QBP operators}

    It remains to estimate the covariance in $\omega_{V_{X_R}+V_{Y_R}}$. We first express this state in terms of $\omega$. Let $\eta_{X_R}$ be the QBP operator which transports $\omega$ to $\omega_{V_{X_R}}$, and let $\eta_{Y_R}^{X_R}$ transport $\omega_{V_{X_R}}$ to $\omega_{V_{X_R}+V_{Y_R}}$. Applying Theorem \ref{thm: QBP} twice gives
    \begin{align*}
        \omega_{V_{X_R}+V_{Y_R}}(D)
        =\omega\bigl(\eta_{X_R}\eta_{Y_R}^{X_R}D(\eta_{Y_R}^{X_R})^*\eta_{X_R}^*\bigr),
        \qquad D\in\mA_\Gamma.
    \end{align*}
    The second operator depends on the first perturbation, but its generator is localized near $Y_R$. We will approximate the two operators separately, by localizing their generators before taking the Dyson series. This also gives control of the norms of the local approximations.

    We first choose the decay exponents so that Proposition \ref{prop: loc of phi} applies at the fixed inverse temperature $\beta$. Let $\kappa_{c,c/2}$ be the constant from Proposition \ref{prop: LiebRobinson}, which depends only on the lattice and $c$. By increasing this constant if necessary, we assume that $\kappa_{c,c/2}\geq1$. We choose $b=c/2$ and
    \begin{align*}
        0<a&<\min \bigl(
        \frac c{16},
        \frac{\pi c}{64\kappa_{c,c/2}(1+|\beta|)
        \bigl(1+\norm H_{\exp(-c\cdot)}+2\norm V_{\exp(-c\cdot)}\bigr)}
        \bigr).
    \end{align*}
    Thus $a$ and $b$ are fixed independently of the supports. They satisfy $b>8a$ and $c>16a$. The interaction bound $C_\inter(a,c)=\frac{\pi}{2\kappa_{c,c/2}}(\frac c{8a}-1)$ from Proposition \ref{prop: loc of phi} satisfies
    \begin{align*}
        \frac{C_\inter(a,c)}{1+|\beta|}
        >2\bigl(1+\norm H_{\exp(-c\cdot)}+2\norm V_{\exp(-c\cdot)}\bigr).
    \end{align*}
    Both interactions $H+sV_{X_R}$ and $H+V_{X_R}+sV_{Y_R}$ therefore satisfy this bound uniformly in $s$, $X$, $Y$ and $R$.

    For $Z\in P_0(\Gamma)$ and $x\in Z$, the definition of the decay norm gives
    \begin{align*}
        \lVert V(Z)\rVert_{\exp(-b\cdot),x}
        \leq3\e^{b\diam(Z)}\lVert V(Z)\rVert.
    \end{align*}
    Proposition \ref{prop: loc of phi} consequently gives
    \begin{align*}
        \lVert \Phi_\beta^{ \tau^{sV_{X_R}}} (V(Z)) \rVert_{\exp( - a \cdot), x}
        &\leq C_\Phi(1+|\beta|)\norm{V(Z)}_{\exp(-b\cdot),x}\\
        &\leq 3C_\Phi(1+|\beta|)\e^{b\diam(Z)}\lVert V(Z)\rVert,
    \end{align*}
    where $C_\Phi>0$ depends only on the lattice and the fixed exponents $a,b,c$, and is uniform over the interactions satisfying the displayed norm bound. In particular, it is independent of $s$ and all supports and localization radii. Let $L\in\N$. For $x\in Z\subseteq X_R$, we have $B_L(x)\subseteq X_{R+L}$. Hence,
    \begin{align*}
        \bigl\lVert (1 - \E_{X_{R + L}}) \Phi_\beta^{\tau^{s V_{X_R}}} (V(Z)) \bigr\rVert
        &\leq 2 \bigl\lVert (1 - \E_{B_L(x)}) \Phi_\beta^{\tau^{ s V_{X_R} } } (V(Z)) \bigr\rVert
        \\
        &\leq 6C_\Phi(1+|\beta|) \, \e^{-a  L} \, \e^{b \, \diam(Z)} \lVert V(Z)\rVert.
    \end{align*}
    The scalar centering term is annihilated by $1-\E_{X_{R+L}}$. Summing the last estimate over $Z\subseteq X_R$ and using $b<c$ therefore gives
    \begin{align*}
        &\bigl \lVert (1 - \E_{X_{R + L}}) \Phi_\beta^{ \tau^{s V_{X_R}}} (V_{X_R} - \omega_{s V_{X_R}} (V_{X_R})) \bigr\rVert
        \\
        &\quad \leq 6C_\Phi(1+|\beta|) \, \e^{-aL} \sum_{x \in X_R} \sum_{x \in Z \in P_0(\Gamma)} \e^{b \, \diam(Z)} \, \lVert V(Z) \rVert
        \\
        &\quad \leq 6C_\Phi(1+|\beta|) \, |X_R| \, \e^{-aL} \, \lVert V \rVert_{\exp(-c\cdot)}, \numberthis \label{eq: proof generalized doc generator U}
    \end{align*}
    The uniform interaction bound also applies to the second QBP operator. Replacing $X_R$ by $Y_R$ and using the dynamics generated by $H+V_{X_R}+sV_{Y_R}$ yields
    \begin{align}
        \bigl\lVert (1 - \E_{Y_{R + L}}) \Phi_\beta^{\tau^{ V_{X_R} + s V_{Y_R}}} (V_{Y_R}) \bigr\rVert \leq 6C_\Phi(1+|\beta|) \, |Y_R| \, \e^{-aL} \, \lVert V \rVert_{\exp(- c \cdot)} \, , \label{eq: proof generalized doc generator W}
    \end{align}
    We now define the local approximations. Let $\eta_{X_R}^L$ be the Dyson series in Equation \eqref{eq: eta} with the QBP generator replaced by
    \begin{align*}
        \E_{X_{R + L}} \Phi_\beta^{\tau^{\sigma V_{X_R}}} \bigl(V_{X_R}-\omega_{\sigma V_{X_R}}(V_{X_R})\bigr),
    \end{align*}
    and define $(\eta_{Y_R}^{X_R})^L$ by the corresponding replacement with
    \begin{align*}
        \E_{Y_{R + L}} \Phi_\beta^{\tau^{V_{X_R}+\sigma V_{Y_R}}} \bigl(V_{Y_R}-\omega_{V_{X_R}+\sigma V_{Y_R}}(V_{Y_R})\bigr).
    \end{align*}
    By construction, $\eta_{X_R}^L\in\mA_{X_{R+L}}$ and $(\eta_{Y_R}^{X_R})^L\in\mA_{Y_{R+L}}$. To estimate their errors and norms, we use contractivity of $\Phi_\beta$ and of the conditional expectations. In particular,
    \begin{align*}
        \norm{\Phi_\beta^{\tau^{sV_{X_R}}}(V_{X_R})}
        &\leq\norm{V_{X_R}}\leq |X_R|\norm V_{\exp(-c\cdot)},\\
        \norm{\Phi_\beta^{\tau^{V_{X_R}+sV_{Y_R}}}(V_{Y_R})}
        &\leq\norm{V_{Y_R}}\leq |Y_R|\norm V_{\exp(-c\cdot)}.
    \end{align*}
    Centering increases these bounds by at most a factor of two. In each term of the difference of the Dyson series, we replace one generator at a time using \cite[Equation 34]{CMTW25} and use Equations \eqref{eq: proof generalized doc generator U} and \eqref{eq: proof generalized doc generator W}. Summing the series gives
    \begin{align*}
        \lVert\eta_{X_R}-\eta_{X_R}^L\rVert
        &\leq h^\prime (|\beta|) \, |X_R| \, \lVert V \rVert_{\exp(-c \cdot)} \, \exp (2 |\beta| \, |X_R| \, \lVert V \rVert_{\exp( - c \cdot)} - a L ) ,\nonumber\\
        \lVert\eta_{Y_R}^{X_R}-(\eta_{Y_R}^{X_R})^L\rVert
        &\leq h^\prime (|\beta|) \, |Y_R| \, \lVert V \rVert_{\exp(-c \cdot)} \, \exp (2 |\beta| \, |Y_R| \, \lVert V \rVert_{\exp( - c \cdot)} - a L) ,
    \end{align*}
    where $h^\prime(|\beta|)\coloneq3C_\Phi|\beta|(1+|\beta|)$. The same generator norm bounds give
    \begin{align*}
        \lVert \eta_{X_R} \rVert, \, \lVert \eta_{X_R}^L \rVert &\leq \exp \bigl( 2 |\beta| \, |X_R| \, \lVert V \rVert_{\exp( - c \cdot)} \bigr)
        \\
        \lVert \eta_{Y_R}^{X_R} \rVert , \, \lVert (\eta_{Y_R}^{X_R})^L \rVert &\leq \exp \bigl(2 |\beta| \, |Y_R| \, \lVert V \rVert_{\exp( - c \cdot)} \bigr) \numberthis \label{eq: proof generalized doc qbp norms}
    \end{align*}
    We now estimate the effect on expectations of replacing both QBP operators by their local approximations. Using the error bound for the replaced factor and the norm bounds for the remaining factors gives
    \begin{align}
        &\big| \omega \bigl( \eta_{X_R} \eta_{Y_R}^{X_R} \, D \, (\eta_{Y_R}^{X_R})^* \eta_{X_R}^* \bigr) - \omega \bigl( \eta_{X_R}^L (\eta_{Y_R}^{X_R})^L \, D \, ((\eta_{Y_R}^{X_R})^L)^* (\eta_{X_R}^L)^* \bigr) \big| \nonumber
        \\
        &\quad \leq 2 h^\prime (|\beta|) \, (|X_R| + |Y_R| )  \, \lVert V \rVert_{\exp(- c \cdot)} \, \exp \bigl( 6 |\beta| \, (|X_R| + |Y_R|) \, \lVert V \rVert_{\exp(-c \cdot)} - a L \bigr) \, \lVert D \rVert,
        \label{eq: proof generalized doc functional approximation}
    \end{align}
    for every $D\in\mA_\Gamma$. This estimate determines the required width $L$. Choose $K>0$, depending only on $a$, the lattice and $\norm V_{\exp(-c\cdot)}$, such that $aK>6C_0\norm V_{\exp(-c\cdot)}$, and set
    \begin{align*}
        L\coloneq\left\lceil|\beta|K(1+d(X,Y))^\vartheta\right\rceil.
    \end{align*}
    By Equation \eqref{eq: proof generalized doc size U W}, the exponent in the error bound is then negative with magnitude proportional to $(1+d(X,Y))^\vartheta$. Absorbing the remaining factor $|X_R|+|Y_R|$ into half of this decay gives
    \begin{align*}
        &2 h^\prime (|\beta|)(|X_R|+|Y_R|)\norm V_{\exp(-c\cdot)}
          \exp\bigl(6|\beta|(|X_R|+|Y_R|)\norm V_{\exp(-c\cdot)}-aL\bigr)\\
        &\quad\leq C_2\exp\bigl(-\alpha(1+d(X,Y))^\vartheta\bigr),
        \numberthis\label{eq: proof generalized doc sequential qbp}
    \end{align*}
    where $\alpha\coloneq|\beta|(aK-6C_0\norm V_{\exp(-c\cdot)})/2>0$ and
    \begin{align*}
        C_2\coloneq1+\frac{h^\prime(|\beta|)C_0\norm V_{\exp(-c\cdot)}}{\alpha}.
    \end{align*}
    These constants depend only on the lattice, $a,b,c$, $\beta$, $K$ and $\norm V_{\exp(-c\cdot)}$, and are independent of both supports.

\subsubsection*{Normalizing the local approximation}

    The exact QBP expression is normalized, whereas replacing the QBP operators by local approximations need not preserve normalization. The same approximation bound controls this error: taking $D=1$ in Equation \eqref{eq: proof generalized doc functional approximation} and using Equation \eqref{eq: proof generalized doc sequential qbp} yields
    \begin{align*}
        |\omega\bigl(
        \eta_{X_R}^L(\eta_{Y_R}^{X_R})^L
        ((\eta_{Y_R}^{X_R})^L)^*(\eta_{X_R}^L)^*
        \bigr)-1|
        \leq
        C_2\e^{-\alpha(1+d(X, Y))^\vartheta}. \numberthis \label{eq: proof doc from doc+li normalization}
    \end{align*}
    Thus, for sufficiently large $d(X, Y)$,
    \begin{align}
        \frac12\leq \omega\bigl(
        \eta_{X_R}^L(\eta_{Y_R}^{X_R})^L
        ((\eta_{Y_R}^{X_R})^L)^*(\eta_{X_R}^L)^*
        \bigr)\leq\frac32,
        \label{eq: proof generalized doc normalization}
    \end{align}
    We restrict to these distances and define the state $\phi_L$ by
    \begin{align*}
        \phi_L(D)
        \coloneq
        \frac{
        \omega\bigl(
        \eta_{X_R}^L(\eta_{Y_R}^{X_R})^L \, D \,
        ((\eta_{Y_R}^{X_R})^L)^*(\eta_{X_R}^L)^*
        \bigr)}
        {\omega\bigl(
        \eta_{X_R}^L(\eta_{Y_R}^{X_R})^L
        ((\eta_{Y_R}^{X_R})^L)^*(\eta_{X_R}^L)^*
        \bigr)}.
    \end{align*}
    Positivity and normalization give $\norm{\phi_L}=1$. Thus, normalizing the local expression changes its value at $D$ by at most $\norm D$ times its normalization error. Together with Equations \eqref{eq: proof generalized doc functional approximation}--\eqref{eq: proof doc from doc+li normalization}, this gives
    \begin{align*}
        |\omega_{V_{X_R} + V_{Y_R}}(D) - \phi_L(D)| \leq 2 C_2 \exp (-\alpha (1 + d(X,  Y))^\vartheta) \lVert D \rVert .
    \end{align*}
    Applying this estimate to $A B$, $A$ and $B$ yields
    \begin{gather}
        \big|
        \omega_{V_{X_R}+V_{Y_R}}(AB)
        -\omega_{V_{X_R}+V_{Y_R}}(A)\omega_{V_{X_R}+V_{Y_R}}(B)
        -\phi_L(AB)+\phi_L(A)\phi_L(B)
        \big|\nonumber\\
        \leq
        6 C_2 \exp \bigl( -\alpha (1 + d(X, Y))^\vartheta \bigr)
        \lVert A\rVert\lVert B\rVert.
        \label{eq: proof generalized doc covariance approximation}
    \end{gather}

\subsubsection*{Using clustering of the initial state}

    It remains to bound the covariance in $\phi_L$. The two local QBP operators have disjoint supports for large separation. Their evenness then allows us to move each operator past observables supported in the other region.
    By Lemma \ref{lem: dynamics is even}, the QBP generators are even. Proposition \ref{Ex+UniqueExpectation} shows that the conditional expectations preserve evenness. Therefore, $\eta_{X_R}^L \in \mA_{X_{R+ L}}^+$ and $(\eta_{Y_R}^{X_R})^L \in \mA_{Y_{R + L}}^+$. Since $R$ and $L$ grow sublinearly in $d(X,Y)$, we have
    \begin{align}
        d(X_{R + L}, Y_{R + L})
        \geq
        d(X, Y)-2R-2L
        \geq\frac{d(X, Y)}{2}
        \label{eq: proof generalized doc distance localized}
    \end{align}
    for all sufficiently large distances.

    To express the covariance of $\phi_L$ in terms of the initial state, we use the operators
    \begin{align*}
        A_L&\coloneq\eta_{X_R}^L A(\eta_{X_R}^L)^*\in\mA_{X_{R+L}},&
        P_L&\coloneq\eta_{X_R}^L(\eta_{X_R}^L)^*\in\mA_{X_{R+L}},\\
        B_L&\coloneq(\eta_{Y_R}^{X_R})^L B((\eta_{Y_R}^{X_R})^L)^*\in\mA_{Y_{R+L}},&
        Q_L&\coloneq(\eta_{Y_R}^{X_R})^L((\eta_{Y_R}^{X_R})^L)^*\in\mA_{Y_{R+L}}.
    \end{align*}
    Here $P_L$ and $Q_L$ account for the normalization and are even. Moreover, $A_L$ is even whenever $A$ is even, and likewise for $B_L$. Commuting the local QBP operators past the observables in the other region gives
    \begin{align*}
        \phi_L(AB)-\phi_L(A)\phi_L(B)
        =
        \frac{
        \omega(A_LB_L)\omega(P_LQ_L)
        -
        \omega(A_LQ_L)\omega(P_LB_L)}
        {\omega(P_LQ_L)^2}.
    \end{align*}
    The numerator would vanish if expectations of products between the two regions factorized. To estimate the error, we write it as four covariances of $\omega$:
    \begin{align*}
        &\omega(A_LB_L)\omega(P_LQ_L)-\omega(A_LQ_L)\omega(P_LB_L)\\
        &\quad=
        \bigl(\omega(A_LB_L)-\omega(A_L)\omega(B_L)\bigr)\omega(P_LQ_L)\\
        &\qquad+\omega(A_L)\omega(B_L)
                 \bigl(\omega(P_LQ_L)-\omega(P_L)\omega(Q_L)\bigr)\\
        &\qquad-\bigl(\omega(A_LQ_L)-\omega(A_L)\omega(Q_L)\bigr)\omega(P_LB_L)\\
        &\qquad-\omega(A_L)\omega(Q_L)
                 \bigl(\omega(P_LB_L)-\omega(P_L)\omega(B_L)\bigr).
    \end{align*}
    Each covariance is between the two enlarged supports, and one of its operators is even. Equation \eqref{eq: proof generalized doc qbp norms} gives
    \begin{align*}
        \norm{A_L}\norm{P_L}
        &\leq\exp(8|\beta|\norm V_{\exp(-c\cdot)}|X_R|)\norm A,\\
        \norm{B_L}\norm{Q_L}
        &\leq\exp(8|\beta|\norm V_{\exp(-c\cdot)}|Y_R|)\norm B.
    \end{align*}
    Taking absolute values in the four-term expansion therefore gives
    \begin{align*}
        &|\omega(A_LB_L)\omega(P_LQ_L)-\omega(A_LQ_L)\omega(P_LB_L)|\\
        &\quad\leq
        4|X_{R+L}|^{n_\cor}f_\cor(|Y_{R+L}|)
           \zeta_\cor(d(X_{R+L},Y_{R+L}))\\
        &\qquad\times
        \exp(8|\beta|\norm V_{\exp(-c\cdot)}(|X_R|+|Y_R|))\norm A\norm B.
    \end{align*}
    By Equation \eqref{eq: proof generalized doc normalization}, the denominator is at least $1/4$. Hence
    \begin{gather*}
        |\phi_L(AB)-\phi_L(A)\phi_L(B)| \numberthis\label{eq: proof generalized doc localized covariance}
        \\
        \leq 16|X_{R+L}|^{n_\cor}f_\cor(|Y_{R+L}|)
           \zeta_\cor(d(X_{R+L},Y_{R+L})) \exp(8|\beta|\norm V_{\exp(-c\cdot)}(|X_R|+|Y_R|))\norm A\norm B.
    \end{gather*}
    We must still check that the support factors and QBP norms in this bound do not grow faster than the decay. This is where the choice of $\vartheta$ is used. By $D$-regularity, Equation \eqref{eq: proof generalized doc size U W} and the definition of $L$,
    \begin{align*}
        |X_{R + L}|+|Y_{R + L}| \leq C_\vol \, C_0 \, (2+|\beta|K)^D \, (1+d (X, Y))^{\vartheta(D+1)}.
    \end{align*}
    This estimate, the decay-of-correlations assumptions and Equation \eqref{eq: proof generalized doc distance localized} give
    \begin{align*}
        &|X_{R + L}|^{n_\cor} f_\cor (|Y_{R + L}|) \zeta_\cor( d(X_{R + L }, Y_{R + L}) )
        \\
        &\quad \leq C_3 \, (1 + d (X, Y))^{n_\cor(D+1)/(D+2)} \, \exp \bigl(C_3 (1 + d (X, Y))^{(D+1)/(D+2)} -\mu_\cor \, d (X, Y)/2 \bigr),
    \end{align*}
    where
    \begin{align*}
        C_3\coloneq C_\cor^2\e^{C_\cor}
        \bigl(C_\vol C_0(2+|\beta|K)^D\bigr)^{n_\cor+1}
    \end{align*}
    depends only on the lattice, $|\beta|K$, $C_\cor$ and $n_\cor$, not on $X$, $Y$, $R$ or $L$. Since $\vartheta(D+1)=(D+1)/(D+2)<1$, the logarithms of the support and QBP norm factors grow at most as $(1+d(X,Y))^{(D+1)/(D+2)}$ and $(1+d(X,Y))^{1/(D+2)}$, respectively. Both are sublinear in the distance of $X$ and $Y$ and can therefore be absorbed into half of the exponential clustering decay, with a threshold independent of the supports. Equation \eqref{eq: proof generalized doc localized covariance} gives
    \begin{align*}
        |\phi_L(AB)-\phi_L(A)\phi_L(B)|
        &\leq16C_3\e^{-\mu_\cor d(X,Y)/4}\norm A\norm B.
    \end{align*}

\subsubsection*{Combining the errors}

    We now return to the covariance in $\omega_V$. There are three contributions: the LI error, the error from localizing the QBP operators, and the exponentially decaying terms from removing the connecting interactions and applying clustering. The factor $(1+d(X,Y))^\vartheta$ in Equation \eqref{eq: proof generalized doc remove crossing} can be absorbed into half of its exponential decay. We may therefore take $\gamma=\min(\mu_\cor,c)/4>0$, which depends only on the two decay rates. Combining Equations \eqref{eq: proof generalized doc LI}, \eqref{eq: proof generalized doc remove crossing}, \eqref{eq: proof generalized doc covariance approximation} and the preceding covariance bound gives
    \begin{align*}
        &|\omega_V(A B) - \omega_V(A) \, \omega_V(B)|
        \\
        & \quad \leq C_4 \, \lVert A\rVert \, \lVert B \rVert \, \bigl( \exp\bigl(C_\li(1+|X|+|Y|)\bigr) \, \e^{- \mu_\li R} + \e^{- \alpha (1 + d(X,  Y))^\vartheta} + \e^{- \gamma d (X, Y)} \bigr) .
    \end{align*}
    Here $C_4>0$ depends only on the lattice, $c$, $\beta$, the interaction norms and the stated clustering and LI data, and is independent of the supports. Since $R\geq1$,
    \begin{align*}
        R&\geq\frac12
          \left(\frac{(1+d(X,Y))^\vartheta}{1+|X|+|Y|}\right)^{1/D}.
    \end{align*}
    Set $\kappa\coloneq\min (\mu_\li/2,\alpha,\gamma)>0$. In this case,
    \begin{align*}
        \left(\frac{(1+d(X,Y))^\vartheta}{1+|X|+|Y|}\right)^{1/D}
        \leq (1+d(X,Y))^{\vartheta/D}
        \leq \min ( (1+d(X,Y))^\vartheta,d(X,Y) )
    \end{align*}
    for all sufficiently large distances. Consequently,
    \begin{align*}
        \e^{-\mu_\li R}
        &\leq \exp\left[-\kappa
          \left(\frac{(1+d(X,Y))^\vartheta}{1+|X|+|Y|}\right)^{1/D}\right],\\
        \e^{-\alpha(1+d(X,Y))^\vartheta}+\e^{-\gamma d(X,Y)}
        &\leq 2\exp\left[-\kappa
          \left(\frac{(1+d(X,Y))^\vartheta}{1+|X|+|Y|}\right)^{1/D}\right]
    \end{align*}
    Using that $\vartheta = 1 / (D + 2)$, the combined estimate gives
    \begin{align*}
        |\omega_V(A B) - \omega_V(A) \, \omega_V(B)|
        &\leq 3C_4\exp\bigl(C_\li(1+|X|+|Y|)\bigr)\lVert A \rVert \, \lVert B \rVert\\
        &\quad \times\exp\left( -\kappa
        \left(\frac{(1+d(X,Y))^{1/(D+2)}}{1+|X|+|Y|}\right)^{1/D}\right)
    \end{align*}
    whenever $1+|X|+|Y|\leq(1+d(X,Y))^\vartheta$. If instead $(1+|X|+|Y|)>(1+d(X,Y))^\vartheta$, then
    \begin{align*}
        \left(\frac{(1+d(X,Y))^{1/(D+2)}}{1+|X|+|Y|}\right)^{1/D}<1,
    \end{align*}
    and the trivial covariance bound gives
    \begin{align*}
        |\omega_V(A B) - \omega_V(A) \, \omega_V(B)| &\leq 2\e^\kappa\exp\bigl(C_\li(1+|X|+|Y|)\bigr)\lVert A \rVert \, \lVert B \rVert\\
        &\quad \times\exp\left(-\kappa
        \left(\frac{(1+d(X,Y))^{1/(D+2)}}{1+|X|+|Y|}\right)^{1/D}\right).
    \end{align*}
    Taking $C\geq\max (3C_4,2\e^\kappa)$ and increasing it to include the bounded distances excluded above proves the support-dependent estimate.

    The more transparent estimate follows from the weighted arithmetic--geometric mean inequality, which gives
    \begin{align*}
        (1+d(X,Y))^{1/((D+1)(D+2))}
        &\leq
        \left(\frac{(1+d(X,Y))^{1/(D+2)}}
        {1+|X|+|Y|}\right)^{1/D}
        +\frac{1+|X|+|Y|}{D+1}.
    \end{align*}
    Inserting this into the support-dependent estimate and increasing $C$ proves the bound in the theorem.
\end{proof}

\medskip
\begin{proof}[Proof of Theorem \ref{thm: QBP SLT}]
Let $A\in\mA_Y$. For $\beta=0$, all KMS states are the tracial state and the claim becomes trivial. We henceforth assume $\beta\neq0$.

Increasing $C_\cor$ if necessary, we assume that $C_\cor\geq1$. By assumption, $m_\cor\geq\nu+Dn_\cor$. We first weaken the assumed correlation estimate, uniformly in $s$ and $M$, by calculating
\[
    C_\cor(1+r)^{-m_\cor}
    \leq C_\cor(1+r)^{-\nu-Dn_\cor}.
\]
Theorem \ref{thm: LPPL without doc} may therefore be applied with this weaker decay function. Its spatial factor satisfies
\[
 (1+r)^{Dn_\cor}C_\cor(1+r/2)^{-\nu-Dn_\cor} \leq
 2^{\nu+Dn_\cor}C_\cor(1+r)^{-\nu}.
\]
Thus $\omega_{sV_M}$ satisfies LPPL with decay $(1+r)^{-\nu}$ uniformly in $s$ and $M$. Since $\nu>D$, the uniform assertion of Theorem \ref{thm: li from lppl} gives LI with decay bounded by a constant times
\[
\sum_{k=\lceil r\rceil}^{\infty}(1+k)^{D-1-\nu}
\leq C(1+r)^{D-\nu}.
\]
Hence, Proposition \ref{prop: convergence from li} shows that, for any IAS $(\Lambda_n)_{n \in \N}$ and $V_n=V_{\Lambda_n}$,
\begin{align*}
    \sup_{s\in[0,1]}|\omega_{sV_n}(B)-\omega_{sV}(B)|\longrightarrow0,
    \qquad B\in\mA_\Gamma.
\end{align*}
Taking this limit in the covariance inequality for two local observables shows that $\omega_{sV}$ satisfies decay of correlations with the same constants as $\omega_{s V_n}$.

We start by showing the continuity claim. To this end, we first observe the following elementary bound. For each $Z\in P_0(\Gamma)$, we choose $z=\mathrm C(Z)$ the center of $Z$. By Lemma \ref{lem: decay of correlations for A_infty} and Proposition \ref{prop: loc of phi}, with the clustering bound just obtained, we have that
\begin{align*}
    &\bigl|\omega_{sV_n}\bigl(\{
       \Phi_\beta^{\tau^{sV_n}}(V(Z)-\omega_{sV_n}(V(Z))),A\}\bigr)\bigr|\\
    &\quad\leq C\norm A f_\cor(|Y|)
        (1+d(z,Y))^{-\nu}
        \norm{\Phi_\beta^{\tau^{sV_n}}(V(Z))}_{\nu, z}\\
    &\quad\leq
        Ch(|\beta|)\norm A f_\cor(|Y|)
        \frac{(1+\diam Z)^\nu\norm{V(Z)}}{(1+d(z ,Y))^\nu}.
        \numberthis\label{eq: extensive common majorant}
\end{align*}
Here $h$ is the polynomial from Proposition \ref{prop: loc of phi} and $C > 0$ is a constant independent of $A$, $V$, $Y$, $z$ and $\beta$. The same bound holds for $\omega_{sV}$ and $\tau^{sV}$. The right-hand side is summable in $Z$, since
\begin{align*}
    \sum_{Z\in P_0(\Gamma)}
       \frac{(1+\diam Z)^\nu\norm{V(Z)}}{(1+d(\mathrm C(Z),Y))^\nu}
    &\leq\norm V_\nu\sum_{x\in\Gamma}(1+d(x,Y))^{-\nu}\\
    &\leq C_\sur|Y|\norm V_\nu
       \sum_{r=0}^\infty(1+r)^{D-1-\nu}<\infty.
\end{align*}

To show continuity, let $s,t\in[0,1]$. Proposition \ref{prop: QBP differential equation} gives
\begin{align*}
    |\omega_{sV}(A)-\omega_{tV}(A)|
    &\leq |s-t|\limsup_{n\to\infty}
        \sup_{u\in[0,1]}\left|\frac{\d}{\d u}\omega_{uV_n}(A)\right|\\
    &\leq\frac{|\beta|}{2}|s-t|\limsup_{n\to\infty}
        \sup_{u\in[0,1]}\sum_{Z\subseteq\Lambda_n}
        \bigl|\omega_{uV_n}\bigl(\{
          \Phi_\beta^{\tau^{uV_n}}(V(Z)-\omega_{uV_n}(V(Z))),A\}\bigr)\bigr|\\
    &\leq |s-t|\,g(|\beta|)\norm A\,|Y|f_\cor(|Y|)\norm V_\nu,
\end{align*}
where we may take the polynomial
\begin{align*}
    g(|\beta|)=\frac{|\beta|}{2}CC_\sur h(|\beta|)
                 \sum_{r=0}^\infty(1+r)^{D-1-\nu}.
\end{align*}
Here we used Equation \eqref{eq: extensive common majorant} and the preceding sum. This proves the continuity claim.

It remains to identify the derivative of the limiting state. For every finite patch, Proposition \ref{prop: QBP differential equation} gives
\begin{align*}
    \frac{\d}{\d s}\omega_{sV_n}(A)
    &=-\frac{\beta}{2}\sum_{Z\subseteq\Lambda_n}
       \omega_{sV_n}\bigl(\{
          \Phi_\beta^{\tau^{sV_n}}(V(Z)-\omega_{sV_n}(V(Z))),A\}\bigr).
\end{align*}
We show that the right-hand side converges uniformly in $s\in[0,1]$ to the corresponding infinite-volume sum. We first fix $Z\in P_0(\Gamma)$ and prove uniform convergence of its summand. Lemma \ref{lem: str conv time evol} gives convergence of the dynamics uniformly in $s$ on compact time intervals. Splitting the integral defining $\Phi_\beta$ at $T>0$, we find
\begin{align*}
    &\sup_{s\in[0,1]}
      \norm{\Phi_\beta^{\tau^{sV_n}}(V(Z))
                 -\Phi_\beta^{\tau^{sV}}(V(Z))}\\
    &\quad\leq
      \sup_{\substack{s\in[0,1]\\|t|\leq T}}
      \norm{\tau_{-t}^{sV_n}(V(Z))-\tau_{-t}^{sV}(V(Z))}
      +2\norm{V(Z)}\int_{|t|>T}f_\beta(t)\,\d t.
\end{align*}
It follows that
\begin{align*}
    \sup_{s\in[0,1]}
      \norm{\Phi_\beta^{\tau^{sV_n}}(V(Z))
                 -\Phi_\beta^{\tau^{sV}}(V(Z))}
    \longrightarrow0.
\end{align*}
Together with the uniform convergence of the states, this gives
\begin{align*}
    &\sup_{s\in[0,1]}
       \norm{\{\Phi_\beta^{\tau^{sV_n}}(V(Z)-\omega_{sV_n}(V(Z))),A\}
              -\{\Phi_\beta^{\tau^{sV}}(V(Z)-\omega_{sV}(V(Z))),A\}}\\
    &\quad\leq2\norm A\sup_{s\in[0,1]}
       \norm{\Phi_\beta^{\tau^{sV_n}}(V(Z))-\Phi_\beta^{\tau^{sV}}(V(Z))}\\
    &\qquad+2\norm A\sup_{s\in[0,1]}
       |\omega_{sV_n}(V(Z))-\omega_{sV}(V(Z))|
       \longrightarrow0.
\end{align*}

We still have to replace the state acting on this observable. The continuity assertion of Lemma \ref{lem: str conv time evol}, together with the same integral estimate, shows that $s\mapsto\Phi_\beta^{\tau^{sV}}(V(Z))$ is norm-continuous. Since $s\mapsto\omega_{sV}(V(Z))$ is continuous, the set
\begin{align*}
    \bigl\{\{\Phi_\beta^{\tau^{sV}}(V(Z)-\omega_{sV}(V(Z))),A\}:s\in[0,1]\bigr\}
\end{align*}
is compact. For $\epsilon>0$, choose an $\epsilon$-net $D_1,\ldots,D_N$ for this set. Since both states have norm one,
\begin{align*}
    &\sup_{s\in[0,1]}\bigl|(\omega_{sV_n}-\omega_{sV})
          \bigl(\{\Phi_\beta^{\tau^{sV}}(V(Z)-\omega_{sV}(V(Z))),A\}\bigr)\bigr|\\
    &\quad\leq2\epsilon+
        \max_{1\leq j\leq N}\sup_{s\in[0,1]}
            |(\omega_{sV_n}-\omega_{sV})(D_j)|.
\end{align*}
The last term converges to zero by the uniform convergence of the states. First letting $n\to\infty$ and then $\epsilon\to0$, we conclude that, for every fixed $Z$,
\begin{align*}
    &\sup_{s\in[0,1]}\Bigl|
       \omega_{sV_n}\bigl(\{
          \Phi_\beta^{\tau^{sV_n}}(V(Z)-\omega_{sV_n}(V(Z))),A\}\bigr)-\omega_{sV}\bigl(\{
          \Phi_\beta^{\tau^{sV}}(V(Z)-\omega_{sV}(V(Z))),A\}\bigr)
       \Bigr|\longrightarrow0.
\end{align*}

We regard a summand in the finite-patch derivative as zero when $Z\nsubseteq\Lambda_n$. For each fixed $Z$, this does not affect the preceding limit, since $Z\subseteq\Lambda_n$ for all sufficiently large $n$. Equation \eqref{eq: extensive common majorant} bounds all summands by the same summable expression, uniformly in $s$ and $n$. Dominated convergence therefore gives uniform convergence of the finite-patch derivatives to
\begin{align*}
    -\frac{\beta}{2}\sum_{Z\in P_0(\Gamma)}
       \omega_{sV}\bigl(\{
          \Phi_\beta^{\tau^{sV}}(V(Z)-\omega_{sV}(V(Z))),A\}\bigr).
\end{align*}
The derivatives for finite patches are continuous by Proposition \ref{prop: QBP differential equation}, Lemma \ref{lem: continuity of phi} and norm continuity of bounded perturbations. Their uniform limit is therefore continuous. For any $s_0 \in [0,1]$, the finite-patch identity
\begin{align*}
    \omega_{sV_n}(A)-\omega_{s_0 V_n}(A)
    =\int_{s_0}^s\frac{\d}{\d u}\omega_{uV_n}(A)\,\d u
\end{align*}
converges uniformly to
\begin{align*}
    \omega_{sV}(A)-\omega_{s_0 V}(A)
    &=-\frac{\beta}{2}\int_{s_0}^s\sum_{Z\in P_0(\Gamma)}
       \omega_{uV}\bigl(\{
          \Phi_\beta^{\tau^{uV}}(V(Z)-\omega_{uV}(V(Z))),A\}\bigr)\,\d u.
\end{align*}
The integrand is continuous. Hence, the fundamental theorem of calculus shows that $s\mapsto\omega_{sV}(A)$ is continuously differentiable and
\begin{align*}
    \frac{\d}{\d s}\omega_{sV}(A)
    &=-\frac{\beta}{2}\sum_{Z\in P_0(\Gamma)}
       \omega_{sV}\bigl(\{
          \Phi_\beta^{\tau^{sV}}(V(Z)-\omega_{sV}(V(Z))),A\}\bigr).
\end{align*}
Finally, absolute convergence and norm convergence of $V_x=\sum_{Z\in R_x}V(Z)$ give
\begin{align*}
    \sum_{Z\in P_0(\Gamma)}
       \omega_{sV}\bigl(\{
          \Phi_\beta^{\tau^{sV}}(V(Z)-\omega_{sV}(V(Z))),A\}\bigr) &=
       \sum_{x\in\Gamma}\sum_{Z\in R_x}
       \omega_{sV}\bigl(\{
          \Phi_\beta^{\tau^{sV}}(V(Z)-\omega_{sV}(V(Z))),A\}\bigr)\\
    &=\sum_{x\in\Gamma}
       \omega_{sV}\bigl(\{
          \Phi_\beta^{\tau^{sV}}(V_x-\omega_{sV}(V_x)),A\}\bigr).
\end{align*}
This proves the stated formula, including absolute and uniform convergence in $s$.
\end{proof}

\appendix

\section{Technical lemmas} \label{sec: technical proofs}

\subsection{Quantum Belief Propagation} \label{app: perturbed KMS}

\medskip
\begin{lemma} \label{lem: impl is entire}
    Let $(\mA, \tau)$ be a $C^*$-dynamical system, $A \in \mA_\ent^\tau$ and $\omega$ a $(\tau, \beta)$-KMS state. Let $(\mH, \pi, \Omega)$ be the associated GNS triple and $H$ the Hamiltonian that implements the dynamics in the GNS Hilbert space. It holds that the map
    \begin{align*}
        z \mapsto \e^{\i z H} \, \pi(A) \, \Omega
    \end{align*}
    is entire with respect to the norm of $\mH$. Moreover, it holds that
    \begin{align*}
        \pi(\tau_z(A)) \, \Omega = \e^{\i z H} \, \pi(A) \, \Omega \, ,
    \end{align*}
    for all $z \in \C$.
\end{lemma}
\begin{proof}
    We first notice that since $\omega$ is $\tau$ invariant, it holds that for $t \in \R$
    \begin{align*}
        \pi( \tau_t(A) ) \, \Omega = \e^{\i t H} \pi(A) \, \Omega \, .
    \end{align*}
    Since $A \in \mA_\ent^\tau$ we find that the left-hand side is differentiable. Hence, by the definition of the generator of unitary groups, it holds that
    \begin{align*}
        \pi \bigl( \mL (\tau_t (A)) \bigr) \, \Omega = H \, \e^{\i t H} \, \pi(A) \, \Omega \, .
    \end{align*}
    Again, the left-hand side is differentiable, and we can iterate this process to find
    \begin{align*}
        \pi(\mL^n (A)) \, \Omega = H^n \, \pi(A) \, \Omega \, ,
    \end{align*}
    for every $n\in\N$. In particular, $\pi(A)\Omega\in D(H^n)$. Fix $R>0$. Cauchy's estimate for the entire map $z\mapsto\tau_z(A)$ gives
    \begin{align*}
        \frac{\norm{H^k\pi(A)\Omega}_{\mH}}{k!}
        &\leq\frac{\norm{\mL^k(A)}}{k!}
        \leq\frac{\sup_{|z|=2R}\norm{\tau_z(A)}}{(2R)^k}.
    \end{align*}
    Now let $N \geq M$ for $N$, $M \in \N$. Using this bound we find that
    \begin{align*}   
        \sup_{|z| \leq R} \bigl \lVert \sum_{k=0}^{N}\frac{(\i z)^k}{k!}H^k\pi(A)\Omega - \sum_{k=0}^{M}\frac{(\i z)^k}{k!}H^k\pi(A)\Omega \bigr \rVert
        &\leq 
        \sup_{|z|\leq R}\bigl\|
             \sum_{k=M+1}^{N}\frac{(\i z)^k}{k!}H^k\pi(A)\Omega
             \bigr\|_{\mH}
             \\
        &\leq\sup_{|z|=2R}\norm{\tau_z(A)}
                   \sum_{k=M+1}^\infty2^{-k}.
    \end{align*}
    Hence the series converges locally uniformly to an entire vector-valued function. To show that it equals the exponential function defined via the spectral theorem let $\mathbf1_{[-L,L]}(H)$ be the spectral projection of $H$ onto $[-L,L]$. Then
    \begin{align*}
        \norm{\e^{R|H|}\mathbf1_{[-L,L]}(H)\pi(A)\Omega}_{\mH}
        &\leq\sum_{k=0}^\infty\frac{R^k}{k!}
                        \norm{H^k\pi(A)\Omega}_{\mH}\\
        &\leq2\sup_{|z|=2R}\norm{\tau_z(A)}.
    \end{align*}
    The spectral theorem and monotone convergence as $L \to \infty$ show that $\pi(A)\Omega\in D(\e^{R|H|})$. The spectral theorem and dominated convergence therefore give, for $|z|\leq R$,
    \begin{align*}
        \e^{\i zH}\pi(A)\Omega
        &=\sum_{k=0}^\infty\frac{(\i z)^k}{k!}H^k\pi(A)\Omega\\
        &=\pi\left(\sum_{k=0}^\infty\frac{(\i z)^k}{k!}\mL^k(A)\right)\Omega
        =\pi(\tau_z(A))\Omega.
    \end{align*}
    Since $R$ was arbitrary, both assertions follow.
\end{proof}

\medskip
\begin{lemma} \label{lem: continuity of phi}
    Let $(\mA, \tau)$ be a $C^*$-dynamical system, $P = P^* \in \mA$ and $s \in [0, 1]$. Denote by $\tau^{s P}$ the perturbed dynamics. Then the map 
    \begin{align*}
        s \mapsto \Phi_\beta^{\tau^{s P}} (A) \, ,   
    \end{align*}
    where $\Phi_\beta$ is defined in Equation \eqref{def: Phi}, is continuous (in norm) for any $A \in \mA$.
\end{lemma}
\begin{proof}
    For $\beta=0$ the map is constant because $\Phi_0=\mathrm{id}$. Assume $\beta\neq0$. Let $s$ and $s_0 \in [0, 1]$. We calculate
    \begin{align*}
        \Phi_\beta^{\tau^{s P}} (A) - \Phi_\beta^{\tau^{s_0 P}} (A) &= \int_\R \d t \, f_\beta (t) \, \bigl( \tau_{- t}^{s P} (A) - \tau_{- t}^{s_0 P} (A) \bigr) 
        \\
        &= \int_\R \d t \, f_\beta (t) \, \bigl( (\tau_{- t}^{s_0 P})^{(s - s_0) P} (A) - \tau_{- t}^{s_0 P} (A) \bigr) \, ,
    \end{align*}
    where the second equality follows by the uniqueness of the dynamics generated by a given derivation. Duhamel's formula gives
    \[
        \norm{\tau_{-t}^{sP}(A)-\tau_{-t}^{s_0P}(A)}
        \leq 2|t|\,|s-s_0| \, \norm P \, \norm A.
    \]
    Hence,
    \[
        \norm{\Phi_\beta^{\tau^{sP}}(A)-
        \Phi_\beta^{\tau^{s_0P}}(A)}
        \leq2|s-s_0|\norm P\norm A
        \int_\R |t|f_\beta(t)\,\d t
        \longrightarrow0.
    \]
    This proves the statement.
\end{proof}

\subsection{Spaces of localized operators and interactions}

We introduce the fermionic conditional expectation. To this end first note that $\mA_\Gamma$ has a unique state $\omega^{\tr}$ that satisfies
\begin{align*}
    \omega^{\tr} (AB) = \omega^{\tr} (BA)
\end{align*}
for all  $A$, $B \in \mA_\Gamma$, called the tracial state (e.g.\ \cite[Definition 4.1, Remark 2]{AM03}).

\medskip

\begin{proposition}\label{Ex+UniqueExpectation}
    For each $M \subseteq \Gamma$ there exists a unique linear map 
    \begin{align*}
        \E_M : \mA_\Gamma \to \mA_M\,,
    \end{align*}
    called the conditional expectation with respect to $\omega^\tr$, such that
    \begin{equation*}
      \forall A \in \mA_\Gamma \; \; \forall B \in \mA_M \, :\quad \omega^{\tr}(A B) = \omega^{\tr} (\E_M(A) B) \,.
    \end{equation*}
    It is unital, positive and has the properties 
    \begin{eqnarray*}
        \forall M \subseteq \Gamma \;\; \forall A, C \in \mA_M \;\; \forall B \in \mA_\Gamma \, : &&  \E_M (A\, B\, C) = A \, \E_M(B) \, C 
        \\[1mm]
        \forall M_1, M_2 \subseteq \Gamma \, : && \E_{M_1} \circ \E_{M_2} = \E_{M_1 \cap M_2} \,
        \\[1mm]
        \forall M \subseteq \Gamma \,: && \E_M \mA_\Gamma^+ \subseteq \mA_\Gamma^+
        \\[1mm]
        \forall M \subseteq \Gamma \; \;  \forall A \in \mA_\Gamma\,:&&  \lVert \E_M(A) \rVert \leq \lVert A \rVert \,.
    \end{eqnarray*}
\end{proposition}
\begin{proof}
This statement was proved for $\Gamma=\Z^D$ in \cite[Theorem 4.7]{AM03}. The proof uses only the CAR structure and the tracial state, and therefore applies to the present lattice.
\end{proof}

\medskip
\begin{lemma} \label{lem: derivations}
      Let $\alpha > 0$ and $\Phi \in  \mP_{D +\alpha}$. It holds that $\mA_{D + \alpha} \subseteq D(\mL_{\Phi})$. For all $ A\in \mA_{D + \alpha}$ the sums
    \[\sum_{M\in P_0(\Gamma)} [\,\Phi(M), \,A\,] \quad \text{and} \quad \sum_{x \in \Gamma} [\, \Phi_x, \, A \, ] \]
    converge absolutely and
    \[ \mL_{\Phi } \, A 
    \;=\;  \sum_{M\in P_0(\Gamma)} [\,\Phi(M)  , \,A\,]
    \;=\;  \sum_{x\in \Gamma} [\, \Phi_x, \, A \, ] 
    \, .
    \]
    Furthermore, there is a constant $c_\alpha$, independent of $\Phi$ and $A$, such that for all~$x\in \Gamma$
    \[
    \norm{ \mL_{\Phi } \, A }_{\alpha,x} \leq c_\alpha \, \norm{\Phi}_{D + \alpha} \, \norm{A}_{D + \alpha, x}
    \,.
    \]
\end{lemma}
\begin{proof}
    Fix $x\in\Gamma$ and set
    \begin{align*}
        A_0=\E_{B_1(x)}A,\qquad
        A_j=\E_{B_{2^j}(x)}A-\E_{B_{2^{j-1}}(x)}A,\quad j\geq1.
    \end{align*}
    Then $A=\sum_{j=0}^\infty A_j$ in norm and $A_j\in\mA_{B_{2^j}(x)}$. By surface regularity and the definition of the decay norm, we find that
    \[
        \norm{A_j} \, |B_{2^j}(x)|
        \leq 2^{2D+\alpha+1}C_\vol\,2^{-j\alpha}\norm A_{D+\alpha,x},
    \]
    for all $j\geq0$. We also have the bound
    \[
        \sum_M\norm{[\Phi(M),A_j]}
        \leq2\norm{A_j} \, |B_{2^j}(x)| \, \norm\Phi_{D+\alpha}.
    \]
    Summing over $j$ proves absolute convergence of the commutator sum. The local partial sums converge to $A$ and their derivations converge in norm. Closedness of $\mL_\Phi$ therefore gives $A\in D(\mL_\Phi)$ and the first identity. Absolute convergence also permits regrouping the terms by their centers and proves the second identity.

    To bound the decay norm, fix $r\geq4$. If $2^j\leq r/2$, every nonzero commutator contributing to
    $(1-\E_{B_r(x)})\mL_\Phi A_j$ has $M\cap B_{2^j}(x)\neq\emptyset$ and $M\nsubseteq B_r(x)$, hence $\diam M\geq r/2$. Thus
    \begin{align*}
        \norm{(1-\E_{B_r(x)})\mL_\Phi A}
        &\leq4(1+r/2)^{-D-\alpha}\norm\Phi_{D+\alpha}
              \sum_{2^j\leq r/2}\norm{A_j} \, |B_{2^j}(x)| + 4 \norm\Phi_{D+\alpha}
              \sum_{2^j>r/2}\norm{A_j} \, |B_{2^j}(x)|\\
        &\leq 2^{2D+\alpha+3}C_\vol\norm\Phi_{D+\alpha}\norm A_{D+\alpha,x}
          \left((1+r/2)^{-D-\alpha}\sum_{j=0}^\infty2^{-j\alpha}
                +\sum_{2^j>r/2}2^{-j\alpha}\right)\\
        &\leq c_\alpha(1+r)^{-\alpha}
                     \norm\Phi_{D+\alpha}\norm A_{D+\alpha,x}.
    \end{align*}
    The geometric sums converge since $\alpha>0$. The norm bound above covers $r<4$, and increasing $c_\alpha$ proves the claim.
\end{proof}

\begin{lemma}\label{lem:commutator bound}
    Let $F$ be a decay function, $x \in \Gamma$, $Y \in P_0 (\Gamma)$, and $A\in \mA_Y$, $B \in \mA_{F}$, where one of them is even. It holds that
    \begin{align*}
        \norm{[A,B]} \leq \frac{2}{F(1)} F(d(x, Y)) \,\norm{A} \, \norm{B}_{F, x}\, .
    \end{align*}
\end{lemma}

\begin{proof}
    Let $x \in \Gamma$, $Y \in P_0 (\Gamma)$ and $A \in \mA_Y$, $B \in \mA_F$. First, we notice that the statement is trivial for $x \in Y$. Therefore, we only consider the case $d(x, Y) \geq 1$. We define $\gamma \coloneq d(x, Y) - 1$ and calculate
    \begin{align*}
        \lVert [A, B] \rVert \leq \lVert [A, \E_{B_{\gamma} (x)} B] \rVert + \lVert [A, (1 - \E_{B_\gamma (x)} ) \, B] \rVert \, .
    \end{align*}
    The first term vanishes by definition of $\gamma$. Hence, we find
    \begin{align*}
        \lVert [A, B] \rVert &\leq \lVert [A, (1 - \E_{B_\gamma (x)} ) \, B] \rVert 
        \\
        & \leq 2 \lVert A \rVert \, \lVert (1 - \E_{B_\gamma (x)}) \, B \rVert 
        \\
        &\leq 2 F (\gamma) \, \lVert A \rVert \, \lVert B \rVert_{F, x} \, . 
    \end{align*}
    By convexity of the logarithm of $F$ we can bound $F(\gamma ) \leq \frac{1}{F(1)} F(d(x, Y))$. This estimate yields the claim.
\end{proof}

\medskip
\begin{lemma} \label{lem: continuity decay norm}
    Let $F$ be a decay function and $A$, $B \in \mA_F$. It follows that
    \begin{align*}
        \lVert A B \rVert_{F, x} \leq 2 \, \lVert A \rVert_{F, x} \, \lVert B \rVert_{F, x} \, , 
    \end{align*}
    for all $x \in \Gamma$. Moreover, for each $x\in\Gamma$, the set of all $A\in\mA_\Gamma$ with $\norm A_{F,x}<\infty$ is complete with respect to $\norm{\cdot}_{F,x}$.
\end{lemma}

\begin{proof}
    Fix $x\in\Gamma$ and $k\in\N_0$. We calculate
    \begin{align*}
        \norm{AB-\E_{B_k(x)}(AB)} &=\norm{(1-\E_{B_k(x)})
        \bigl(AB-\E_{B_k(x)}(A)\E_{B_k(x)}(B)\bigr)}\\
        &\quad\leq2\norm{AB-\E_{B_k(x)}(A)\E_{B_k(x)}(B)}\\
        &\quad\leq2\norm{A-\E_{B_k(x)}(A)}\norm B
        +2\norm A\norm{B-\E_{B_k(x)}(B)}.
    \end{align*}
    After division by $F(k)$ and taking the supremum over $k$, this gives
    \[
    \begin{aligned}
        \norm{AB}_{F,x}
        &\leq\norm A\norm B
        +2(\norm A_{F,x}-\norm A)\norm B
        +2\norm A(\norm B_{F,x}-\norm B)\\
        &\leq2\norm A_{F,x}\norm B_{F,x}.
    \end{aligned}
    \]
    For completeness, let $(A_n)_{n\in\N}$ be Cauchy in
    $\norm{\cdot}_{F,x}$. It is Cauchy in the $C^*$-norm and hence converges
    in norm to some $A\in\mA_\Gamma$. For every $n$ and $k$,
    \[
    \begin{aligned}
      \frac{\norm{(A-A_n)-\E_{B_k(x)}(A-A_n)}}{F(k)}
      &=\lim_{m\to\infty}
      \frac{\norm{(A_m-A_n)-\E_{B_k(x)}(A_m-A_n)}}{F(k)}\\
      &\leq\limsup_{m\to\infty}\norm{A_m-A_n}_{F,x}.
    \end{aligned}
    \]
    The right-hand side is independent of $k$. Taking the supremum
    over $k$, and also using
    $\norm{A-A_n}\leq\limsup_{m\to\infty}\norm{A_m-A_n}_{F,x}$, gives
    \[
    \norm{A-A_n}_{F,x}
    \leq2\limsup_{m\to\infty}\norm{A_m-A_n}_{F,x}.
    \]
    Letting $n\to\infty$ proves that
    $\norm{A-A_n}_{F,x}\to0$.
\end{proof}

\subsection{Locality of KMS states}

\medskip
\begin{lemma} \label{lem: decay of correlations for A_infty}
    Let $\beta \in \R$ and $\omega$ be a $(\tau, \beta )$-KMS state on $\mA_\Gamma$. Let $F$ be some decay function. Assume that $\omega$ satisfies decay of correlations with respect to $\zeta_\cor$, $f_\cor$, and $n_\cor$. Assume that 
    \begin{align*}
        F(r) \leq C (1 + r)^{D \, n_\cor} \, \zeta_{\cor} (r) \, , \qquad r \geq 0 \, , 
    \end{align*}
    for some constant $C > 0$. Then there exists another constant $C^\prime > 0$ only depending on $\Gamma$, $n_\cor$, $C$ and $F(1)^{-1}$ such that
    \begin{align*}
        |\omega(A B ) - \omega(A) \, \omega(B)| \leq C^\prime \, \lVert A \rVert \, \lVert B \rVert_{F, x} \, f_\cor(|Y|) \, (1 + d(x, Y))^{D \, n_\cor} \, \zeta_{\cor} (d(x, Y) / 2 ) \, ,
    \end{align*}
    for all $Y \in P_0 (\Gamma)$, $x \in \Gamma$, and $A \in \mA_Y$, $B \in \mA_F$, where one of them is even.
\end{lemma}
\begin{proof}
    Let $A\in\mA_Y$ and $B\in\mA_F$ be such that one of them is even. If $d(x, Y)=0$, then the assumption $F(0)=1\leq C\zeta_\cor(0)$ and $f_\cor(|Y|)\geq1$ give
    \begin{align*}
        |\omega(AB)-\omega(A)\omega(B)|
        &\leq2\norm A\norm B\\
        &\leq2C\norm A\norm B_{F,x}f_\cor(|Y|)
        (1+d(x, Y))^{Dn_\cor}\zeta_\cor(d(x, Y)/2).
    \end{align*}
    We therefore assume $d(x, Y)\geq1$ and set $k=\lfloor d(x, Y)/2\rfloor$. We calculate
    \begin{align*}
        |\omega(AB)-\omega(A)\omega(B)|
        &\leq |\omega(A(B-\E_{B_k(x)}B))
        -\omega(A)\omega(B-\E_{B_k(x)}B)|\\
        &\quad+|\omega(A\E_{B_k(x)}B)
        -\omega(A)\omega(\E_{B_k(x)}B)|.
    \end{align*}
    Since $k=\lfloor d(x, Y)/2\rfloor$, convexity of $\log F$ gives $F(k) \leq F(d(x, Y)/2) / F(1)$. The first term is therefore bounded by
    \begin{align*}
        &|\omega(A(B-\E_{B_k(x)}B))
        -\omega(A)\omega(B-\E_{B_k(x)}B)|\\
        &\quad\leq2F(k)\norm A\norm B_{F,x}\\
        &\quad\leq\frac{2C}{F(1)}\norm A\norm B_{F,x}
        f_\cor(|Y|)(1+d(x, Y))^{Dn_\cor}\zeta_\cor(d(x, Y)/2).
    \end{align*}
    Furthermore, $B_k(x)\cap Y=\emptyset$ and $d(B_k(x),Y) \geq d(x, Y)/2$. The conditional expectation preserves evenness. Hence one of
    $\E_{B_k(x)}B$ and $A$ is even, and these two operators commute because their supports are disjoint. Applying decay of correlations gives
    \begin{align*}
        |\omega(A\E_{B_k(x)}B)
        -\omega(A)\omega(\E_{B_k(x)}B)|  &=|\omega((\E_{B_k(x)}B)A)
        -\omega(\E_{B_k(x)}B)\omega(A)|\\
        &\leq\norm{\E_{B_k(x)}B}\norm A
        |B_k(x)|^{n_\cor}f_\cor(|Y|)
        \zeta_\cor(d(B_k(x),Y))\\
        &\leq C_\vol^{n_\cor}\norm A\norm B_{F,x}
        f_\cor(|Y|)(1+d(x, Y))^{Dn_\cor}\zeta_\cor(d(x, Y)/2).
    \end{align*}
    Adding both estimates proves the claim.
\end{proof}

\medskip
\begin{lemma} \label{lem: dynamics is even}
    Consider the $C^*$-algebra $\mA_\Gamma$ and let $\tau^P$ be the dynamics generated by the derivation $\mL_\Phi + \mL_P$, where $\Phi \in \mP_{D+\alpha}$ with $\alpha>D$ and $P = P^* \in \mA_\Gamma^+$. Then the dynamics maps even observables to even observables, that is
    \begin{align*}
        \tau_t^P \, \mA_\Gamma^+ \subseteq \mA_\Gamma^+ \, ,
    \end{align*}
    for all $t \in \R$.
\end{lemma}
\begin{proof}
    First, notice that $\Theta^2 = \mathrm{id}$. Let $A \in \mA_0$ and calculate
    \begin{align*}
        - \i \, \partial_t \Theta \, \tau_t^P  \, \Theta (A) &= \Theta  \, \tau_t^P (\mL_\Phi (\Theta (A)) + \mL_{P} (\Theta (A)) )
        \\
        &= \Theta \, \tau_t^P \bigl( \sum_{x \in \Gamma} [\Phi_x, \Theta (A)]  + [P, \Theta (A)] \bigr)
        \\
        &= \Theta \, \tau_t^P \, \Theta (\mL_\Phi(A) + \mL_P(A) ) \, ,
    \end{align*}
    where we have used Lemma \ref{lem: derivations}. By uniqueness of the dynamics that is generated by $\mL_\Phi + \mL_P$ it follows that $\Theta \, \tau^P \, \Theta = \tau^P$.
\end{proof}

\medskip
\begin{lemma}\label{lem: str conv time evol}
    Let $H,V\in\mP_F$ for a decay function $F$, and assume that $\mL_{H+sV}$ generates a dynamics $\tau^{sV}$ for every $s\in[0,1]$. For an IAS $(\Lambda_n)_{n\in\N}$, we set $V_n  \coloneq V_{\Lambda_n}$. Then, for every $A\in\mA_\Gamma$ and $T>0$,
    \begin{align*}
        \lim_{n \to \infty} \sup_{\substack {s \in [0,1]}} \sup_{|t| \leq T} \norm{ \tau_t^{s V_n} (A) - \tau_t^{sV} (A)} = 0.
    \end{align*}
    Moreover, $s\mapsto\tau_t^{sV}(A)$ is norm-continuous, uniformly for $|t|\leq T$.
\end{lemma}
\begin{proof}
    The algebra $\mA_0$ is a core for $\mL_{H+sV}$ by the definition of this derivation as a closure. For $A\in\mA_Y$ and $s_n\to s$, we calculate
    \begin{align*}
        &\norm{(\mL_H+s_n\mL_{V_n}-\mL_{H+sV})A} \leq
        2 \norm A \, |s_n-s| \sum_{\substack{Z \in P_0 (\Gamma) \\ Z \cap Y \neq \emptyset }} \norm{V(Z)}
        +2\norm A\sum_{\substack{Z \in P_0 (\Gamma) \\ Z\cap Y\neq\emptyset,  Z \nsubseteq \Lambda_n}}\norm{V(Z)}. \numberthis \label{eq: proof lemma convergence}
    \end{align*}
    We have the elementary bound on the sum
    \begin{align*}
        \sum_{Z\cap Y\neq\emptyset}\norm{V(Z)}
        \leq\sum_{y\in Y}\sum_{Z\ni y}\norm{V(Z)}
        \leq |Y|\norm V_F<\infty.
    \end{align*}
    Hence, the first term in Equation \eqref{eq: proof lemma convergence} tends to zero as $n \to \infty$. Also the second term in this equation vanishes in the limit $n \to \infty$ by dominated convergence. All the time evolutions are isometries. The first Trotter--Kato approximation theorem \cite[Theorem III.4.8]{EN00}, applied to $\pm\i(\mL_H+s_n\mL_{V_n})$ and $\pm\i\mL_{H+sV}$, yields
    \begin{align*}
        \sup_{|t|\leq T}\norm{\tau_t^{s_nV_n}(A)-\tau_t^{sV}(A)}\longrightarrow0.
    \end{align*}
This convergence holds for every $A\in\mA_\Gamma$.
    We first establish the continuity in $s$ which will be used below. For $A\in\mA_Y$ and $s_n\to s$ we have
    \begin{align*}
        \norm{(\mL_{H+s_nV}-\mL_{H+sV})A}
        &\leq2|s_n-s|\norm A
        \sum_{Z\cap Y\neq\emptyset}\norm{V(Z)}\\
        &\leq2|s_n-s|\,|Y|\norm V_F\norm A
        \longrightarrow0.
    \end{align*}
By the same argument as above, we find that
    \begin{align*}
        \sup_{|t|\leq T}\norm{\tau_t^{s_nV}(A)-\tau_t^{sV}(A)}
        \longrightarrow0,\qquad A\in\mA_\Gamma.
    \end{align*}
    This proves the continuity assertion. We now show that the convergence in $n$ is uniform in $s\in[0,1]$. If this failed, there would be a subsequence $n_j$ and parameters $s_j$ with a positive lower bound on
    \begin{align*}
        \sup_{|t|\leq T}
        \lVert\tau_t^{s_jV_{n_j}}(A)-\tau_t^{s_jV}(A)\rVert.
    \end{align*}
    By compactness, we may pass to a further subsequence such that $s_j\to s$. We calculate
    \begin{align*}
        &\sup_{|t|\leq T}\norm{\tau_t^{s_jV_{n_j}}(A)-\tau_t^{s_jV}(A)}\\
        &\quad\leq
        \sup_{|t|\leq T}\norm{\tau_t^{s_jV_{n_j}}(A)-\tau_t^{sV}(A)}
        +\sup_{|t|\leq T}\norm{\tau_t^{s_jV}(A)-\tau_t^{sV}(A)}
        \longrightarrow0.
    \end{align*}
    The first term converges to zero by the preceding convergence result for the sequence $s_j\to s$, and the second term converges to zero by the continuity just proved. This contradiction proves uniform convergence.
\end{proof}

\section{Locality estimates} \label{app: LR bounds}

\medskip
\begin{proposition} \label{prop: LiebRobinson} 
    Let $\Phi \in \mP_{D + \alpha}$ be an interaction with $\alpha > D$ and let $\sigma \in ((D + 1)  / (\alpha + 1), 1)$ be arbitrary. Then $\mL_{\Phi}$ generates a unique group of automorphisms $\tau$. Furthermore, there exists a polynomially growing function $f \colon \R_+ \to \R_+$, such that $\tau$ satisfies a Lieb--Robinson bound of the form
    \begin{align*}
        \zeta_\lr (X, Y, |t|)
        \leq  \min( |X|, |Y|) \, f(\lvert t \rvert) \bigl(\exp \bigl( \min( 0, v \lvert t \rvert - d(X,Y)^{1 - \sigma }) \bigr) + (1 + d (X, Y))^{- \sigma \alpha} \bigr) \, .
    \end{align*}
    Moreover, for any $C > 0$, $f$ and $v > 0$ can be chosen uniformly for all $\Phi$ with $ \lVert \Phi \rVert_{D + \alpha} < C$.
    
    Independently, if $\Phi\in\mP_{\exp(-\mu\cdot)}$ for some $\mu>0$, then, for every $0<\lambda<\mu$, there are constants $C_{\mu,\lambda},\kappa_{\mu,\lambda}>0$, depending only on $\Gamma$, $\mu$ and $\lambda$, such that
    \begin{align*}
        \zeta_\lr(X,Y,|t|)
        &\leq C_{\mu,\lambda}\min (|X|, |Y| ) \, 
        \exp\bigl(\kappa_{\mu,\lambda}\lVert\Phi\rVert_{\exp(-\mu\cdot)}|t|
                   -\lambda d(X,Y)\bigr).
    \end{align*}
    In particular, the coefficient of $|t|$ is linear in the interaction norm.
\end{proposition}
\begin{proof}
    The existence of dynamics follows by standard Lieb--Robinson bounds. The result for polynomially decaying interactions is shown in \cite[Theorem 6]{TW25}. The result for exponentially decaying interactions follows from standard Lieb--Robinson bounds for short-ranged interactions; see, for instance, \cite[Theorem 4.3]{BP17}.
\end{proof}

\medskip
\begin{lemma} \label{lem: perturbed LRB}
    Let $H \in \mP_\nu$ with $\nu > 2 D$ and $P = P^* \in \mA_X^+$. Assume that $H$ generates a dynamics $\tau$ that satisfies a Lieb--Robinson bound with $\zeta_\lr$. Then the perturbed dynamics $\tau^P$ satisfies a Lieb--Robinson bound of the following type
    \begin{align*}
        \lVert [\tau_t^P (A), \, B] \rVert \leq \lVert A \rVert \, \lVert B \rVert \bigl( \zeta_\lr (Y, Z, |t|) + 2 |t| \, \lVert P \rVert \min_{W \in \{ Y, Z\}} \zeta_\lr (X, W, |t|)  \bigr) \, ,
    \end{align*}
    for all $A \in \mA_Y^+$, $B \in \mA_Z$ and $t \in \R$, provided that $t\mapsto\zeta_\lr(Y,Z,t)$ is non-decreasing for every $Y,Z\in P_0(\Gamma)$.
\end{lemma}
\begin{proof}
    This statement was shown for finite lattices in \cite[Lemma 42]{CMTW25}. We give a derivation in our setting. To this end, we first compare the perturbed and unperturbed dynamics. Let $C \in \mA_W$. By differentiation we find that
    \begin{align*}
        \tau_t^P(C)-\tau_t(C)
        &=\i\int_0^t\tau_s^P\bigl([P,\tau_{t-s}(C)]\bigr)\,\d s.
    \end{align*}
    Since $P$ is even, the Lieb--Robinson bound and its monotonicity in time imply that
    \begin{align*}
        \norm{\tau_t^P(C)-\tau_t(C)}
        &\leq |t|\sup_{|u|\leq|t|}\norm{[P,\tau_u(C)]}\\
        &=|t|\sup_{|u|\leq|t|}\norm{[\tau_{-u}(P),C]}\\
        &\leq |t|\norm P\norm C\,\zeta_\lr(X,W,|t|).
    \end{align*}
    Applying this estimate with $C=A$ gives
    \begin{align*}
        \norm{[\tau_t^P(A),B]}
        &\leq\norm{[\tau_t(A),B]}
        +2\norm B\norm{\tau_t^P(A)-\tau_t(A)}\\
        &\leq\norm A\norm B\bigl(
        \zeta_\lr(Y,Z,|t|)
        +2|t|\norm P\zeta_\lr(X,Y,|t|)\bigr).
    \end{align*}
    On the other hand, applying this identity to $C=B$, gives
    \begin{align*}
        \norm{[\tau_t^P(A),B]}
        &=\norm{[A,\tau_{-t}^P(B)]}\\
        &\leq\norm{[A,\tau_{-t}(B)]}
        +2\norm A\norm{\tau_{-t}^P(B)-\tau_{-t}(B)}\\
        &\leq\norm A\norm B\bigl(
        \zeta_\lr(Y,Z,|t|)
        +2|t|\norm P\zeta_\lr(X,Z,|t|)\bigr).
    \end{align*}
    Taking the minimum over the two bounds gives the statement.
\end{proof}

\medskip
\begin{proof}[Proof of Proposition \ref{prop: loc of phi}]
    For $\beta=0$, the assertions about $\Phi_\beta$ follow from $\Phi_0=\mathrm{id}$. The bound on the dynamics is proved below and does not depend on $\beta$. In the calculations involving $f_\beta$, assume $\beta\neq0$.
    
    Let $k \in \N_0$ and $z \in \Gamma$. We denote by $\tau^{X, P}$ the dynamics generated by $\mL_\Psi + \mL_{V_X + P}$, and proceed to bound
    \begin{align*}
        \lVert (1 - \E_{B_k (z)} ) \, \Phi_\beta^{\tau^{X, P}} (A) \rVert & \leq
        \int_\R \d t \, f_\beta(t) \, \lVert (1 - \E_{B_k (z)}) \, \tau_{-t}^{X, P} (A) \rVert
        \\
        &\leq \sup_{|t| < T} \, \lVert (1 - \E_{B_k (z)}) \, \tau_{-t}^{X, P} (A) \rVert + 2 \, \lVert A \rVert \, \int_{|t| \geq T} \d t \, f_\beta (t) \, .
    \end{align*}
The explicit kernel \cite[Eq.~(26)]{CMTW25} gives
    \begin{align*}
        \int_{|t|\geq T}f_\beta(t)\,\d t
        &=\frac8{\pi^2}\sum_{j=0}^\infty
            \frac{\e^{-(2j+1)\pi T/|\beta|}}{(2j+1)^2}
        \leq\e^{-\pi T/|\beta|}.
    \end{align*}
    Hence
    \begin{align*} \label{eq: naive QBP locality}
        \int_\R \d t \, f_\beta(t) \, \lVert (1 - \E_{B_k (z)}) \, \tau_{-t}^{X, P} (A) \rVert \leq \sup_{|t| < T} \, \lVert (1 - \E_{B_k (z)}) \, \tau_{-t}^{X, P} (A) \rVert + 2 \, \e^{- \frac{\pi}{|\beta| } T } \, \lVert A \rVert  \, , \numberthis 
    \end{align*}
    for any $T > 0$. We now proceed to bound $\lVert (1 - \E_{B_k (z)}) \, \tau_{-t}^{X, P} (A) \rVert$.
    In \cite[Lemma 2.10]{BTW26} similar bounds on the decay norm of dynamics with $V = 0$ and $P = 0$ are shown. We follow the idea of the proof. 
    
    We first define finite-volume approximations of $\Psi$, $V$, and $P$. We set
    \begin{align*}
        \Psi_{k}(M) \coloneq  \Psi(M)   \, 
    \end{align*}
    for all $M \subseteq B_{k}(z)$  and $\Psi_{k}(M) = 0 $ otherwise.
    Similarly, for $V$, we define the finite volume approximation $V_{X, k}$ of $V_X = \sum_{Z \subseteq X} V (Z)$ by
    \begin{align*}
        V_{X, k} (M) &\coloneq  V(M) \, ,
    \end{align*}
    for all $M \subseteq B_{k} (z) \cap X$ and $V_{X, k} (M) = 0$ otherwise. Notice that these choices ensure that $\lVert \Psi_k \rVert_F \leq \lVert \Psi \rVert_F$ and $\lVert V_{X, k} \rVert_F \leq \lVert V \rVert_F$.
    For the bounded perturbation $P$, we define the approximation $P_{ k}$ of $P$ near $z$ by
    \begin{align*}
        P_{k} &\coloneq \E_{B_{k / 4} (z)} \, P \, .
    \end{align*}
    We denote the automorphism generated by $\mL_{\Psi_k + V_{X, k}} + \mL_{P_k}$ with $\tau^{X, P, k}$. Note that the automorphisms $\tau^{X, P, k}$ map $\mA_{B_k(z)}$ to itself, since all local terms of $\Psi_k$ and $V_{X, k}$ and $P_k$ lie in $\mA_{B_k(z)}$.
    
    Let $A \in \mA_F$ be arbitrary. We will begin to bound
    \begin{align*}
        \lVert (1-\E_{B_k(z)}) \, \tau_t^{X, P} (A) \rVert 
        \leq 2 \, \lVert (1 - \E_{B_{k / 4}(z)}) \, A \rVert + \lVert (1-\E_{B_k(z)}) \, \tau_t^{X, P} (\E_{B_{k / 4}(z)}\, A )\rVert   \, 
    \end{align*}
    by something that decays as $F(k)$ in $k$. This bound will be sufficient for short times. For long times we will use the trivial bound $\lVert (1-\E_{B_k(z)}) \, \tau_t^{X, P} (A) \rVert  \leq 2 \, \lVert A \rVert$ and the decay in $|t|$ of the integral. The first term is bounded by
    \begin{align*}
        2\norm{(1-\E_{B_{k/4}(z)})A}
        &\leq2F(\lfloor k/4\rfloor)\norm A_{F,z}
        \leq\frac{2}{F(1)}F(k/4)\norm A_{F,z}.
    \end{align*} 
    For the second term, we use the fundamental theorem of calculus and find
    \begin{align*}
        &\norm{(1-\E_{B_k(z)})\tau_t^{X,P}(\E_{B_{k/4}(z)}A)}\\
        &\quad=\norm{(1-\E_{B_k(z)})
                      (\tau_t^{X,P}-\tau_t^{X,P,k})\E_{B_{k/4}(z)}A}\\
        &\quad\leq2\norm{(\tau_t^{X,P}-\tau_t^{X,P,k})\E_{B_{k/4}(z)}A}\\
        &\quad\leq2\int_{\min(0,t)}^{\max(0,t)}
                  \norm{\partial_s(\tau_s^{X,P}\tau_{t-s}^{X,P,k})
                             \E_{B_{k/4}(z)}A}\,\d s\\
        &\quad=2\int_{\min(0,t)}^{\max(0,t)}
           \norm{(\mL_\Psi+\mL_{V_X+P}-\mL_{\Psi_k+V_{X,k}}-\mL_{P_k})
                    \tau_{t-s}^{X,P,k}(\E_{B_{k/4}(z)}A)}\,\d s.
    \end{align*}
    Note that $\tau_{t - s}^{X, P, k} (\E_{B_{{k / 4}}(z)} \, A)$ is local and therefore lies in the domain of $\mL_{\Psi} + \mL_{V_X + P}$. By Lemma \ref{lem: derivations} and the triangle inequality,
    \begin{align*}
        &\norm{(1-\E_{B_k(z)})\tau_t^{X,P}(\E_{B_{k/4}(z)}A)}\\
        &\quad\leq2\int_{\min(0,t)}^{\max(0,t)}\d s
          \sum_{\substack{Z\cap B_{k/2}(z)\neq\emptyset\\Z\cap B_k(z)^c\neq\emptyset}}
          \norm{[\Psi(Z)+\mathbf1_{Z\subseteq X}V(Z),
                       \tau_{t-s}^{X,P,k}(\E_{B_{k/4}(z)}A)]}\\
        &\qquad+2\int_{\min(0,t)}^{\max(0,t)}\d s
          \sum_{\substack{Z\cap B_{k/2}(z)=\emptyset\\Z\cap B_k(z)^c\neq\emptyset}}
          \norm{[\Psi(Z)+\mathbf1_{Z\subseteq X}V(Z),
                       \tau_{t-s}^{X,P,k}(\E_{B_{k/4}(z)}A)]}\\
        &\qquad+2\int_{\min(0,t)}^{\max(0,t)}
          \norm{[(1-\E_{B_{k/4}(z)})P,
                       \tau_{t-s}^{X,P,k}(\E_{B_{k/4}(z)}A)]}\,\d s.
    \end{align*}
    All sums run over $Z\in P_0(\Gamma)$. The last term is bounded by
    \begin{align*}
        \frac4{F(1)}|t|F(k/4)\norm P_{F,z}\norm A.
    \end{align*}
    Let $G$ be a decay function, chosen below, such that
    \begin{align*}
        \sup_{r\geq0}(1+r)^DG(r/2)<\infty,\qquad
        \sum_{r=0}^\infty(1+r)^{D-1}G(r/4)<\infty.
    \end{align*}
    Every term in the first sum has diameter at least $k/2$. This sum, including the time integral, is therefore bounded by
    \begin{align*}
        &4|t|\norm A\,|B_{k/2}(z)|F(k/2)G(k/2)
                   (\norm\Psi_{FG}+\norm V_{FG})\\
        &\quad\leq4C|t|F(k/2)(\norm\Psi_{FG}+\norm V_{FG})\norm A,
    \end{align*}
    where $C$ depends only on the lattice and $G$.

    For the second sum, choose $x\in Z\cap B_k(z)^c$ and insert the conditional expectation onto $B_{d(x,z)/4}(x)$ in each interaction term. The part outside this ball vanishes unless $\diam Z>d(x,z)/4$. Its contribution is bounded by
    \begin{align*}
        &8|t|\norm A\sum_{x\in B_k(z)^c}
          \sum_{\substack{Z\ni x\\\diam Z>d(x,z)/4}}
                \bigl(\norm{\Psi(Z)}+\mathbf1_{Z\subseteq X}\norm{V(Z)}\bigr)\\
        &\quad\leq8|t|F(k/4)(\norm\Psi_{FG}+\norm V_{FG})
                      \norm A\sum_{x\in B_k(z)^c}G(d(x,z)/4)\\
        &\quad\leq C^\prime|t|F(k/4)
                     (\norm\Psi_{FG}+\norm V_{FG})\norm A.
    \end{align*}
    Here $C^\prime$ depends only on the lattice and $G$. For the local part, invariance of the norm under the dynamics and Lemma \ref{lem: perturbed LRB} give
    \begin{align*}
        &2\int_{\min(0,t)}^{\max(0,t)}\d s
          \sum_{x\in B_k(z)^c}
          \sum_{\substack{Z\ni x\\Z\cap B_{k/2}(z)=\emptyset}}
          \norm{[\tau_{s-t}^{X,P,k}
             (\E_{B_{d(x,z)/4}(x)}(\Psi(Z)+\mathbf1_{Z\subseteq X}V(Z))),
                         \E_{B_{k/4}(z)}A]}\\
        &\quad\leq2|t|\norm A(1+2|t|\norm P)
          \sum_{x\in B_k(z)^c}\zeta_\lr(B_{d(x,z)/4}(x),B_{k/4}(z),|t|)\\
        &\qquad\qquad\times\sum_{Z\ni x}
                    \bigl(\norm{\Psi(Z)}+\mathbf1_{Z\subseteq X}\norm{V(Z)}\bigr)\\
        &\quad\leq2|t|\norm A(1+2|t|\norm P)
          (\norm\Psi_{FG}+\norm V_{FG})
          \sum_{x\in B_k(z)^c}\zeta_\lr(B_{d(x,z)/4}(x),B_{k/4}(z),|t|).
          \numberthis\label{eq: locality QBP LRB}
    \end{align*}
    It remains to bound $\zeta_\lr \bigl( B_{d(x, z) / 4}(x), \, B_{k / 4} (z), |t| \bigr)$. Here we use the fact that the distance of the supports is at least $3 d(x, z) / 4 - k /4$. For $d(x, z) > k$ this is at least $k / 4$ and at least $d(x, z) / 4$. We distinguish two cases: polynomially and exponentially decaying interactions.

    \textbf{Polynomial decay.} Choose $\eps>0$ such that $3D+\nu+\eps<\mu$, and set
    $F(r)=(1+r)^{-\nu}$ and $G(r)=(1+r)^{-3D-\eps}$. We choose $\sigma$ > 0 such that
    \begin{align*}
        \max  \bigl( \frac{D+1}{2D+\nu+\eps+1},
        \frac{D+\nu}{D+\nu+\eps/2}
        \bigr )<\sigma<1.
    \end{align*}
    We also set $\delta \coloneq \sigma(D+\eps/2)-D > 0$. Proposition \ref{prop: LiebRobinson} gives
    \begin{align*}
        &\zeta_\lr(B_{d(x,z)/4}(x),B_{k/4}(z),|t|)\\
        &\quad\leq C_\vol (1+k / 4)^D f(|t|)
        \Bigl(\exp\bigl(\min (0,v|t|-(k/8+d(x,z)/8)^{1-\sigma})\bigr)\\
        &\hspace{17em} +(1+k/8+d(x,z)/8)^{-\sigma(2D+\nu+\eps)}\Bigr).
    \end{align*}
    For the second summand we calculate
    \begin{align*}
        (1+k/8+d(x,z)/8)^{-\sigma(2D+\nu+\eps)} &\leq
        (1+k/8)^{-\sigma(D+\nu+\eps/2)}
        (1+d(x,z)/8)^{-\sigma(D+\eps/2)}\\
        &\leq
        (1+k/8)^{-D-\nu}(1+d(x,z)/8)^{-D-\delta}.
    \end{align*}
    We notice that for any $a$, $b \geq 0$ it holds that 
    \begin{align*}
        1 + \max (0, a - b) \geq \frac{1+ a}{ 1 + b} .
    \end{align*}
    We choose $\alpha>2\max (D+\nu,D+\delta)/(1-\sigma)$. Since
    $\exp(\min(0,b-a))=\exp(-\max(0,a-b))$, we calculate
    \[
    \begin{aligned}
        &\exp(\min (0, v|t|-(k/8+d(x,z)/8)^{1-\sigma} ) )\\
        &\quad\leq C_\alpha
        \bigl(1+\max(0,(k/8+d(x,z)/8)^{1-\sigma}-v|t|)\bigr)^{-\alpha}\\
        &\quad\leq C_\alpha(1+v|t|)^\alpha
        \bigl(1+(k/8+d(x,z)/8)^{1-\sigma}\bigr)^{-\alpha},
    \end{aligned}
    \]
    for some constant $C_\alpha > 0$. Moreover,
    \[
    \begin{aligned}
      &\bigl(1+(k/8+d(x,z)/8)^{1-\sigma}\bigr)^{-\alpha}\\
      &\quad\leq
      \bigl(1+(k/8)^{1-\sigma}\bigr)^{-\alpha/2}
      \bigl(1+(d(x,z)/8)^{1-\sigma}\bigr)^{-\alpha/2}\\
      &\quad\leq 8^{\alpha (1 - \sigma)}
      (1+k)^{-\alpha(1-\sigma)/2}
      (1+d(x,z))^{-\alpha(1-\sigma)/2}\\
      &\quad\leq 8^{\alpha (1 - \sigma)}
      (1+k)^{-D-\nu}(1+d(x,z))^{-D-\delta},
    \end{aligned}
    \]
    Hence, for some polynomial $\widetilde f$,
    \begin{align*}
        &\sum_{x\in B_k(z)^c}
        \zeta_\lr(B_{d(x,z)/4}(x),B_{k/4}(z),|t|) \leq C \widetilde f(|t|)(1+k)^{-\nu},
    \end{align*}
    where $C>0$ depends only on the surface-regularity constants and the fixed exponents $\alpha$, $\sigma$, $\nu$ and $\delta$. Combining all bounds, we get 
    \begin{gather*}
        \lVert  \tau_t^{X, P} (A) \rVert_{\nu, z} \leq f^\prime (|t|) \,  \, \lVert A \rVert_{\nu, z} \, \bigl( 1 + \lVert P \rVert_{ \nu, z} \bigr)   \, 
        \bigl( 1 + \lVert \Psi \rVert_{\nu + 3D + \eps} + \lVert V \rVert_{\nu + 3D + \eps} \bigr) \,  \, ,
    \end{gather*}
    where $f^\prime$ is some polynomial. This shows the first claim. To obtain the stated temperature dependence, use
    $f_\beta(t)=|\beta|^{-1}f_1(t/|\beta|)$. If
    $f^\prime(r)=\sum_{j=0}^N c_jr^j$, then
    \[
    \begin{aligned}
      \int_\R f_\beta(t)f^\prime(|t|)\,\d t
      &=\sum_{j=0}^Nc_j|\beta|^j
        \int_\R f_1(u)|u|^j\,\d u.
    \end{aligned}
    \]
    All moments on the right-hand side are finite because $f_1$ decays exponentially. Hence this expression is a polynomial in $|\beta|$, which proves the second claim in the long-range regime.

    \textbf{Exponential decay.} We choose $F(r) = \exp(-\min (b, c/2)r)$ and $G(r) = \exp(-cr/2)$ and define $I \coloneq \lVert \Psi \rVert_{\exp(- c \cdot)} + \lVert V \rVert_{\exp(-c \cdot)}$. We use the Lieb--Robinson bound from Proposition \ref{prop: LiebRobinson} with $\mu=c$ and $\lambda=c/2$. We obtain the bound
    \begin{align*}
        &\zeta_\lr\bigl(B_{d(x,z)/4}(x),B_{k/4}(z),|t|\bigr) \leq C(1+k)^D \exp\left(\kappa_{c,c/2}I|t|-\frac{ck}{8} -\frac{c\,d(x,z)}8\right),
    \end{align*}
    for some constants $C$, $\kappa_{c,c/2} > 0$ that are independent of $x$, $z$, $k$ and $t$. Therefore, the expression in Equation \eqref{eq: locality QBP LRB} is bounded by
    \begin{align*}
        C'I|t|\norm A(1+2|t|\norm P)(1+k)^D
          \exp\left(\kappa_{c,c/2}I|t|-\frac{ck}{8}\right),
    \end{align*}
    for a constant $C^\prime > 0$ that is independent of $x$, $z$, $k$ and $t$. Combining this with the preceding bounds gives
    \begin{align*}
        &\norm{(1-\E_{B_k(z)})\tau_t^{X,P}(A)}\\
        &\quad\leq
        C^{\prime \prime}\norm A_{\exp(-b\cdot),z}
        (1+\norm P_{\exp(-b\cdot),z})\\
        &\qquad\times\left((1+|t|)\exp(-\min (b,c/2 )k/4)
        +I|t|(1+|t|)(1+k)^D
          \exp\left(\kappa_{c,c/2}I|t|-\frac{ck}{8}\right)\right),
    \end{align*}
    for a constant $C^{\prime \prime } > 0$ that is independent of $x$, $z$, $k$ and $t$. We choose
    \begin{align*}
        C_\inter(a,c)\coloneq\frac{\pi}{2\kappa_{c,c/2}}
            \left(\frac c{8a}-1\right),
        \qquad
        T\coloneq\frac{c|\beta|k}{8(\pi+\kappa_{c,c/2}I|\beta|)}.
    \end{align*}
    Since $c>8a$, we have $C_\inter(a,c)>0$. For $k\geq1$, the two time-dependent exponents satisfy
    \begin{align*}
        \kappa_{c,c/2}IT-\frac{ck}{8}
        &=-\frac{\pi T}{|\beta|}
        =-\frac{c\pi k}{8(\pi+\kappa_{c,c/2}I|\beta|)}
        \leq-\frac{2ac}{c+8a}k,
    \end{align*}
    where we used $I|\beta|<C_\inter(a,c)$. Moreover,
    \begin{align*}
        T&\leq\frac{c|\beta|k}{8\pi},\\
        1+T&\leq\left(1+\frac c{8\pi}\right)(1+|\beta|)(1+k),\\
        IT(1+T)&\leq\frac{cC_\inter(a,c)}{8\pi}
                 \left(1+\frac c{8\pi}\right)(1+|\beta|)(1+k)^2.
    \end{align*}
    Using these bounds in Equation \eqref{eq: naive QBP locality} gives
    \begin{align*}
        &\int_\R f_\beta(t)
          \norm{(1-\E_{B_k(z)})\tau_{-t}^{X,P}(A)}\,\d t\\
        &\quad\leq C^{\prime\prime\prime}(1+|\beta|)
             \norm A_{\exp(-b\cdot),z}(1+\norm P_{\exp(-b\cdot),z})\\
        &\qquad\times\left((1+k)\e^{-\min(b,c/2)k/4}
                     +(1+k)^{D+2}\e^{-2ac k/(c+8a)}\right).
    \end{align*}
    The constant $C^{\prime\prime\prime}$ is independent of $x$, $z$, $k$ and $\beta$. Since $\min(b,c/2)/4>a$ and $2ac/(c+8a)>a$, multiplication by $\e^{ak}$ leaves a quantity bounded uniformly in $k$. The term $k=0$ is bounded by $2\norm A$. Adding the operator norm proves the exponential statement.
\end{proof}

\paragraph{Acknowledgements.}
The author is grateful to Stefan Teufel and Tom Wessel for their support and many helpful discussions, and thanks Ángela Capel, Tim Möbus and Marius Wesle for helpful exchanges. This work was supported by the Deutsche Forschungsgemeinschaft (DFG, German Research Foundation) – 470903074.

\paragraph{AI-use statement.}
This manuscript, including all mathematical proofs, was written entirely by the author. OpenAI's ChatGPT Pro (GPT-5.6 Sol) was used solely as a review aid and for language refinement. The author takes full responsibility for its content.

\printbibliography

\end{document}